\documentclass[acmsmall,screen]{acmart}
\setcopyright{rightsretained}
\acmPrice{}
\acmDOI{10.1145/3434295}
\acmYear{2021}
\copyrightyear{2021}
\acmSubmissionID{popl21main-p96-p}
\acmJournal{PACMPL}
\acmVolume{5}
\acmNumber{POPL}
\acmArticle{14}
\acmMonth{1}

\usepackage{booktabs}   
\usepackage{subcaption} 
\usepackage{xspace}
\usepackage{xcolor}
\usepackage{array}
\usepackage{paralist}
\usepackage{bm}
\usepackage{color}
\usepackage{mathpartir}

\usepackage{amsmath, amsthm}
\usepackage{cellspace, paralist}
\usepackage{bussproofs}  \EnableBpAbbreviations

\usepackage{url}
\usepackage{balance}

\setdefaultleftmargin{10pt}{10pt}{}{}{}{}

\usepackage[linesnumbered,vlined,ruled]{algorithm2e}
\SetAlFnt{\fontsize{9.5}{11}\selectfont} 

\begin{document}

\title[Generating collection transformations]{Generating  collection transformations from proofs}         


\author{Michael Benedikt}
\affiliation{
  \department{Computer science department}              
  \institution{Oxford University}            
  \country{United Kingdom}                    
}

\author{Pierre Pradic}
\affiliation{
  \department{Computer science department}              
  \institution{Oxford University}           
  \country{United Kingdom}                   
}

\newcommand{\trans}{{\mathcal T}}
\newcommand{\coq}{\kw{COQ}}
\newcommand{\freevars}{\kw{FV}}
\newcommand{\funall}{\kw{Fun}_{\kw{All}}}
\newcommand{\funfin}{\kw{Fun}_{\kw{Fin}}}
\newcommand{\ljboundedvarq}{\kw{LJBoundVarQ}}
\newcommand{\ljboundedq}{\kw{LJBoundedQ}}
\newcommand{\ljboundedconn}{\kw{LJBoundedConn}}
\newcommand{\ljdeltazero}{LJ\deltazero}
\newcommand{\rnandL}{$\wedge$-\textsc{L}}
\newcommand{\rnandR}{$\wedge$-\textsc{R}}
\newcommand{\rnandRB}{$\wedge$-\textsc{RB}}
\newcommand{\rnorL}{$\vee$-\textsc{L}}
\newcommand{\rnorR}{$\vee$-\textsc{R}}
\newcommand{\rnimplL}{$\Rightarrow$-\textsc{L}}
\newcommand{\rnimplLB}{$\Rightarrow$-\textsc{LB}}
\newcommand{\rnimplR}{$\Rightarrow$-\textsc{R}}
\newcommand{\rnallL}{$\forall$-\textsc{L}}
\newcommand{\rnallR}{$\forall$-\textsc{R}}
\newcommand{\rnexL}{$\exists$-\textsc{L}}
\newcommand{\rnexR}{$\exists$-\textsc{R}}
\newcommand{\rnallLBV}{$\forall$-\textsc{LBV}}
\newcommand{\rnallRBV}{$\forall$-\textsc{RBV}}
\newcommand{\rnexLBV}{$\exists$-\textsc{LBV}}
\newcommand{\rnexRBV}{$\exists$-\textsc{RBV}}
\newcommand{\rnallLB}{$\forall$-\textsc{LB}}
\newcommand{\rnallRB}{$\forall$-\textsc{RB}}
\newcommand{\rnexLB}{$\exists$-\textsc{LB}}
\newcommand{\rnexRB}{$\exists$-\textsc{RB}}
\newcommand{\rnax}{\textsc{AX}}
\newcommand{\rntop}{$\top$-\textsc{R}}
\newcommand{\rnbot}{$\bot$-\textsc{L}}
\newcommand{\rneqconst}{$=$-\textsc{CL}}
\newcommand{\rneqL}{$=$-\textsc{L}}
\newcommand{\rneqR}{$=$-\textsc{R}}
\newcommand{\rnpaireta}{$\times_\eta$}
\newcommand{\rnpairbeta}{$\times_\beta$}
\newcommand{\rnuniteta}{$\unit_\eta$}

\newcommand{\msu}{,}
\newcommand{\nativemember}{\in}
\newcommand{\derivedmember}{\in'}
\newcommand{\macroin}{\mathrel{\tilde\in}}
\newcommand{\macrosubseteq}{\mathrel{\widetilde\subseteq}}
\newcommand{\macroeq}{\mathrel{{\widetilde=}}}

\newcommand{\groupq}{\kw{Group}}
\newcommand{\projq}{\kw{Proj}}
\newcommand{\filterq}{\kw{Filter}}
\newcommand{\nrcget}{\kw{Get}}
\newcommand{\nrcwget}{\nrc\kw{[}\nrcget\kw{]}}
\newcommand{\rqfo}{\kw{RQFO}}
\newcommand{\unit}{\kw{Unit}}
\newcommand{\nra}{\kw{NRA}}
\newcommand{\nrc}{\kw{NRC}}
\newcommand{\inschema}{\aschema_{in}}
\newcommand{\outschema}{\aschema_{out}}

\newcommand{\convert}{\kw{Convert}}
\newcommand{\cast}{\kw{Cast}}
\newcommand{\obj}{\kw{o}}
\newcommand{\oneobjth}{\kw{O}}
\newcommand{\oneobj}{\kw{obj}}
\newcommand{\inobj}{\kw{o}_{in}}
\newcommand{\outobj}{\kw{o}_{out}}
\newcommand{\map}{\kw{Map}}
\newcommand{\flatten}{\kw{Flatten}}
\newcommand{\pairwith}{\kw{PairWith}}
\newcommand{\booltype}{\kw{Bool}}
\newcommand{\booltt}{\kw{tt}}
\newcommand{\boolff}{\kw{ff}}
\newcommand{\is}{\kw{is}}
\newcommand{\instances}{\kw{Instances}}
\newcommand{\inst}{\kw{Inst}}
\newcommand{\sig}{{\mathcal SIG}}
\newcommand{\sorts}{\kw{Sorts}}
\newcommand{\bool}{\booltype}
\newcommand{\depth}{\kw{Depth}}
\newcommand{\false}{\kw{False}}
\newcommand{\smallsort}{\sort_0}
\newcommand{\bigsort}{\sort_1}
\newcommand{\smallsorts}{\sorts_0}
\newcommand{\bigsorts}{\sorts_1}
\newcommand{\sort}{\kw{S}}
\newcommand{\interp}{{\mathcal I}}
\newcommand{\ur}{{\mathcal U}}
\newcommand{\inttype}{{\bbN}}
\newcommand{\ursort}{\ur}
\newcommand{\true}{\kw{True}}
\newcommand{\verify}{\kw{Verify}}
\newcommand{\deltazero}{\Delta_0}
\newcommand{\bqf}{\Delta_0}
\newcommand{\depjoin}{\join^{dep}}
\newcommand{\tuples}{\kw{Tuples}}
\newcommand{\tuple}{\kw{Tuple}}
\newcommand{\tupleof}{\tuple}
\newcommand{\collapse}{\kw{Collapse}}
\newcommand{\vals}{\kw{Vals}}
\newcommand{\exptime}{\kw{EXPTIME}}
\newcommand{\pspace}{\kw{PSPACE}}
\newcommand{\ptime}{\kw{PTIME}}
\newcommand{\aseq}{\kw{Seq}}
\newcommand{\efi}[1]{\textbf{Efi}: \textcolor{red}{#1}}
\newcommand{\balder}[1]{\textcolor{red}{B: #1}}
\newcommand{\proves}{\vdash}
\newcommand{\join}{\bowtie}
\newcommand{\accpart}{\kw{AccPart}}
\newcommand{\accs}{\kw{AcSch}}
\newcommand{\accsb}{\kw{AcSch}^\leftrightarrow}
\newcommand{\accsneg}{\kw{AcSch}^\neg}
\newcommand{\command}{\kw{Command}}
\newcommand{\abag}{\kw{B}}
\newcommand{\comms}{\kw{Comms}}
\newcommand{\accout}{\kw{AccOut}}

\newcommand{\accax}{\kw{AccAx}}
\newcommand{\forax}{\kw{ForAx}}
\newcommand{\backax}{\kw{BackAx}}
\newcommand{\aschema}{{\mathcal SCH}}
\newcommand{\domainof}{\kw{Domain}}

\newcommand{\rels}{{\mathcal R}}
\newcommand{\meths}{\kw{Meths}}
\newcommand{\methoddef}{\kw{MethodDefs}}
\newcommand{\projax}{\kw{ProjAx}}

\newcommand{\globdom}{\geq_{GD}}
\newcommand{\uspjneg}{\kw{USPJ}^{\neg}}
\newcommand{\kw}[1]{{\mathsf{#1}}\xspace}
\newcommand{\profinfo}{ \kw{Profinfo}}
\newcommand{\studinfo}{ \kw{Studentinfo}}
\newcommand{\universitydirectory}{\kw{Udirectory}}
\newcommand{\udirectory}{\kw{Udirect}}
\newcommand{\eid}{\kw{eid}}
\newcommand{\lname}{\kw{lname}}
\newcommand{\onum}{\kw{onum}}
\newcommand{\accessible}{\kw{accessible}}
\newcommand{\access}{\kw{Access}}
\newcommand{\acc}[1]{\kw{Accessed} #1}
\newcommand{\infacc}[1]{\kw{InferredAcc} #1}
\newcommand{\bestplan}{\kw{BestPlan}}
\newcommand{\bestcost}{\kw{BestCost}}
\newcommand{\oldbestcost}{\kw{OldBestCost}}
\newcommand{\candidates}{\kw{Candidates}}
\newcommand{\plan}{\kw{Plan}}
\newcommand{\plantree}{\kw{PlanTree}}
\newcommand{\config}{\kw{config}}
\newcommand{\mt}{\kw{mt}}
\newcommand{\aplan}{\kw{PL}}
\newcommand{\cost}{\kw{Cost}}
\newcommand{\pointer}{\kw{Pointer}}
\newcommand{\parent}{\kw{Parent}}
\newcommand{\children}{\kw{Children}}
\newcommand{\oldcost}{\kw{OldCost}}
\newcommand{\newcost}{\kw{NewCost}}
\newcommand{\atomiccost}{\kw{AtomicCost}}
\newcommand{\haspointer}{\kw{HasPointer}}
\newcommand{\isroot}{\kw{IsRoot}}
\newcommand{\prune}{\kw{Prune}}
\newcommand{\goodplans}{\kw{GoodPlans}}
\newcommand{\oldplans}{\kw{OldPlans}}

\newcommand{\backup}{\kw{BackUpBestCost}}

\newcommand{\uplan}{\kw{ChasePlan}}

\newcommand{\ids}{\kw{Ids}}
\newcommand{\names}{\kw{Names}}

\newcommand{\B}{\mathcal{B}}
\newcommand{\correct}[1]{\textcolor{red}{\textbf{Check:} \textbf{#1}}}
\renewcommand{\phi}{\varphi}

\newtheorem{op}{Open question}
 \newtheorem{algo}{Algorithm}
\newtheorem{question}{Question}

\newcommand{\myparagraph}[1]{{\bf #1.}}
\newcommand{\myeat}[1]{}

\newenvironment{myexmp}{\refstepcounter{theorem}\par\medskip
\noindent\textsc{{Example~\thetheorem.}}}{\null\hfill$\triangleleft$\medskip}

\newcommand{\structt}{\kw{Struct}}
\newcommand{\sett}{\kw{Set}}
\newcommand{\issing}{\kw{IsSing}}
\newcommand{\sing}{\kw{Sing}}
\newcommand{\istwo}{\kw{IsTwo}}
\newcommand{\case}{\kw{case}}

\newcommand{\totalit}{\kw{TotalIt}}
\newcommand{\firstmatch}{\kw{ItFirstMatch}}
\newcommand{\bestmatch}{\kw{ItBestPlan}}
\newcommand{\totalfact}{\kw{TotalFacts}}
\newcommand{\dominance}{\kw{Dominance}}
\newcommand{\globalequiv}{\kw{GlobalEquiv}}
\newcommand{\costprune}{\kw{CostPrune}}
\newcommand{\totalfacts}{\kw{Facts}}
\newcommand{\image}{\kw{Im}}

\newcommand{\minput}{\kw{input}}
\newcommand{\moutput}{\kw{output}}

\newcommand{\<}{\langle}
\renewcommand{\>}{\rangle}

\newcommand{\bbN}{\mathbb{N}}                
\newcommand{\eqdef}{\mathrel{:=}}
\newcommand{\isfst}{\kw{IsFst}}
\newcommand{\issnd}{\kw{IsSnd}}
\newcommand{\allpairs}{\kw{AllPairs}}
\newcommand{\pair}{\kw{Pair}}
\newcommand{\bnfeq}{\mathrel{::=}}
\newcommand{\bnfalt}{\mid}
\newcommand{\subtype}{\leq}
\newcommand{\subobject}{\leq}
\newcommand{\interpsort}{\tau}
\newcommand{\isnonuniformdefiner}{\delta}


\begin{abstract}
Nested relations,  built up from atomic types via product and set types,
form a rich data model.
Over the last decades the \emph{nested relational calculus}, $\nrc$, has emerged as a standard language
for defining transformations on nested collections. $\nrc$ is a strongly-typed functional language which
allows building up transformations using tupling and projections, a singleton-former,
and a map operation that lifts transformations on tuples to transformations on sets. 

In this work we describe an alternative declarative method of describing transformations
in logic. 
A formula with distinguished inputs
and outputs gives an \emph{implicit definition} if one can prove
that for each input there is only one output that satisfies it.
Our main result shows that one can synthesize transformations
from proofs  that a formula provides an implicit definition, where
the  proof is in an intuitionistic calculus
that captures a natural style of reasoning about nested
collections.  Our polynomial time synthesis procedure is based on
an analog of Craig's  interpolation lemma, starting with
a provable containment between terms representing nested collections
and generating an $\nrc$ expression that interpolates between them.

We further show that $\nrc$ expressions that implement an implicit definition
can  be found when there is a classical proof of functionality, not just when
there is an intuitionistic one.
That is, whenever a formula implicitly
defines a transformation, there is an $\nrc$ expression that implements it.
\end{abstract}



\maketitle

\section{Introduction}

\emph{Nested relations} are a natural data model for hierarchical data.
Nested relations are objects within a type system built up from basic types via tupling
and a set-former.
In the 1980's  and 90's, a number of algebraic languages were proposed for defining
transformations on 
nested collections. 
Eventually a standard
language emerged, the
 \emph{nested relational calculus} ($\nrc$). 
The language is strongly-typed and  functional, with transformations 
built up via tuple manipulation
operations  as well as operators for lifting 
transformations over a type $T$ to transformations taking 
as input a set of objects of type $T$, such as singletons
constructors and a mapping operator.
One common formulation of these uses variables and
a ``comprehension'' operator for forming new objects from old ones
\cite{natj}, while  an alternative
algebraic formalism
 presents the language as a set of operators that can be freely composed. It was
shown that each $\nrc$ expression can be evaluated in  polynomial time in the size
of a finite data input, and that 
when the input and output is ``flat'' (i.e. only one level of nesting), $\nrc$ expresses exactly the 
transformations in the standard relational database  language relational algebra.
Wong's thesis \cite{limsoonthesis} summarizes the argument made by this line of work
 ``$\nrc$ can be profitably regarded
as the `right' core for nested relational languages''. 
$\nrc$ has been the basis
for most work on transforming nested relations. It is the basis for a number of 
commercial   tools \cite{dremel}, including those embedding nested data transformations  in
programming languages \cite{linq}, in addition to having influence in the effective implementation of data transformations in 
functional programming languages   \cite{gibbonshenglein,ringads}.

Although $\nrc$ can be applied to other collection types, such as bags and lists,
we will focus here on just nested sets.  We will  show a new connection between
$\nrc$ and first-order logic.
There is a natural logic for describing properties of nested relations, the well-known
$\deltazero$ formulas, built up from equalities using
quantifications $\exists x \in \tau$ and $\forall y \in \tau$ where $\tau$ is a term.
For example, formula $\forall x \in c ~ \pi_1(x) \in \pi_2(x)$ might describe a property
of a nested relation $c$ that
is a set of pairs, where the first component of a pair is of
some type $T$  and the second component is a set containing elements
of type $T$.
A $\deltazero$ formula $\Sigma(\inobj^1 \ldots \inobj^k, \outobj)$ over 
variables $\inobj^1 \ldots \inobj^k$ and variable $\outobj$  thus defines
a relationship between 
$\inobj^1 \ldots \inobj^k$ and $\outobj$. For such a  formula to define a transformation it must
be \emph{functional}: it must enforce that $\outobj$ is determined by the values of 
$\inobj^1 \ldots \inobj^k$.
More generally, if we have a formula $\Sigma(\inobj^1 \ldots \inobj^k,  \outobj, \vec a)$, we say that 
$\Sigma$
\emph{implicitly defines $\outobj$ as a function of $\inobj^1 \ldots \inobj^k$} if:

\medskip

(*) For each two bindings $\sigma_1$ and $\sigma_2$
of the variables $\inobj^1 \ldots \inobj^k, \vec a, \outobj$ 
to nested relations satisfying $\Sigma$, if $\sigma_1$ and $\sigma_2$
agree on each  $\inobj^i$, then  they agree on $\outobj$.

\medskip

That is, $\Sigma$ entails that the value of $\outobj$ is a partial function of  the value
of $\inobj^1 \ldots \inobj^k$. 
\myeat{
Such a formula $\Sigma$ is \emph{total for $\vec i$} if
for each valuation for $\vec i$, there is  an extension
to a valuation for $\vec i, \vec a, o$ such that
$\Sigma(\vec i, o, \vec a)$. A formula that is total for $\vec i$ and defines $o$ as a function
of $\vec i$  specifies a transformation
from $\vec i$ to $o$ in the obvious way.
}

Note that when we say ``for each binding of variables to  nested relations''  in
the definitions above, we include infinite nested relations as well
as finite ones. 
An alternative characterization of $\Sigma$ being an implicit definition, which will be more relevant 
to us in the sequel,  is that there is a proof that $\Sigma$ defines  a functional relationship.
Note that (*) is a first-order entailment: 
$\Sigma(\inobj^1 \ldots \inobj^k, \outobj, \vec a) \wedge \Sigma(\inobj^1 \ldots \inobj^k, \outobj', \vec a')
\models \outobj=\outobj'$
where in the entailment we omit some first-order ``sanity axioms''
about tuples and sets.
We refer to a proof of (*) for a given $\Sigma$ and subset of
the input variables $\inobj^1 \ldots \inobj^k$, as a \emph{proof that $\Sigma$ implicitly defines
 $\outobj$ as a function
of  $\inobj^1 \ldots \inobj^k$}, or simply a \emph{proof of functionality}
dropping $\Sigma$, $\outobj$, and $\inobj^1 \ldots \inobj^k$ when they are clear
from context.
By the completeness theorem of first-order logic,  whenever $\Sigma$
defines $\outobj$ as a function of
$\inobj^1 \ldots \inobj^k$
according to the semantic definition above, this is witnessed by a
proof, in any of the standard complete proof calculi for classical first-order logic
(e.g. tableaux, resolution).
Such a proof will use  the sanity axioms referred to above,
which capture extensionality of sets, the compatibility
of the membership relation with the type hierarchy, and properties of projections and tupling.

\begin{myexmp} \label{ex:invertible}
We consider a specification in logic involving two nested collections, $F$ and $G$.
 The collection $F$ is  of type $\sett(\ur \times \ur)$, where $\ur$
refers to the basic set of elements, the ``Ur-elements'' in the sequel.
That is, $F$ is a set of pairs. The collection $G$ is of
of type $\sett(\ur \times \sett(\ur))$, a set whose
members are pairs, the first component an element
and the second a set.

Our specification $\Sigma$ will state that for each element $g$ in $G$
there is an element $f_1$ appearing as the first component of a pair in $F$, such that
$g$ represents $f_1$, in the sense that its first component is $f_1$ and
its second component accumulates all elements paired with $f_1$ in $F$.
 This can be specified easily by a $\deltazero$ formula:
\begin{align*}
\forall g \in G ~\exists f \in F \quad&  \pi_1(g)=\pi_1(f)~
\wedge~ \forall x \in \pi_2(g) ~  \<\pi_1(f), x \>  \in F \\
\wedge~~&\forall f' \in F ~~ \left[\pi_1(f')=\pi_1(f) \rightarrow \pi_2(f')  \in \pi_2(g)\right]
\end{align*}

$\Sigma$ also states that for each element $f_1$ lying  within 
a pair in $F$  there is a corresponding element $g$ of $G$ that pairs $f_1$
with all of the elements  linked  with $f$ in $F$. 
\begin{align*}
\forall f  \in F ~ \exists g \in G \quad& \pi_1(g)=\pi_1(f) ~\wedge~ \forall x \in \pi_2(g) ~ \<\pi_1(f), x \> \in F\\
\wedge~~& \forall f' \in F ~~
\left[\pi_1(f')=\pi_1(f) \rightarrow ~ \pi_2(f') \in \pi_2(g)\right]
\end{align*}

We can prove from $\Sigma$ that $G$ is a function of $F$, and thus $\Sigma$
implicitly defines a transformation from $F$ to $G$. We give the argument
informally here. Fixing $F,G$ and $F,G'$ satisfying $\Sigma$, we will
prove that if $g \in G$ then $g \in G'$. The proof begins by using the conjunct in
the first item to obtain an $f \in F$. We can then use the  
second  item on $G'$ to obtain 
a $g' \in G'$. 
We now need to prove that $g' = g$. Since $g$ and $g'$ are pairs, it suffices to show that their two projections are the same. We can easily see that $\pi_1(g)=\pi_1(f)=\pi_1(g')$, so it suffices
to prove $\pi_2(g')=\pi_2(g)$. Here we will make use of extensionality, arguing for 
containments between $\pi_1(g')$ and $\pi_2(g)$ in both directions.
In one direction we consider an $x \in \pi_2(g')$, and we need  to show
$x$ is in $\pi_2(g)$. By the second conjunct in the second item we have $\<\pi_1(f), x\> \in F$.
Now using the first item we can argue that $x \in \pi_2(g)$. 
In the other direction we consider $x \in \pi_2(g)$, we can apply the first
item to claim $\<\pi_1(f), x\> \in F$ and then employ the second item
to derive $x \in \pi_2(g')$.

Now let us consider $G$ as the input and $F$ as the output.
We cannot say that $\Sigma$ describes $F$ as a \emph{total} function of $G$, since 
$\Sigma$  enforces constraints on $G$:
that the second component of a pair in $G$ cannot be empty, and
that any two pairs in $G$ that agree on the first component must agree on the second.
But we can prove from $\Sigma$ that $F$ is a partial function of $G$: fixing
$F,G$ and $F',G$ satisfying $\Sigma$, we can prove that $F=F'$. 
\end{myexmp}



Our first main contribution is a polynomial time synthesis procedure    that takes
as input  a proof that $\Sigma$ implicitly defines $o$  as a function
of  $\inobj^1 \ldots \inobj^k$, generating  an $\nrc$ expression with input 
$\inobj^1 \ldots \inobj^k$ that implements the transformation that $\Sigma$ defines.
We require a  proof of functionality in a certain
intuitionistic calculus. Although the calculus is not complete for classical
entailment, we argue that it is quite rich and show that it is equivalent  to 
certain prior intuitionistic calculi.

\begin{myexmp} \label{ex:invertibleexp}
Let us return to Example \ref{ex:invertible}.
From a proof in our calculus that $\Sigma$ defines $G$ as
a function of $F$, our synthesis algorithm will produce  an
expression in $\nrc$ that generates $G$ from $F$. This will
be an expression that simply  ``groups on the first component''.

From a proof from $\Sigma$ that $F$ is a function of $G$,
our algorithm
will generate an $\nrc$ expression that forms $F$ by flattening $G$.
\end{myexmp}



We also show that this phenomenon applies 
when there is a  classical proof of functionality, not just
an intuitionistic one. That is, we show that whenever a formula
$\Sigma$ projectively implicitly defines a transformation $\trans$, that 
transformation can be expressed in a slight variant of $\nrc$. 
The result can be seen as an analog of the well-known Beth definability
theorem for first-order logic \cite{beth},  stating that a property of a  first-order structure
is defined by a first-order open formula exactly when it is implicitly defined by a first-order sentence. In the process we prove an \emph{interpolation theorem}, showing
that whenever we have provable containments between nested relations,
there is an $\nrc$ expression that sits between them.
Overall our results show a close connection between logical specifications
of transformations on nested collections and the functional
transformation language $\nrc$, a result which is not anticipated
by the prior theory.


\myparagraph{Organization}
We overview related work in Section \ref{sec:related} and
provide preliminaries in Section \ref{sec:prelims}. 
Section \ref{sec:effectivebeth} details our proof
calculus and the algorithm  that synthesizes definitions
from proofs. 
We include an example
(Figure~\ref{fig:distinguishers-prooftree}) of how one would use it
to prove functionality of an expression, and
an illustration
of how  our synthesis algorithm would generate an $\nrc$ expression
from the proof (Example \ref{ex:synth}).
Section \ref{sec:fointerp} concerns  another logic-based specification
that can be transformed into $\nrc$ expressions, based on the notion of interpretations.
Section \ref{sec:bethmodeltheoretic} shows
that even for classical proofs there is a corresponding $\nrc$ expression. This conversion
goes through the interpretation representation introduced in 
Section \ref{sec:bethmodeltheoretic}. We show a general result
that implicit definitions in multi-sorted logic can be converted to
interpretations, and then use the results of Section \ref{sec:bethmodeltheoretic}
to argue that these interpretations can be converted to $\nrc$ expressions.

We close with conclusions in Section \ref{sec:conc}.
In the body of the paper we focus on explaining the results and some
proof ideas, with most proof details
deferred to the supplementary materials.


\section{Related work} \label{sec:related}
\myeat{
One connection between logic and data transformation languages is well-known
in the  context of transformations of ordinary ``flat'' relations.
Codd's theorem  \cite{coddcomplete}, building on earlier
work in algebraization of logic,  identifies transformations
in the database language relational algebra with those in first-order logic that are safe, in that
they depend only on the interpretation of the relation symbols.
}
In the context of transformations of ordinary ``flat'' relations,
Segoufin and Vianu  \cite{SVconf} 
showed that transformations definable in relational algebra 
are the same as those that satisfy a variant of implicit definability (``determinacy'').
The result of \cite{SVconf} makes use of a refinement of Craig's interpolation theorem due to Otto \cite{otto}. The use of interpolation theorems in moving from implicit to explicit is well-established,
dating back to  Craig's proof of the Beth definability theorem \cite{craig57beth}.
Segoufin and Vianu's result is motivated
by the ability to evaluate transformations  defined over one set of ``base predicates'' using
another set of ``view predicates'', where the views are defined implicitly by a background theory
relating them  to the base predicate.
The idea that one can use interpolation algorithms
to synthesize
transformations from implicit specifications first appears in the work of Toman and Weddell \cite{tomanweddell}
and has been developed in a number of directions subsequently \cite{interpbook}.
In the absence of nesting of sets, the relationship between formulas
and terms of an algebra is much more straightforward; relational algebra defines
exactly those transformations whose output is a comprehension by a first-order
formula over the elements that are in the projection of some relation.
In the presence of nesting the relationship of algebra and logic
is more complex, and so in this work we will need to develop some different techniques
(e.g.  a new kind of interpolation result) to analyze the relationship
between logical and algebraic
definability.

The development of the nested relational model, culminating
in the convergence on the language $\nrc$, has a long
history. 
The thesis of Wong \cite{limsoonthesis} and the related paper of Buneman et al. \cite{natj} 
gave an elegant presentation of $\nrc$, and summarize the equivalences known
between a number of variations on the syntax. 
Connections with logic are implicit in results stating that
$\nrc$ queries can be ``simulated'' by flat queries: see
\cite{conservativity,simulation}. Further discussion on these
simulations can be found in Section \ref{sec:fointerp}.

More powerful languages than $\nrc$ were also considered,
including an extension with an operator
for forming the powerset of a set. This extension can
be captured using the natural logic with membership
\cite{abiteboulbeeri}.  The increased expressiveness implies  correspondingly
higher complexity (e.g. non-elementary in combined complexity), and perhaps
for this reason the subsequent development has focused on $\nrc$.
Much of the  development of $\nrc$ in the last decades has 
focused primarily on integration with functional languages 
\cite{gibbonshenglein,ringads,linq}, rather than synthesis or expressiveness.

Quite independently of work on logics
for nested relations in  computer science, researchers in other
areas have investigated the relationships
between various restricted algebras for manipulating sets.
Gandy \cite{gandy}  defines a class of \emph{Basic functions},
and compares them to functions definable by  $\deltazero$ formulas.
Later languages build on Gandy's work, particularly
for  a finer-grained  analysis of the  constructible sets \cite{jensen}.
An important distinction from the setting of $\nrc$ is that
these works do not restrict to sets built up from finitely many levels
of nesting above the Ur-elements. For instance, 
Gandy showed that there are Basic functions 
checking whether an input is an ordinal, or is the
ordinal $\omega$; in fact, he showed that there are Basic functions
that are not primitive recursive. 
In  the setting of \cite{gandy}, the $\deltazero$ functions are strictly more 
expressive than the Basic functions.

Model theorists have looked at generalizing the Beth definability theorem 
that relates implicit and explicit definability to
the case where the ``implicitly definable structure''  has new elements,
not just new relations.
Hodges and his collaborators \cite{hodgesbook, hodgesdugald} explore this in some restricted cases. 
Our approach in Section \ref{sec:bethmodeltheoretic}
 to showing a relationship between implicitly definable
transformations and interpretations is inspired
by the  unpublished  draft \cite{madarasz}, motivated
 from the perspective of algebraic
logic, which provides model-theoretic tools for connecting semantic and syntactic notions
of definability in multi-sorted logic.
\myeat{
\cite{hodgesbook} defines
the notion of \emph{rigidly relatively categorical} which is the single-sorted
analog of our implicitly interpretable. \cite{hodgesbook} does not prove
any  connection to explicit interpretability, although he proves the equivalence
with coordinisability; all of the ingredients in our arguments are present in his exposition.
The later unpublished draft \cite{madarasz} extends these ideas to a multi-sorted setting, but without full proofs.
}

Our effective result yields an algorithm translating intuitionistic proofs of functionality
into NRC definitions. 
In contrast, extraction procedures related to the Curry-Howard correspondence typically take
as input constructive proofs, possibly with cuts, of statements
of the type $\forall x \exists y ~\varphi(x,y)$ witnessing that $\varphi(x,y)$ defines
a total relation
and turn those proofs into programs for functions
$f$ such that $\forall x~\varphi(x,f(x))$ hold.
Our procedure works on cut-free proofs that a formula defines a \emph{partial function}
using techniques more closely related to interpolation.
This leaves open the question of 
extracting $\nrc$ terms from constructive totality proofs.
Sazonov \cite{sazonovcoll} addressed this question for an untyped  analogue of $\nrc$.
He uses weak set theories based on intuitionistic Kripke-Platek set theory.
These theories are richer than the ones we use for functionality proofs.


\section{Preliminaries} \label{sec:prelims}

Despite their long history of study in several communities, we know of no succinct presentation of
the basics of nested collection transformation languages. 
So we will give a quick introduction  here that assumes no background.
Indeed, for the issues that we will be concerned with in this work, the 
aspects of these transformation languages that have been the focus
of most past work (e.g. integration with functional languages \cite{linq, ezra} and complexity
of evaluation \cite{koch})
 will not be critical.


\myparagraph{Nested relations}
We deal with schemas that describe objects of various
 \emph{types} given by the following grammar.
$$T, \; U \bnfeq \ur \bnfalt T \times U \bnfalt \unit \bnfalt \sett(T)$$
For simplicity throughout the remainder
we will assume only two basic types: the one-element type $\unit$
and $\ur$,
whose inhabitant are not specified further; according to the application we may
think of $\ur$ as being infinite or empty.
We call this  set the \emph{Ur-elements}.
From the Ur-elements and a unit type we can build up the set of types via
product and the power set operation.  
We use standard conventions for abbreviating types, with the $n$-ary product abbreviating
an iteration of binary products.
A \emph{nested relational schema} consists of declarations of 
variable names associated to objects of given types. 

\begin{myexmp} \label{ex:aschema}
An example nested relational schema declares two objects
$R: \sett(\ur \times \ur)$ and $S: \sett( \ur \times \sett(\ur))$.
That is, $R$ is a set of pairs of Ur-elements: a standard ``flat'' binary relation.
$S$ is a collection of pairs whose first elements are Ur-elements
and whose second elements are sets of Ur-elements.
\end{myexmp}

The types have a natural interpretation, which we refer to
as the \emph{universe over $\ur$}.  The unit type has a unique member and the members of $\sett(T)$
are the sets of members of $T$.
An \emph{instance} of such a schema is defined in the obvious way, or a
$\ur$-instance if we want to emphasize the set of Ur-elements on which it is based.
Notice that nested relational schemas allow one to describe programming language data structures that are 
built up inductively via the tupling and set constructors, rather than just sets of tuples. Thus the literature often refers also to
the types above as ``object types'' and to the ``complex object data model'' \cite{limsoonthesis,abiteboulbeeri}. In this
work we will sometimes refer to the interpretation of a variable in an instance of a nested relational
schema as an \emph{object}. The \emph{subobjects} of
an object are defined in the obvious way. For example, if
$o$ is an object of type $\sett(T)$, then it is of
the form $\{t_1, \ldots \}$, where each $t_i$ is a subobject of $o$
of type $T$.

For the schema in Example \ref{ex:aschema} above, assuming that $\ur = \mathbb{N}$, one possible instance has
$R = \{ \<4,  6\>, \<7,  3\>\}$ and 
$S = \{ \<4 , \{  6,  9 \} \> \}$.
\myeat{
$R = \{ \<c_4,  c_6\>, \<c_7,  c_3\>\}$ and 
$S = \{ \<c_4 , \{  c_6,  c_9 \} \> \}$.
Note that a relational schema as defined above
is a special case of a nested relational schema,
in which all the declarations have type $\sett(\ur \times \ldots \times \ur)$.
}

\myparagraph{Transformation languages for nested relations}
A \emph{nested relational transformation}  (over input schema $\inschema$ and output schema $\outschema$) is a function that
takes as input an instance of $\inschema$, and returns an instance of $\outschema$.
For example, suppose our input schema consists of a
declaration
$R: \sett(\ur \times \ur)$ and our output schema consists also of a declaration
$S : \sett(\ur \times (\sett(\ur))$.
Then one possible transformation would return the nested relation formed
by grouping on the first position: informally returning a set of pairs $\<a, s\>$ where $a$
is any Ur-element appearing in the first component of a tuple in the input $R$, and $s$
nt is the set of $b$ such that $\<a,b\>$ is in $R$.

\myparagraph{Transformation equivalence}
We say that two transformations are \emph{equivalent} if they
agree on all instances (finite and infinite) of  a given input schema over any set of Ur-elements. It will turn
out that for the transformations we are interested in, ``over any set of Ur-elements'' can be freely replaced
by ``over any infinite set of Ur-elements'' or ``over some fixed infinite set of Ur-elements''.
When we say that a transformation $\trans$ is \emph{expressible} in some class of transformations $C$, we mean that
there is a transformation $\trans'$ in $C$ that is equivalent to $\trans$ in the sense above.

\myparagraph{Nested Relational Calculus}
We review the main language for declaratively transforming  nested relations, Nested
Relational Calculus ($\nrc$). Each expression is  associated with 
an \emph{output type}, which are in the type system described above.
We let $\booltype$ denote the type $\sett(\unit)$. Then $\booltype$ has exactly
two elements, and will be used to simulate Booleans.

The grammar and typing rules of $\nrc$ expressions  are presented in
Figure \ref{fig:nrc_type}.
\begin{figure}[t]
\begin{mathpar}
\small
\inferrule*{ }{\Gamma, \; x : T, \; \Gamma' \vdash x : T}
\\
\inferrule*{ }{\Gamma \vdash () : \unit}
\and
\inferrule*{\Gamma \vdash \textit{e}_1 : T_1 \and \Gamma \vdash \textit{e}_2 : T_2}{\Gamma \vdash \< \textit{e}_1, \textit{e}_2 \> : T_1 \times T_2}
\and
\inferrule*{\Gamma \vdash \textit{e} : T_1 \times T_2 \and \mathsmaller{i \in \{1,2\}}}{\Gamma \vdash \pi_i(e) : T_i}
\\
\inferrule*{\Gamma \vdash \textit{e} : T}{\Gamma \vdash \{\textit{e}\} : \sett(T)}
\and
\inferrule*{\Gamma \vdash \textit{e}_1 : \sett(T_1) \and \Gamma, \; x : T_1 \vdash \textit{e}_2 : \sett(T_2)}{\Gamma \vdash \bigcup \left\{ \textit{e}_2 \mid x \in \textit{e}_1 \right\} : \sett(T_2)}
\\
\inferrule*{ }{\Gamma \vdash \emptyset_T : \sett(T)}
\and
\inferrule*{\Gamma \vdash \textit{e}_1 : \sett(T) \and \Gamma \vdash \textit{e}_2 : \sett(T)}{\Gamma \vdash \textit{e}_1 \cup \textit{e}_2: \sett(T)}
\and
\inferrule*{\Gamma \vdash \textit{e}_1 : \sett(T) \and \Gamma \vdash \textit{e}_2 : \sett(T)}{\Gamma \vdash \textit{e}_1 \setminus \textit{e}_2: \sett(T)}
\end{mathpar}
\caption{$\nrc$ syntax and typing rules}
\label{fig:nrc_type}
\end{figure}


The definition
of the free and bound variables of an expression is standard. For example,
the union operator $\bigcup \{ E \mid x \in R \}$ binds variable $x$.

The semantics of these expressions should be fairly evident.
If $E$ has type $T$, and has input variables $x_1 \ldots x_n$
of types $T_1 \ldots T_n$, respectively, then the semantics associates with $E$
a function that 
given a binding associating each free variable a value of the appropriate type,
returns an object of type $T$.
For example, the expression $()$ always returns the empty tuple, while
$\emptyset$ returns the empty set of type $T$. The expression
$\{e\}$ evaluates to $\{o\}$, where $e$ evaluates to $o$.

In the sequel, we thus assume that every $\nrc$ expression is implicitly
associated with an \emph{input schema}, which declares 
 a list of free variables and their input types,  $X_1: T_1 \ldots X_n: T_n$,
along
with an output type $S$.
We may write $E : T_1, \ldots, T_n \rightarrow S$
and
refer to $S$ as the \emph{output type} of $E$.
We often abuse notation by identifying an $\nrc$ expression with the associated
transformation. For example, if $E$ is an $\nrc$ expression and $\inobj$ is an object of the input
type of $E$, we will write $E(\inobj)$ for the output of 
(the function defined by) $E$ on $\inobj$.

As explained in \cite{limsoonthesis}, the following transformations  are definable
with their expected semantics.
\begin{compactitem}
\item  For every type $T$ there is an $\nrc$ expression
$=_T$ of type $\bool$ representing equality of elements of type $T$.
In particular, there is an expression $=_\ur$ representing equality between
Ur-elements. 
\item  For every type $T$ there is an $\nrc$ expression $\in_T$ of type
$\bool$ representing membership between an element of type $T$
in an element of type $\sett(T)$.
\end{compactitem}

Further, if $E$ is a $\nrc$ expression  with
free variable $x$ of type $T$ and $F$ is an expression of type $T$,
then the $\nrc$ expression
\[
\bigcup \{\{E\} \mid x \in \{F\} \}
\]
represents the query obtained by running $E$ with  $x$ set to the output of $F$.
Combining this with the first observations above, we can see that
for expressions $E_1$ and $E_2$ of type $T$, we have an expression representing
$E_1 =_T E_2$ of type $\bool$. Using this, we will often treat  $=_T$ and
$\in_T$ as additional constructors
of the language.

Boolean operations $\wedge, \vee, \neg$ can also be represented as $\nrc$ 
expressions with output type $\bool$.
For example $\neg~ x$ is just $\{()\} \setminus x$. Applying the observation
about composition as we did above, we see that given $E$ of type $\bool$ we can obtain
an expression $\neg ~ E$ of type $\bool$, and thus as we did with $=_T$  and $\in_T$ we
will treat the Boolean operations as primitives.

Arbitrary arity tupling and projection operations  $\<E_1,\ldots E_n\>$, $\pi_j(E)$ for $j >2$
can be seen as abbreviations for a composition of binary operations.
Further
\begin{compactitem}
\item 
If $B$ is an expression of type $\bool$ and $E_1, E_2$ expressions of
type $\sett(T)$, then there is an expression $\case(B,E_1,E_2)$ of type $\sett(T)$ that implements ``if $B$ then $E_1$ else $E_2$''. 
\item If $E_1$ and $E_2$ are expressions of type $\sett(T)$, then there is an expression
$E_1 \cap E_2$ of type $\sett(T)$.
\end{compactitem}
The derivations of these are not difficult.
For example, the conditional
required by the first item is given by:
\[
\bigcup \{E_1 \mid ~ x \in B\} \cup \bigcup \{E_2 \mid x \in (\neg ~ B)\} 
\]

\begin{myexmp} \label{ex:nrc}
Consider an input schema including a binary relation $F: \sett( \ur \times \ur ) $.
The transformation $\trans_{\projq}$ with input
$F$ returning the projection of $F$ on the first component can be expressed in
$\nrc$ as  $\bigcup \{ \{\pi_1(f) \} \mid f \in F \}$.
The transformation $\trans_{\filterq}$ with input $F$ and also $v$ of type
$\ur$ that filters $F$ down to those pairs which agree with $v$ on the first component can
be expresses in $\nrc$ as $\bigcup \left\{\case([\pi_1(f) =_\ur v], \{f\}, \emptyset) \mid f \in F \right\}$.
Consider now the transformation $\trans_{\groupq}$
that groups $F$ on the first component, returning an object of type
$\sett( \ur \times  \sett(\ur))$; this is the first
transformation mentioned in Example \ref{ex:invertibleexp}. The transformation
 can be expressed in $\nrc$ as $\bigcup \left\{ \{\< v,  \bigcup \{ \{\pi_2(f)\} \mid f \in \trans_{\filterq} \} \>\} \mid v \in \trans_{\projq}\right\}$.
Finally, consider the second transformation $\trans_{\flatten}$ mentioned in Example \ref{ex:invertibleexp}, that flattens an input $G$ of type $\sett(\ur \times \sett(\ur))$ .
This can be expressed in $\nrc$ as 
\begin{align*}
\bigcup \left\{   \bigcup \{ \{ \<\pi_1(g), x\>  \} \mid x \in \pi_2(g)  \}  \mid g \in G \right\}
\end{align*}
\end{myexmp}



The language $\nrc$ cannot define certain natural transformations whose output type is $\ur$,
such as, for instance, $\case(B,E_1,E_2)$ for $E_1$ and $E_2$ of sort $\ur$.
To get a canonical language for such transformations, we let $\nrcwget$ denote the extension of $\nrc$ with the family of operations
$\nrcget_T : \sett(T) \to T$ that extracts the unique element from a singleton.
$\nrcget$ was considered in \cite{limsoonthesis}, with connection to parallel
evaluation explored in \cite{suciuthesis}.
The semantics are: if $E$ returns a singleton set $\{x\}$, then $\nrcget_T(E)$
returns $x$; otherwise it returns some default object of the appropriate type.
The semantics of $\nrcget_T(x)$ on non-singleton $x$ is not particularly important;
to fix ideas, we can define for each type $T$ a default element $d_T$ that will be
the output of $\nrcget_T(x)$ when $x$ is not a singleton assuming that we have a
constant $c_0$ in $\ursort$: take $d_\ursort = c_0$,
$d_{\sett(T)} = \emptyset$, $d_\unit = ()$ and $d_{T_1 \times T_2} = (d_{T_1}, d_{T_2})$.
In \cite{suciuthesis}, it is shown that $\nrcget$ is not expressible
in $\nrc$ at sort $\ursort$.
However, $\nrcget_T$ for general $T$ is definable from $\nrcget_\ursort$
and the other $\nrc$ constructs. 




\myparagraph{$\deltazero$ formulas}
We need a logic appropriate for talking about nested relations.
A natural and well-known subset of first-order logic formulas with
a set membership relation are the $\deltazero$ formulas.
They are built up from equality of Ur-elements via the Boolean operators
$\vee , \neg$ as well as relativized existential and universal quantification. All terms involving
tupling and projections are allowed. 
Formally, we deal with multi-sorted first-order logic, with sorts corresponding
to each of our types. We use the following syntax for $\deltazero$ formulas and terms.
Terms are built using tupling and projections.
All formulas and terms are assumed to be well-typed in the obvious way, with the
expected sort of $t$ and $u$ being $\ursort$ in expressions
$t =_\ursort u$ and $t \neq_\ursort u$, while
in $t \in_T u$  the sort of $t$ is $T$ and the sort of $u$ is  $\sett(T)$.

$$
\begin{array}{lcl}
t, u &\bnfeq& x \bnfalt () \bnfalt \< t, u \>  \bnfalt \pi_1(t) \bnfalt \pi_2(t)
\\
\varphi, \psi &\bnfeq& t =_\ursort t' \bnfalt t \neq_\ursort t' \bnfalt \top \bnfalt \bot \bnfalt \varphi \vee \psi \bnfalt \phi \wedge \psi \bnfalt \forall x \in_T t~ \varphi(x) \bnfalt \exists x \in_T t~ \varphi(x)
\end{array}
$$
Note that there is no primitive negation or equalities for sorts other than $\ursort$.
This does not limit expressiveness of formulas with respect to classical semantics.
Negation $\neg \phi$ may be defined by induction on $\phi$ by dualizing every connective;
we write $\phi \Rightarrow \psi$ for $\neg \phi \vee \psi$ in the sequel.
Equality, inclusion and membership predicates may be defined as notations by induction on the involved types.
$$
\begin{array}{rcl !\qquad rcl}
t \in_T u &\eqdef& \exists z' \in u\; t =_T z'
&
t \subseteq_T u &\eqdef& \forall z \in_T t ~~ z \in_T u
\\
t =_{\sett(T)} u &\eqdef& {t \subseteq_T u} ~~\wedge ~~{u \subseteq_T t}&
t =_\unit u &\eqdef& \top \mbox{ \small (since all elements of  this type are equal) }\\
\multicolumn{6}{c}{t =_{T_1 \times T_2} u ~~\eqdef~~ {\pi_1(t) =_{T_1} \pi_1(u)}~~ \wedge~~ {\pi_2(t) =_{T_2} \pi_2(u)}
}\\
\end{array}
$$
Here we have not defined $\in$ at higher types as an atomic predicate, but
rather as a derived predicate. 
We can think of the kind of entailments we want to prove in terms
of these derived predicates, without use of a set-extensionality axiom:
$$(\forall z \in_T x~~ z \in_T y) ~~\wedge~~ (\forall z \in_T y~~ z \in_T x) ~~~\Rightarrow~~~ x =_{\sett(T)} y$$
Alternatively, we can think of them as new primitives with extensionality
as an axiom relating them to the other primitives we have given above.

\myeat{
in certain parts of our
holds in first-order logic without having to assume any additional axiom.
This will allow us to omit  y critical axioms needed to reason about $\deltazero$ formulas
are that the equality $=_\ur$ is a congruence for terms and relations,
and $\pi_i(t_1,t_2) = t_i$ holds for every $t_1, t_2$ and $i \in \{1,2\}$.
The distinction between
 is not a practical advantage,
since in proof steps about $\in_T$ and the related derived predicates like $=_T$,
we  will have proof rules that are very similar to the use of extensionality.
But the fact that we are dealing with a proof system with only equality
and products will be useful in comparing our proof system to prior intuitionistic
systems.
}

The notion of a formula $\phi$ \emph{entailing} another formula
$\psi$, writing $\phi \models \psi$, is the standard one in first-order logic,
meaning that every model of $\phi$ is a model of $\psi$.



\myparagraph{$\nrc$ and $\deltazero$ formulas}
Since we have a Boolean type in $\nrc$, one may ask about the expressiveness
of $\nrc$ for defining transformations of shape $T_1, \ldots, T_n \to \bool$.
It turns out that they are equivalent to $\deltazero$ formulas.
This gives one justification for focusing on $\deltazero$ formulas.

\begin{proposition} \label{prop:verify} There is a polynomial time
algorithm taking a $\deltazero$ formula $\phi(\vec x)$ as input and producing
 an $\nrc$ expression $\verify_\phi(\vec x)$ of type $\bool$
such that $\verify_\phi(\vec x)$ returns true if and only if $\phi(\vec x)$ holds.
\end{proposition}

This useful result is proved by an easy induction over $\phi$.





\section{Synthesizing transformations from intuitionistic proofs} \label{sec:effectivebeth}
We will now present our first main result, concerning synthesis of nested relational
transformations from proofs.

  We consider an input  schema $\inschema$ with one
   input object $\inobj$ and an output schema
 with one output object $\outobj$. Using product objects, we can
  easily model any nested relational transformation in this way.
  We deal with a $\deltazero$ formula $\phi(\inobj, \outobj, \vec a)$
  with distinguished variables $\inobj, \outobj$.
 Recall from the introduction that such  a formula 
\emph{implicitly defines $\outobj$ as a function of $\inobj$} if for each nested relation
  $\inobj$ there is at most
  one $\outobj$ such that $\phi(\inobj, \outobj, \vec a)$ holds for some
  $\vec a$.
  A formula $\phi(\inobj, \outobj, \vec a)$ \emph{projectively
  implicitly defines} a transformation $\trans$ from $\inobj$ to $\outobj$
  if  for each $\inobj$, $\phi(\inobj, \outobj, \vec a)$ holds for some
  $\vec a$ if and only if $\trans(\inobj)=\outobj$. 
We drop ``projectively'' if $\vec a$ is empty.


\begin{myexmp}\label{ex:impnrc}
Consider the transformation $\trans_{\groupq}$ from Example \ref{ex:nrc}.
It has a simple implicit $\deltazero$ definition as given in
Example \ref{ex:invertible}, which we can restate as follows. First, define the auxiliary formula
$\chi(x, p,R)$ stating that $\pi_1(p)$ is $x$ and $\pi_2(p)$ is the set of $y$ such
that $\<x, y\>$ is in   $R$ (the ''fiber of $R$ above $x$''):
$$\chi(x,p,R) ~~~\eqdef~~~ \pi_1(p) = x~~ \wedge~~ \left( \forall t' \in R~ \left[\pi_1(t') = x ~\Rightarrow~ \pi_2(t') \in \pi_2(p)  \right]\right) ~~\wedge~~ \forall z\in \pi_2(p)~ \< x,z \>  \in R$$
Then $T_{\groupq}$ is implicitly defined by $ \forall t \in R~~\exists p \in q~~\chi(\pi_1(t),p,R)) \wedge \forall p \in q~~ \chi(\pi_1(p),p,R)$.
\end{myexmp}


\begin{figure}[!]
\begin{mathpar}
\small
%
\inferrule*[left={Contraction}]{\Theta; \; \Gamma,  \varphi,  \varphi \vdash t \in_T u}{\Theta; \; \Gamma,  \varphi \vdash t \in_T u}
\and
%
%
%
%
\inferrule*
[left={$\in_\ursort$-R}]
{ }{\Theta, \; t \in_\ursort u; \; \Gamma \vdash t \in_\ursort u}
\\
\inferrule*
[left={$=_{\sett}$-R}]
{\Theta; \; \Gamma \vdash t \subseteq_T u \and
\Theta; \; \Gamma \vdash u \subseteq_T t
 }{\Theta; \; \Gamma \vdash t =_{\sett(T)} u}
\and
\inferrule*
[left={$=_\times$-R}]
{
\Theta; \; \Gamma \vdash \pi_1(t) =_{T_1} \pi_1(u)
\and
\Theta; \; \Gamma \vdash \pi_2(t) =_{T_2} \pi_2(u)
 }{\Theta; \; \Gamma \vdash t =_{T_1 \times T_2} u}
\and
\inferrule*
[left={$=_\unit$-R}]
{ }{\Theta; \; \Gamma \vdash t =_\unit u}
\and
\inferrule*
[left={$=_\ur$-R}]
{\Theta, \; t \in_\ursort z; \; \Gamma \vdash u \in_\ursort z \and z \notin \freevars(\Theta, \Gamma, t, u)}{\Theta; \; \Gamma \vdash t =_\ursort u}
\and
\inferrule*
[left={$\subseteq$-R}]
{\Theta, \; z \in_T t; \; \Gamma \vdash z \in_T u \and z \notin \freevars(\Theta; \; \Gamma, t, u)}{\Theta; \; \Gamma \vdash t \subseteq_T u}
\and
\inferrule*
[left={$\in_{\sett}$-R}]
{\Theta, \; t \in_{\sett(T)} v; \; \Gamma \vdash t =_{\sett(T)} u}{\Theta, \; t \in_{\sett(T)} v; \; \Gamma \vdash u \in_{\sett(T)} v}
\\

\inferrule*
[left={$\bot$-L}]
{ }{\Theta; \; \Gamma,\;  \bot \vdash t \in_T u}
\and
\inferrule*
[left={$\wedge$-L}]
{\Theta; \; \Gamma, \;  \phi,\;  \psi  \vdash t \in_T u}{\Theta; \; \Gamma, \; \phi \wedge \psi \vdash t \in_T u}
\and
\inferrule*
[left={$\vee$-L}]
{\Theta; \; \Gamma, \;  \phi \vdash t \in_T u \qquad
\Theta; \; \Gamma, \;  \psi \vdash t \in_T u
}{\Theta; \; \Gamma, \;  \phi \vee \psi \vdash t \in_T u}

\\

\inferrule*
[left={$\forall$-L}]
{\Theta, \; t \in_{T} z; \; \Gamma, \; \phi[t/y] \vdash v \in_{T'} w}{\Theta , \; t \in_T z; \; \Gamma, \; \forall y \in_T z ~~ \phi\vdash v \in_{T'} w}
\and
\inferrule*
[left={$\exists$-L}]
{\Theta, \; x \in_T y; \; \Gamma, \; \phi \vdash t \in_{T'} v \qquad x \notin \freevars(\Theta, \Gamma, y, t, v)}{\Theta; \; \Gamma, \; \exists x \in_T y~~\phi \vdash t \in_{T'} v}
\\

\inferrule*
[left={$=$-subst}]
{\Theta[y/x]; \; \Gamma[y/x] \vdash (v \in_T w)[y/x]}
{\Theta; \; \Gamma, \; x =_\ursort y \vdash v \in_T w}
\and
\inferrule*
[left={$\neq$-L}]
{ }{\Theta; \; \Gamma, \;  t \neq_\ursort t \vdash u \in_T v}
\\
\inferrule*
[left={$\times_\beta$}]
{\Theta[t_i/y]; \; \Gamma[t_i/y] \vdash (t \in_T u)[t_i/y] \and i \in \{1,2\}}{\Theta[\pi_i(\<t_1,t_2\>)/y]; \; \Gamma[\pi_i(\<t_1,t_2\>)/y] \vdash (t \in_T u)[\pi_i(\<t_1,t_2\>)/y]}
\and
\inferrule*
[left={$\times_\eta$}]
{\Theta[\<x_1,x_2\>/x]; \; \Gamma[\<x_1,x_2\>/x] \vdash (t \in_T u)[\<x_1,x_2\>/x] \and x_1,x_2 \notin \freevars(\Theta; \; \Gamma, t, u)}{\Theta; \; \Gamma \vdash t \in_T u}
\end{mathpar}

\caption{Our intuitionistic sequent calculus for proofs of implicit definability}
\label{fig:ourproofsystem}
\end{figure}


\myeat{In addition, we will restrict to implicit definitions of objects $o$ whose type is not
at an Ur-element. That is, we will ignore transformations that return a single Ur-element.
This restriction is justified because from a $\Sigma(\ldots, o) $ implicitly defining  
$o$  as an argument
of some inputs, where $o$
has Ur-element type, we can easily generate an implicit definition $\Sigma'(\ldots o') $ describing the singleton
sets containing such a $o$. Assuming we can extract an $\nrc$ expression $E_{\Sigma'}$ capturing the value of $o'$,
we can then use $\nrcget(E'_{\Sigma'})$ as an $\nrcwget$ expression capturing the  transformation implicitly defined
by $\Sigma$.}


\myparagraph{Restricted proof system} 
Our synthesis result requires a proof of functionality within a restricted proof system.
We present a special-purpose sequent calculus in Figure~\ref{fig:ourproofsystem} deriving judgments $\Theta; \; \Gamma \vdash \phi$
where $\Gamma$ is a multi-set of $\deltazero$ formulas, $\Theta$ a multi-set of membership formulas $t \in u$, and $\phi$ is a $\deltazero$
formula with one of the following shapes: $t \in_T u$, $t =_T u$ or $t \subseteq_T u$.
A multi-set of formulas will also be called a \emph{context}, and
above we write $C \msu \; C'$ for the
concatenation of contexts $C$ and $C'$.
Informally, a judgment $\Theta; \; \Gamma \vdash \phi$ is meant to be read as
``If all the containments in $\Theta$ and formulas in $\Gamma$ hold, then $\phi$ does''.
In the figure, we use $\freevars$ to denote the free variables of a context,  and we use $\phi[t/x]$ to denote
the result of substituting $t$ for $x$ in $\phi$.

The main essential restriction on the proof system is that it is intuitionistic. There is no 
way to deduce $\Theta; \; \Gamma \vdash \phi$ from $\Theta; \; \Gamma, \; \neg \phi \vdash \bot$
in general.
Informally, this means that we forbid reasoning by contradiction. In particular, this means that some
sequents are classically valid but not derivable in our calculus. For instance, consider
$w \in r ; \; \forall x \in l~ l \in r, \; \forall y \in w~ l \in r \vdash l \in r$.
This is seen to be classically valid by considering separately the following three cases: $l$ non-empty, $w$ non-empty and
$l = w = \emptyset$. However, it is also easy to check that this cannot be derived intuitionistically.
The other restrictions, such as the specific shape of formulas on the right-hand side for many rules,
do not limit the power of the system when it comes to functionality proofs, but allow us to prove
our main extraction result more easily.

It is straightforward to capture the informal reasoning used to argue
for functionality in Example 
\ref{ex:invertible} within our proof system.
We also note that many natural proof rules are \emph{admissible} in our system; they are
conservative in terms of the set of proofs that they enable. We collect the most useful cases in
Figure~\ref{fig:admissible-rules}. Showing that they are admissible is done by rather elementary
inductions, and it can be noted that eliminating those additional proof rules
can be done in polynomial time in the size of proof trees and the types of the involved formulas.
This list is not meant to be exhaustive, as it can be shown that the derivable sequents 
in our system are exactly those derivable in more standard sequent calculus for multi-sorted intuitionistic logic that appear in the prior literature (see e.g.~\cite[Section 4.1]{jacobsbook}).
We offer a detailed discussion of the correspondence between our
proof system and several previously known intuitionistic calculi in the supplementary materials.

A technicality is that in our presentation of the proof system there is a slight asymmetry between
how the set predicates $=_T$, $\subseteq_T$ and $\in_T$ are treated
on the left and on the right. The  proof rules  decomposing formulas
on the right, such as $\subseteq$-\textsc{R},
are specialized to deal with the semantics of these predicates.
They are justified either based on extensionality --  if one thinks
of these predicates as primitive --  or by definition, if one
thinks of these predicates as derived.
On the other hand, on the left side we require that all of our
formulas in $\Gamma$ are described in the basic grammar
of  $\deltazero$ formulas, which does not have these predicates
as atomic. We do this only for convenience,
to avoid having additional proof rules
capturing  extensionality in decomposing formulas on the left.
\begin{figure}
\begin{mathpar}
\small \inferrule*[left=wk]{\Theta; \; \Gamma \vdash \psi}{\Theta; \; \Gamma, \; \varphi \vdash \psi}
\and
\inferrule*[left=ax]{ }{\Theta; \; \psi \vdash \psi}
\and
\inferrule*[left={$\in$-l}]{\Theta, \; t \in_T u; \; \Gamma \vdash \psi}{\Theta; \; \Gamma, \; t \in_T u \vdash \psi}
\and
\inferrule*[left={$\subseteq$-l}]
{\Theta, \; t \in_T v; \; \Gamma \vdash \psi}
{\Theta, \; t \in_T u; \; \Gamma, \; u \subseteq_T v \vdash \psi}
\and
\inferrule*[left={$\Rightarrow$}-l] {\Theta; \; \Gamma, \; \theta \vdash \psi }{\Theta; \; \Gamma, \; \varphi \Rightarrow \theta, \; \varphi \vdash \psi}
\and
\inferrule*[left={$=$-r}]{ }{\Theta; \; \Gamma \vdash t =_T t}
%
\end{mathpar}
\caption{Some typical admissible rules.}
\label{fig:admissible-rules}
\end{figure}


\myparagraph{Provably implicit definitions}
By an intuitionistic proof that $\Sigma(\inobj,\outobj, \vec a)$ implicitly defines
$\outobj$ as a function of  $\inobj$  we mean
a formal derivation of a sequent $\Sigma(\inobj, \outobj, \vec a), \; \Sigma(\inobj, \outobj', \vec a') \vdash \outobj =_T \outobj'$ in our proof system. 
\myeat{
Such a proof witnesses that 
$\Sigma(\inobj,\outobj, \vec a)$ projectively
implicitly defines $\outobj$ over $\inobj$.
}

We can now state our main  result on effectively generating
$\nrc$ expressions from proofs:
\begin{theorem} \label{thm:betheffective} There is a $\ptime$ procedure
which takes as input an intuitionistic  proof 
that $\Sigma(\inobj, \outobj, \vec a)$
defines $\outobj$ as a function of  $\inobj$, and returns
an $\nrc$ expression $E$
such that whenever $\Sigma(\inobj, \outobj, \vec a)$ holds,
then $E(\inobj)= \outobj$.
\end{theorem}

Let us provide a detailed example to illustrate Theorem~\ref{thm:betheffective}.

\begin{myexmp}
\label{ex:distinguishers}
Given a set of sets of Ur-elements $X \in \sett(\sett(\ursort))$, say that an Ur-element $a$
\emph{distinguishes} a set $x \in X$ if $x$ is the unique element of $X$ containing $a$.
Consider the transformation taking as input such an $X$ and returning the set of Ur-elements that distinguish
some element of $X$.
This is implicitly definable by a $\deltazero$ formula $\Sigma(X,o)$ stating that
every $a$ in  $o$ distinguishes some element of $X$ and conversely.
Writing this in our restricted syntax for $\deltazero$ formulas, in which
membership of higher-order objects must be expressed  using bounded quantification and equality,
we obtain an implicit definition
$$
\def\arraystretch{1.3}
\begin{array}{c}
\Sigma(X, o) ~~\eqdef~~ \left(\forall a \in o ~~ \exists x \in X ~~ \psi(X,x,a)\right) ~\wedge~ \left(\forall x \in X~~ \forall a \in x ~ \left[\chi(X,x,a) ~\Rightarrow~ a \in_{\ur} o \right]\right) \quad \text{where} \\
\chi(X,x,a) ~~\eqdef~~ \forall y \in X~\left(a \in_{\ur} y ~\Rightarrow~ x =_{\sett(\ur)} y\right)
\quad \text{and}\quad \psi(X, x, a) ~~\eqdef~~ a \in_{\ur} x ~\wedge~ \chi(X,x,a)
\end{array}
$$
Note that when  $a \in_{\ur} x$, $x =_{\sett(\ur)} y$, and $a \in_{\ur} o$ occur
on the left side of a sequent, they should be thought of as abbreviations for
more complex formulas built up through bounded quantification.
Similarly $\Rightarrow$ is a derived connective, built up from the Boolean
operations allowed in $\deltazero$ formulas in the obvious way.
\end{myexmp}

\newcommand\myLeftLabel[1]{\LeftLabel{\scriptsize #1}}
\newcommand\myRightLabel[1]{\RightLabel{\scriptsize #1}}

\begin{figure}
\[
\text{
\AXC{ }
\dashedLine
\myRightLabel{(7)}
\myLeftLabel{\textsc{ax}}
\UIC{${\color{red} {\color{black} z} \in o}, \; x \mathrel{\color{red} \in} X, \; z \mathrel{\color{red} \in} x; \; {\color{blue} {\color{black} z} \in o'} \vdash z \in {\color{blue} o'}$}
\myRightLabel{(6)}
\myLeftLabel{\rnimplL}
\dashedLine
\UIC{${\color{red} {\color{black} z} \in o}, \; x \mathrel{\color{red} \in} X, \; z \mathrel{\color{red}\in} x; \; {\color{red} \chi({\color{black}X},{\color{black}x},{\color{black}z})}, \; {\color{blue} \chi({\color{black}X},{\color{black}x},{\color{black}z}) \Rightarrow {\color{black}z} \in o'} \vdash z \in {\color{blue} o'}$}
\myLeftLabel{\rnallL}
\myRightLabel{(5)}
\UIC{${\color{red} z \in o}, \; x \mathrel{\color{red} \in} X, \; {\color{red} z \in {\color{black} x}}; \; {\color{red} \chi({\color{black}X},{\color{black}x},z)}, \; {\color{blue} \forall a \in {\color{black} x}~~ \left(\chi({\color{black} X},{\color{black} x},a) \Rightarrow a \in o'\right)} \vdash {\color{red} z} \in {\color{blue} o'}$}
\myLeftLabel{\textsc{=-subst}}
\UIC{${\color{red} z \in o}, \; x \mathrel{\color{red}\in} X, \; {\color{red} z' \in {\color{black} x}}; \; {\color{red} z =_\ursort z'}, \; {\color{red} \chi({\color{black} X},{\color{black} x},z)}, \; {\color{blue} \forall a \in {\color{black} x}~~ \left(\chi({\color{black} X},{\color{black} x},a) \Rightarrow a \in o'\right)} \vdash {\color{red} z} \in {\color{blue} o'}$}
\myLeftLabel{\rnexL}
\UIC{${\color{red} z \in o}, \; x \mathrel{\color{red} \in} X; \; {\color{red} z \in {\color{black} x}}, \; {\color{red} \chi({\color{black} X},{\color{black} x},z)}, \; {\color{blue} \forall a \in {\color{black} x}~~ \left(\chi({\color{black} X},{\color{black} x},a) \Rightarrow a \in o'\right)} \vdash {\color{red} z} \in {\color{blue} o'}$}
\myLeftLabel{\rnandL}
\UIC{${\color{red} z \in o}, \; x \mathrel{\color{red} \in} X; \; {\color{red} \psi({\color{black} X},{\color{black} x},z)}, \; {\color{blue} \forall a \in {\color{black} x}~~ \left(\chi({\color{black} X},{\color{black} x},a) \Rightarrow a \in o'\right)} \vdash {\color{red} z} \in {\color{blue} o'}$}
\myRightLabel{(4)}
\myLeftLabel{\rnallL}
\UIC{${\color{red} z \in o}, \; {\color{red} x \in {\color{black} X}}; \; {\color{red} \psi({\color{black} X},x,z)}, \; {\color{blue} \forall y \in {\color{black} X}~~\forall a \in y~~ \left( \chi({\color{black} X},y,a) \Rightarrow a \in o'\right)} \vdash {\color{red} z} \in {\color{blue} o'}$}
\myLeftLabel{\rnandL}
\UIC{${\color{red} z \in o}, \; {\color{red} x} \mathrel{\color{red} \in} X; \; {\color{red} \psi({\color{black} X},x,z)}, \; {\color{blue} \Sigma({\color{black} X},o')} \vdash {\color{red} z} \in {\color{blue} o'}$}
\myRightLabel{(3)}
\myLeftLabel{\rnexL}
\UIC{${\color{red} z \in o}; \; {\color{red} \exists x \in {\color{black} X} ~~ \psi({\color{black} X},x,z)}, \; {\color{blue} \Sigma({\color{black} X},o')} \vdash {\color{red} z} \in {\color{blue} o'}$}
\myLeftLabel{\rnallL}
\UIC{${\color{red} z \in o}; \; {\color{red} \forall a \in o ~~ \exists x \in {\color{black} X} ~~ \psi({\color{black} X},x,a)}, \; {\color{blue} \Sigma({\color{black} X},o')} \vdash {\color{red} z} \in {\color{blue} o'}$}
\myLeftLabel{\rnandL}
\UIC{${\color{red} z \in o}; \; {\color{red} \Sigma({\color{black} X},o)}, \; {\color{blue} \Sigma({\color{black} X},o')} \vdash {\color{red} z} \in {\color{blue} o'}$}
\myRightLabel{(2)}
\myLeftLabel{$\subseteq$-\textsc{R}}
\UIC{$\cdot \; ; \; {\color{red} \Sigma({\color{black} X},o)}, \; {\color{blue} \Sigma({\color{black} X},o')} \vdash {\color{red} o} \subseteq {\color{blue} o'}$}
\myRightLabel{(1)}
\myLeftLabel{$=_\sett$-\textsc{R}}
\dashedLine
\UIC{$\cdot \; ; \; {\color{red} \Sigma({\color{black} X},o)}, \; {\color{blue} \Sigma({\color{black} X},o')} \vdash {\color{red} o} = {\color{blue} o'}$}
\DisplayProof}
\]
\caption{Formal proof tree of functionality for Example~\ref{ex:distinguishers}. Admissible rules are denoted with dashed lines and some instances of the admissible weakening rule (\textsc{wk}) are omitted for legibility.
Formulas and variables specific to the left and right-hand side are respectively colored in red and blue.
}
\label{fig:distinguishers-prooftree}
\end{figure}



Figure~\ref{fig:distinguishers-prooftree} contains a formal derivation of functionality for $\Sigma(X,o)$.
We may render this proof informally as follows (putting references to proof steps in Figure~\ref{fig:distinguishers-prooftree} in parentheses).

\begin{proof}[Proof of functionality of Example~\ref{ex:distinguishers}]
Assume $\Sigma(X,o)$ and $\Sigma(X,o')$. To show $o = o'$, we need to show
that $o \subseteq o'$ and $o' \subseteq o$. Since the roles of $o$ and $o'$ are symmetric,
without loss of generality, it suffices to give the proof that $o \subseteq o'$ (1).
So fix $z \in o$ (2). Since $\Sigma(X,o)$ holds, according to its first conjunct,
we have in particular that there exists some $x \in X$ such that $\psi(X,x,z)$ holds (3).
Because $\Sigma(X,o')$ holds and $x \in X$, the second conjunct tells us that for every $a \in x$,
we have $\chi(X,x,a) \Rightarrow a \in o'$ (4).
Recall that $\psi(X,x,z)$ is the conjunction of $z \in x$ and $\chi(X,x,z)$, so that we may
deduce that $\chi(X,x,z) \Rightarrow z \in o'$ (6) and thus $z \in o'$ (7).
\end{proof}

As per Theorem~\ref{thm:betheffective}, the transformation defined in Example~\ref{ex:distinguishers}
is $\nrc$-definable as 
$$\bigcup \left\{ \case(\verify_{\theta}(X,a), \{a\}, \emptyset) \mid a \in \bigcup X\right\} \qquad \text{\small with $\quad \theta(X,a) = \exists x \in X~ \psi(X,x,a)$}$$
where $\verify$ is the filtering function given by Proposition \ref{prop:verify}.


We emphasize that our results apply to proofs of functionality over \emph{any
subsignature} of the input. In particular they apply to synthesize inverses
of transformations, a problem of considerable interest in several communities
\cite{invertloris, invertswarat}:
\begin{myexmp} \label{ex:invertiblereturn}
Return to the setting of Example \ref{ex:invertible},
and suppose that we  are interested in  the transformation over an input
object $G$ of type $\sett(\ur \times \sett(\ur))$ which simply ``flattens'' $G$.
We write this explicitly in $\nrc$,
as we did in Example \ref{ex:nrc}:
\[
E=\bigcup \left\{ \bigcup \{ \{ \< \pi_1(g), t \> \} \mid t \in \pi_2(g) \}  \mid g \in G\right\}
\]
From $E$ we can automatically generate a $\deltazero$ formula such as $\Sigma$ from Example \ref{ex:invertible},
stating that $F$ is the output of $G$ under $E$. Indeed, this is true for any
$\nrc$ transformation: one just encodes the semantics
of $\nrc$ in logic.  

This transformation is invertible, 
as mentioned in  Example \ref{ex:invertible}, and we can prove
 its invertibility in our calculus.
Our synthesis algorithm will generate from this proof an
expression in $\nrc$ that represents the  inverse, namely an expression
that groups $F$ to form $G$.
\end{myexmp}


\begin{myexmp} \label{ex:views}
Another  application are for the synthesis result of
Theorem \ref{thm:betheffective}  is to rewrite
transformations using cached results, a variation
on the idea of ``rewriting with views'' in relational databases 
\cite{dataint,tomanweddell,alonviews,NSV,afratichirkovabook}.

Consider a sequence where assigns to variable $J$ of type $\sett(\ur \times \ur)$ the intersection of
 $A$ and $B$,  and later assigns to variable $S$ of type $\sett(\ur)$ the set of elements that
have a self-loop in both $A$ and $B$.
\[
J := A \cap B; \ldots; S :=
\bigcup \left\{ \bigcup \{ \case(\pi_1(a) = \pi_2(a) = \pi_1(b) = \pi_2(b), \{\pi_1(a)\}, \emptyset) \mid a \in A\} \mid b \in B\right\}; \ldots
\]
One can easily see that $S$ is a function of $J$. And from a proof of functionality, our
method produces a rewriting of
the assignment producing $S$, using an  $\nrc$ expression that makes use of $J$. An example
of such a rewriting is
\begin{align*}
S := \bigcup \{  \case(\pi_1(j) = \pi_2(j), \{\pi_1(j)\}, \emptyset) \mid j \in J \}
\end{align*}
Such a rewriting of $S$ using the cached value of $J$ may be much more efficient than
recomputing $S$ from scratch.
\end{myexmp}


We now turn to explaining the ingredients that underlie the procedure of Theorem \ref{thm:betheffective}.

\myparagraph{Interpolation for $\deltazero$ formulas} Often a key ingredient in 
moving from implicit to explicit definition is  an \emph{interpolation theorem},
stating that for each entailment between formulas $\phi_L$ and $\phi_R$ there is an  intermediate formula
(an  \emph{interpolant} for the entailment), 
which is entailed by $\phi_L$ and entails $\phi_R$ while using only symbols common to $\phi_L$ and
$\phi_R$.
We can show using a standard inductive approach to interpolation (e.g. \cite{fittingbook}) that our calculus admits  efficient 
interpolation.

\begin{proposition}
\label{prop:interpolation}
Let $\Theta_L$, $\Theta_R$, $\Gamma_L$ and $\Gamma_R$ be contexts
and $\psi$ a formula and call $C = \freevars(\Theta_L, \Gamma_L) \cap \freevars(\theta_R, \Gamma_R)$ the set of common free variables.
For every derivation
$\Theta_L, \; \Theta_R; \; \Gamma_L, \; \Gamma_R \vdash \psi$
there exists a $\deltazero$ formula $\theta$ with $\freevars(\theta) \subseteq C$ such that the following holds
$$\Theta_L; \; \Gamma_L \models \theta \qquad \qquad \text{and} \qquad \qquad \Theta_R; \; \Gamma_R, \; \theta \models \psi$$
Further the interpolant $\theta$ can be found in polynomial time from the derivation.
\end{proposition}

The interpolation result above should be thought of  as giving us the result we want for 
transformations of \emph{Boolean} type.
From it we can derive that a formula whose truth value  is implicitly defined by a set of input
variables must be given as a $\deltazero$ formula over those inputs. 
By Proposition \ref{prop:verify}, these
formulas can be converted to $\nrc$.


\myparagraph{The higher-type interpolation lemma}
Our main result is deduced from a more general  interpolation result, 
which says that whenever a binary relationship between variables, such
as the containment relationship $t \subseteq_T u$,  is provable from
a theory that is partitioned into left and right formulas, and the variables
$t$ and $u$ appear exclusively in distinct sides of the partition, then
there is an \emph{interpolating expression} in $\nrcwget$, taking as input the
variables common to the left and right partitions. For an equality
relationship between variables, the synthesized expression will take
as input the common variables on the left and right
and select an object that is equal to the
variables participating in the equality.
For membership relationships $t \in u$,
our algorithm derives a bounding expression $E$ taking inputs in the common
signature such that $t \in E$; this could be strengthened to $t \in E \subseteq u$.
The result
 bears some similarity with other extraction procedures that produce
a program from a proof, such as those based on the Curry-Howard
correspondence.
However, it is formally much closer to the kind of interpolation theorem from
logic mentioned earlier in connection to Proposition \ref{prop:interpolation}.
\myeat{
An interpolation theorem can
be restated as saying that when there is a provable entailment
between  formulas $\phi_1$ and $\phi_2$, there is a formula $\chi$
sitting between them -- entailed by $\phi_1$ and entailing $\phi_2$, 
where $\chi$ uses only symbols that are common to $\phi_1$ and $\phi_2$.
}
In the past, interpolation results
have been applied to extract program invariants \cite{vampireinterp1,mcmillaninterp1}; here
we are proving and applying interpolation results to produce a different kind
of program artifact.

\begin{lemma} \label{lem:invar} [Higher-type Interpolation Lemma]
Let $\Theta = \Theta_L, \Theta_R$ be a $\in$-context and $\Gamma = \Gamma_L, \Gamma_R$ a context.
Suppose that $t$ and $u$ are terms of suitable types such that $\freevars(t) \subseteq \freevars(\Theta_L,\Gamma_L)$
and $\freevars(u) \subseteq \freevars(\Theta_R, \Gamma_R)$ and call
$C = \freevars(\Theta_L, \Gamma_L) \cap \freevars(\Theta_R, \Gamma_R)$ the set of common free variables.
Then we have:
\begin{compactitem}
\item If $\Theta; \; \Gamma \vdash t =_{T} u$ is derivable, there is an $\nrcwget$ expression $E$ of type $T$ such that
\[\Theta; \; \Gamma \models t = E = u \quad \text{and} \quad \freevars(E) \subseteq C\]
\item If $\Theta; \; \Gamma ~\vdash~ t \subseteq_T u$ is derivable, there is an $\nrcwget$ expression $E$ of type $\sett(T)$ such that
\[\Theta; \; \Gamma \models t \subseteq E \subseteq u \quad \text{and} \quad \freevars(E) \subseteq C\]
\item If $\Theta;\; \Gamma~\vdash~ t \in_T u$ is derivable, then there is an $\nrcwget$ expression $E$ of type $\sett(T)$ such that
\[\Theta; \; \Gamma \models t \in E \quad \text{and} \quad \freevars(E) \subseteq C\]
\end{compactitem}
Further the desired expressions can be constructed in time polynomial in the proof.
\end{lemma}

\begin{proof}[Proof of Theorem~\ref{thm:betheffective}]
A proof that $\Sigma(\inobj,  \outobj, \vec a)$ 
defines $\outobj$ as a function of $\inobj$ is
exactly a proof that
$\Sigma(\inobj, \outobj, \vec a) \;, \;  \Sigma(\inobj, \outobj', \vec a')  ~ \vdash 
\outobj  =_{T} \outobj'$
where $\outobj'$ and $\vec a'$ are  new variables. Applying Lemma~\ref{lem:invar} with 
$\Theta$ empty, $\Gamma_L = \Sigma(\inobj, \outobj, \vec a)$, 
 and $\Gamma_R = \Sigma(\inobj, \outobj', \vec a')$
yields an $\nrcwget$ expression $E(\inobj)$ such that 
$\Sigma(\inobj,\outobj, \vec a), \; \Sigma(\inobj, \outobj', \vec a') \models 
\outobj = E(\inobj) = \outobj'$.
Hence we have $\Sigma(\inobj, \outobj, \vec a) \models \outobj = E(\inobj)$ and the proof of Theorem \ref{thm:betheffective} is complete.
\end{proof}


Lemma~\ref{lem:invar} is proven by induction on the derivation, which requires examining every proof rule in Figure~\ref{fig:ourproofsystem}.
The more interesting cases are the left-hand side rules for first-order connectives
(\rnandL, \rnorL, \rnallL~and \rnexL)
and the rules for the right-hand side formulas $\in_\sett$-\textsc{R} and $=_\ur$-\textsc{R}.
Regarding the left-hand side rules, since the right-hand side formula of both the premise and 
conclusion
 is of the shape $t \in_T u$, the inductive invariant requires us to output an $\nrc$ expression bounding the term $t$.
To prove the inductive step, we use the binary union operator $E_1 \cup E_2$ of $\nrc$ for the
rule \rnorL~and the big union operator $\bigcup \{ E \mid x \in y \}$ for the rule \rnexL.
On the other hand, the inductive steps for the rules \rnandL~and \rnallL~do not require modifying the
expression obtained as part of the induction hypothesis. 
To treat the inductive steps corresponding to the rules
$\subseteq$-\textsc{R} and $=_\ur$-\textsc{R}, we use a combination of 
the usual ``Boolean'' interpolation (Proposition~\ref{prop:interpolation}) and the conversion of $\deltazero$ formulas to 
expressions of Boolean type in $\nrc$ (Proposition~\ref{prop:verify}).

\begin{myexmp} \label{ex:synth}
Let us illustrate the algorithm provided by Lemma~\ref{lem:invar} on the proof tree in Figure~\ref{fig:distinguishers-prooftree} by providing the corresponding intermediate $\nrc$ expressions that are synthesized, starting from top to bottom: from step (7) to (5), the $\nrc$ expression is the singleton $\{z'\}$. After the conclusion of the subsequent \rnexL~rule, the expression becomes
$$\bigcup \{ \{z'\} \mid z' \in x\}$$
which is semantically equivalent to $x$.
After the next \rnexL~rule at step (3), we obtain
$$\bigcup \left\{ \bigcup \{ \{ z'\} \mid z' \in x\} \mid  x \in X \right\}$$
which is equivalent to the union $\bigcup X$.
The final expression is then obtained right after step (2), by first computing an
interpolant $\theta(X,z)$ such that $z \in o \wedge \phi(X,o) \models \theta(X,z)$
and $\theta(X,z) \wedge \phi(X,o') \models z \in o'$. Computing according to the
procedure underlying Proposition~\ref{prop:interpolation} yields $\theta(X,a) = \exists x \in X~ \psi(X,x,a)$
and the final $\nrc$ expression
$$
\bigcup \left\{ \case(\verify_{\theta}(X,a), \{a\}, \emptyset) \mid a \in \bigcup \left\{ \bigcup \{ \{z'\} \mid z' \in x \} \mid x \in X \right\}\right\}
$$
\end{myexmp}


We now detail two cases of the inductive argument required to prove Lemma~\ref{lem:invar},
the other cases being relegated to the supplementary materials. We also omit the routine complexity analysis of the
underlying algorithm.

\subparagraph*{\bf Rule $\forall$-\textsc{L}:}
Assume that the last proof rule used introduces a universal quantifier on the left.
\begin{mathpar}
\inferrule*
[left={\rnallL}]
{\Theta, \; w \in_T y ; \; \Gamma, \; \phi[w/x] \vdash t \in_{T'}}{\Theta, \; w \in_T y; \; \Gamma, \; \forall x \in_T y~~\phi \vdash t \in_{T'} v}
\end{mathpar}
To simplify matters, assume that $w$ is a variable.
We apply the induction hypothesis to obtain a NRC expression, say $E'$ with $\freevars(E') \subseteq \{w\} \cup C$, by splitting the $\Theta, \; w \in_T y; \;\Gamma, \; \phi[w/x]$ in the obvious way (e.g., if $\forall x \in_T y ~~\phi$ was on the left context in the conclusion, we make $\phi[w/x]$ part of the left context in the premise).
If $w \notin \freevars(E')$, then it also satisfies the invariant in the conclusion. Otherwise, it must be the case that $y \in C$. Hence, we may show that the invariant is satisfied by
\[ E ~~=~~ \bigcup\{E' \mid w \in y\}\]


\subparagraph*{\bf Rule $\subseteq$-\textsc{R}:}
If the last proof rule used introduces an inclusion on the right
\begin{mathpar}
\inferrule*
[left={$\subseteq$-\textsc{R}}]
{\Theta, \; z \in_T t; \; \Gamma \vdash z \in_T u \qquad z \notin \freevars(\Theta, \; \Gamma, t, u)}{\Theta; \; \Gamma \vdash t \subseteq_T u}
\end{mathpar}
then the inductive hypothesis gives us an expression $E'$ such that
$\Theta, \; z \in_T t; \; \Gamma \models  z \in_T E'$ and $\freevars(E') \subseteq C$.
Apply interpolation to the premise so as to obtain a $\deltazero$ formula $\theta$ with $\freevars(\theta) \subseteq \{z\} \cup C$ such that
$\Theta_L; \; \Gamma_L, \; z \in_T t \models \theta$ and
$\Theta_R; \; \Gamma_R, \; \theta \models z \in_T u$.
In this case, we take $E = \{ z \in E' \mid \theta \}$,
which is $\nrcwget$-definable as
$$\bigcup \{ \case(\verify_\theta, \{ z \}, \emptyset) \mid z \in E' \}$$
Now, let us assume that $\Gamma$ holds and show that $t \subseteq E$ and $E \subseteq u$.
\begin{compactitem}
\item Suppose that $z \in t$. By the induction hypothesis, we know that $z \in E'$.
But we also know that $\Gamma_L$ is satisfied, so that $\theta$ holds.
By definition, we thus have $z \in E$.
\item Now suppose that $z \in E$, that is, that $z \in E'$ and $\theta$ holds.
The latter directly implies that $z \in u$ since $\Gamma_R$ holds.
\end{compactitem}



\section{Interpretations and nested relations} \label{sec:fointerp}
We will be interested in extending our synthesis result to classical
proofs. But first we give another characterization of $\nrc$,
an equivalence with transformations defined by \emph{interpretations}.

We first review the notion of an interpretation, which has become
a common way of defining transformations using logical expressions \cite{interpmikolaj,interpcolcombet}.
Let $\inschema$ and $\outschema$ be  multi-sorted vocabularies.
A first-order interpretation with input signature $\inschema$ and output
signature $\outschema$ consists of:
\begin{compactitem}
\item for each output sort $\sort'$, a sequence of input sorts $\interpsort(\sort') = \vec{\sort}$,
\item a formula $\phi^{\sort'}_\equiv(\vec x_1, \vec x_2)$ for each output sort $\sort'$ in $\outschema$ (where both tuples of variables $\vec x_1$ and $\vec x_2$ have types $\interpsort(\sort')$),
\item a formula $\phi^{\sort'}_\domainof(\vec x_1)$ for each output sort $\sort'$ in $\outschema$ (the variables $\vec x_1$ have types $\interpsort(\sort')$),
\item a formula $\phi_R(\vec x_1 , \ldots \vec x_n)$ for every relation $R$
of arity $n$ in  $\outschema$ (where the variables $\vec x_i$ have types $\interpsort(\sort'_i)$, provided the $i$-th argument of $R$ has sort $\sort'_i$),
\item for every function symbol $f(x_1, \ldots, x_k)$ of $\outschema$ with output sort $\sort'$ and input $x_i$ of sort $\sort_i$, a sequence of terms
$\overline{f}_1(\vec x_1, \ldots, \vec x_k), \ldots, \overline{f}_m(\vec x_1, \ldots, \vec x_k)$ with sorts $\interpsort(\sort_{out})$ and $\vec x_i$ of sorts $\interpsort(\sort_i)$.
\end{compactitem}
subject to the following constraints:
\begin{compactitem}
\item $\phi_\equiv^\sort(\vec x, \vec y)$ should define a partial equivalence relation, i.e. be symmetric
and transitive,
\item $\phi_\domainof^\sort(\vec x)$ should be equivalent to $\phi_\equiv^\sort(\vec x, \vec x)$,
\item $\phi_R(\vec x_1, \ldots, \vec x_n)$ and
$\phi_\equiv^{\sort_i}(\vec x_i, \vec y_i)$ for $1 \le i \le n$,
where $\sort_i$ is the output sort associated with position $i$
of the relation $R$, should jointly imply
$\phi_R(\vec y_1, \ldots, \vec y_n)$.
\item the formulas $\varphi_\equiv^\sort$ should be congruent with the interpretation of terms:
for every output function symbol $f(x_1, \ldots, x_k)$ represented by terms $\overline{f}_1(\vec x_1, \ldots, \vec x_k)$, \ldots, $\overline{f}_m(\vec x_1, \ldots, \vec x_k)$, writing $\vec x$ for the concatenation of $\vec x_1, \ldots, \vec x_k$, and $\vec y$ for the concatenation of $\vec y_1, \ldots, \vec y_k$, we enforce
\[\forall \; \vec x \; \vec y\quad
\left(\bigwedge^k_{i=1} \varphi^{\sort_i}_\equiv\left(\vec x_i, \vec y_i\right) \quad\Longrightarrow\quad
\varphi^{\sort'}_\equiv\left(
\overline{f}_1(\vec x), \ldots, \overline{f}_m(\vec x),
\overline{f}_1(\vec y), \ldots, \overline{f}_m(\vec y)
\right) \right)\]
where $\sort'$ is the sort of the output of $f$ and the $\sort_i$ correspond to the arities.
\end{compactitem}

In $\phi^{\sort}_\equiv$ and $\phi^{\sort}_\domainof$,  each $\vec x_1, \vec x_2$ is a tuple
containing
variables of sorts agreeing with the prescribed sequence of input sorts for $\sort'$.
Given a structure $M$ for the input sorts and a sort $\sort$ we call
 a binding of these variables to input elements of the appropriate
input sorts an \emph{$M,\sort$ input match}.
If in output relation $R$
 position $i$ is of sort $\sort_i$, then  in $\phi_R(\vec t_1 , \ldots \vec t_n)$ we
require
$\vec t_i$ to be a tuple of variables
of sorts agreeing with the prescribed sequence of input sorts for $\sort_i$.
Each of the above formulas
is over the vocabulary of $\inschema$.
An interpretation
$\interp$ defines a function from structures over vocabulary $\inschema$ to structures
over vocabulary $\outschema$ as follows:
\begin{compactitem}
\item The domain of sort $\sort'$ is the set of equivalence classes of the partial equivalence
relation defined by $\phi_\equiv^{\sort'}$ over the $M,\sort'$ input matches.
\item A relation $R$ in the output schema is interpreted by the set of those tuples $\vec a$ such
that
$\phi_R(\vec t_1 , \ldots \vec t_n)$ holds for some $\vec t_1 \dots \vec t_n$ with each $\vec t_i$ a 
representative of $a_i$.
\end{compactitem}

An interpretation $\interp$ also defines a map $\phi \mapsto \phi^*$ from formulas over $\outschema$ to
formulas over $\inschema$ in the obvious way. This map commutes with all logical connectives and thus
preserves logical consequence.

In the sequel, we are concerned with interpretations preserving certain theories
consisting of sentences in first-order logic. Recall that a \emph{theory}
in first-order logic is just a set of sentences.
Given a theory $\Sigma$ over $\inschema$ and a theory $\Sigma'$ over $\outschema$, we say that
$\interp$ is an interpretation of $\Sigma'$ within $\Sigma$ if $\interp$ is an interpretation
such that for every theorem $\phi$ of $\Sigma'$, $\phi^*$ is a theorem of $\Sigma$.
Since $\phi \mapsto \phi^*$ preserves logical consequence, if $\Sigma'$ is generated by a set
of axioms $A$, it suffices to check that $\Sigma$ proves $\phi^*$ for $\phi \in A$.

Finally, we are also interested in interpretations restricting to the identity on part of the input.
Suppose that $\outschema$ and $\inschema$ share a sort $\sort$. An interpretation $\interp$ of $\outschema$
within $\inschema$ is said to preserve $\sort$ if the output sort associated to $\sort$ is $\sort$ itself
and the induced map of structures is the identity over $\sort$. Up to equivalence, that
means we  fix $\phi_\domainof^T(x)$ to be, up to equivalence, $\top$,
$\phi_\equiv^\sort(x,y)$ to be the equality $x = y$
and map constants of type $\sort$ to themselves.

\myparagraph{Interpretations defining nested relational transformations}
We now consider how to define  nested relational transformations via interpretations.
The main idea will be to restrict all the constituent formulas to be
$\deltazero$ and to relativize the notion of interpretation to a background theory
that corresponds to our sanity axioms about tupling and sets.

We define the notion of \emph{component types} of a type $T$ inductively as follows.
\begin{compactitem}
\item $T$ is a component type of $\sett(T')$ if $T=\sett(T')$ or if it is a component type of $T'$.
\item $T$ is a component type of $T_1 \times T_2$ if $T = T_1 \times T_2$ or if it is a component type of either $T_1$ or $T_2$.
\item The only component types of $\ur$ and $\unit$ are themselves.
\end{compactitem}
Note in particular that if we have a complex object of sort $T$,
the possible sorts over its subobjects are exactly the component types of $T$.

For every type $T$, we build a multi-sorted vocabulary $\aschema_T$ as follows.
\begin{compactitem}
\item The sorts are all component types of $T$, $\unit$ and $\booltype = \sett(\unit)$.
\item The function symbols are the projections, tupling,
the unique element of type $\unit$, the constants $\boolff, \booltt$ of sort $\booltype$
representing $\emptyset, \{()\}$
and a special constant $\obj$ of sort $T$.
\item The relation symbols are the equalities at every sort and the membership
predicates $\in_T$.
\end{compactitem}
Let $T_\oneobj$ be a type  which will represent  the type of
a complex object $\oneobj$.
We build a theory $\Sigma(T_\oneobj)$ on top of $\aschema_{T_{\oneobj}}$ from the following axioms:
\begin{compactitem}
\item Equality should satisfy the congruence axioms for every formula $\phi$
$$\forall x y~~ (x = y ~\wedge~ \phi ~~\Rightarrow~~ \phi[y/x])$$
Note that it is sufficient to require this for atomic formulas to infer it for
all formulas.
\item We require that projection and tupling obey the usual laws for every type of $\aschema_{T_{\oneobj}}$.
$$
\forall x^{T_1} ~ y^{T_2}~ \pi_1(\< x, y \>) = x \qquad
\forall x^{T_1} ~ y^{T_2}~ \pi_2(\< x, y \>) = y \qquad
\forall x^{T_1 \times T_2} ~ \< \pi_1(x), \pi_2(x) \> = x$$
\item We require that $\unit$ be a singleton and every $\sett(T)$ in $\aschema_{T_\oneobj}$
$$\forall x^\unit ~ () = x$$
\item Lastly our theory imposes set extensionality
$$\forall x^{\sett(T)} ~ y^{\sett(T)} ~~ \left([\forall z^T~ (z \in_T x \Leftrightarrow z \in_T y)] \Rightarrow x =_T y\right)$$
\end{compactitem}

Note that in interpretations we associate the input to
a structure that includes a distinguished constant. For example,
an input of type $\sett(\ur)$ will be coded by a structure
with an element relation, an Ur-element  sort, and a constant whose  sort is the
type  $\sett(\ur)$.
In other contexts, like $\nrc$ expressions and implicit definitions
of transformations, we considered inputs to be \emph{free variables}.
This is only a change in terminology, but it reflects the fact that in evaluating the
interpretation on any input $i_0$ we will keep the interpretation of the associated
constant fixed, while we need to look at multiple bindings of the variables in each formula in order to form
the output structure.

We will show that $\nrcwget$ expressions defining
transformations from a nested relation of type $T_1$ to 
a nested relation of type $T_2$  correspond to a subset of interpretations
of $\Sigma(T_2)$ within $\Sigma(T_1)$ that preserve $\ursort$.
The only additional restriction we impose is that all
formulas $\phi_\domainof^T$ and $\phi_\equiv^T$ in the definition of such an interpretation
must be $\deltazero$. This forbids, for instance, universal quantification over the whole
set of Ur-elements. We thus call a first-order interpretation of
$\Sigma(T_2)$ within $\Sigma(T_1)$ consisting of $\deltazero$ formulas
a \emph{$\deltazero$ interpretation of $\Sigma(T_2)$ within $\Sigma(T_1)$}.

We now describe what it means for such an interpretation to define
a transformation from an instance of one nested relational schema
to another; that is, to map one  object to another.
We will denote the distinguished constant lying in the input sort by
$\inobj$ and the distinguished constant in the output sort by
$\outobj$.
Given any object $o$ of type $T$, define $M_o$
as the least structure such that
\begin{compactitem}
\item every subobjects of $o$ is part of $M_o$
\item when $T_1 \times T_2$ is a component type of $T$ and $a_1, a_2$ are objects of sort $T_1, T_2$ of $M_o$, then $\<a_1,a_2\>$ is an object of $M_o$
\item a copy of $\emptyset$ is part of $M_o$ for every sort $\sett(T)$ in $\aschema_T$
\item $()$ and $\{()\}$ are in $M_o$ at sorts $\unit$ and $\bool$.
\end{compactitem}
The map $o \mapsto M_o$ shows how to translate an object to a logical
structure that is appropriate as the input of an interpretation.
Note that $M_o$ satisfies $\Sigma(T)$ and that every sort has at least one element
in $M_o$ and that there is one sort, $\bool$, which contains two elements;
these technicality are important to ensure that interpretation be expressive enough.

We now discuss how the output of an interpretation is mapped back to an object.
The output of an interpretation is a multi-sorted structure
with a distinguished constant $\outobj$ encoding the output nested relational schema,
but it is not technically a nested relational instance
as required by our semantics for nested relational transformations. For example,
an element of $M_{\sett(\ur)}$ is not a set of Ur-elements, but simply a value
connected to Ur-elements by a membership relation.
We can convert the output to a semantically appropriate entity
via a modification of the well-known Mostowski collapse \cite{mostowski}.
We define $\collapse(e,M)$ on elements $e$ of the domain of a structure $M$ for the multi-sorted encoding of
a schema, by structural induction on the type of $e$:
\begin{compactitem}
\item If $e$ has sort $T_1 \times T_2$ then we set
$\collapse(e,M)= \<\collapse(\pi_1(e), M), \collapse(\pi_2(e), M)\>$
\item If $e$ has sort $\sett(T)$, then we set
$\collapse(e,M)=\{\collapse(t,M) \mid t \in e\}$
\item Otherwise, if $e$ has sort $\unit$ or $\ursort$, we set $\collapse(e,M) = e$
\end{compactitem}
We now formally describe how $\deltazero$ interpretations define functions between objects
in the nested relational data model.
\begin{definition}
We say that a nested relational transformation $\trans$ from  $T_1$ to $T_2$
is defined by a $\deltazero$ interpretation $\interp$ if, 
for every object $\inobj$ of type $T_1$,
the structure $M$ associated with $\inobj$ is mapped to $M'$ where $\trans(\inobj)$ is equal to $\collapse(\outobj,M')$.
\end{definition}
We will often identify a $\deltazero$ interpretation with the corresponding transformation, speaking
of its input and output as  a nested relation (rather than the corresponding structure).
For such an interpretation $\interp$ and an input object $\inobj$ we write
$\interp(\inobj)$ for the output of the transformation   defined by $\interp$ on $\inobj$.

\begin{myexmp} \label{ex:nrcinterp}
Consider an input schema consisting of a single binary relation $R: \sett( \ur \times \sett(\ur ))$,
so an input object is a set of pairs, with each pair consisting of an Ur-element and a set
of Ur-elements. The corresponding theory is $\Sigma(\sett(\ur \times \sett(\ur)))$,
which has sorts $\sett(\ur \times \sett(\ur))$, $\ur \times \sett(\ur)$, $\sett(\ur)$ and $\ur$ and
relation symbols $\in_\ur$ and $\in_{\ur \times \sett(\ur)}$ and one equality symbol for each
above sort.
If we consider the following instance of the nested relational schema
\[
R_0= \{ \<a, \{a,b\}\>, \<a, \{a, c\}\>, \<b, \{a, c\}\> \}
\]
Then the corresponding encoded structure $M$ consists of:
\begin{compactitem}
\item $M^{\sett(\ur\times \sett(\ur))}$ containing only the constant $R_0$
\item $M^{\ur \times \sett(\ur)}$ consisting of the elements of $R_0$,
\item $M^\ur$ consisting of  $\{a, b, c\}$ 
\item  $M^{\sett(\ur)}$ consisting of the sets $\{a,b\}$, $\{a, c\}$,
\item $M^\unit = \{()\}$ and $M^\booltype = \{\emptyset, \{()\}\}$
\item the element relations interpreted in the natural way
\end{compactitem}

Consider the transformation that groups on the first component, returning
an output object of type $\sett(\ur \times \sett(\sett(\ur)))$.
This is a variation of the grouping transformation from Example
\ref{ex:invertibleexp} and Example \ref{ex:nrc}.
On the example input $R_0$ the transformation would return
\[
\{ \< a, \{ \{a,b\}, \{a, c\} \}\> , \<b, \{\{a, c\}\}\> \}
\]

The output  would be represented by
a structure having  sorts \\
$\sett(\ur \times \sett(\sett(\ur))), \ur \times \sett(\sett(\ur)), \ur, 
\sett(\sett(\ur))$ and $\sett(\ur)$ in addition to $\unit$ and $\bool$.
It is easy to capture this transformation with
a  $\deltazero$ interpretation. For example,
the interpretation  could code
the output sort $\sett(\ur \times \sett(\sett(\ur)))$ as $\sett(\ur \times
\sett(\ur))$, representing
each group by the corresponding Ur-element. 
\myeat{
\michael{I made the handwavy paragraph above  slightly less terrible, but it  would be good to add some detail
and also proofread.}
}
\end{myexmp}



We will often make use of the following observation about interpretations:

\begin{proposition}
\label{prop:fointerp-comp}
$\deltazero$ interpretations can be composed, and their composition corresponds to
the underlying composition of transformations.
\end{proposition}

The composition of nested relational interpretations amounts to the usual composition
of FO-interpretations (see e.g. \cite{xqueryinterp}) and an easy check that
the additional requirements we impose on nested relational interpretations are 
preserved.

We can now state the equivalence of $\nrc$ and interpretations formally:

\begin{theorem} \label{thm:coddnrc}
Every transformation in $\nrcwget$ can be translated effectively to a
$\deltazero$ interpretation.
Conversely, for every $\deltazero$ interpretation, one can effectively form an equivalent
$\nrcwget$ expression. The translation from $\nrcwget$ to interpretations
can be done in $\exptime$ while the converse translation can be performed in $\ptime$.
\end{theorem}
This characterization holds when equivalence is
over finite nested relational inputs and also when arbitrary nested relations are allowed
as inputs to the transformations.

From this theorem one can easily derive many of the ``conservativity results''; e.g.
 \cite{conservativity}, which states that every nested relational algebra 
query from flat type $(\sett(\ur^n)$  to flat types can be expressed in relational algebra:
we simply convert to an interpretation and then note that in going backward
from an interpretation to an $\nrc$ expression we will not introduce additional levels of
nesting on top of those present in the input and output.

Note that a number of very similar results occur in the literature.
The underlying idea in one direction is
that 
 one can  ``shred'' a transformation
of collections to work on a flat representation.
This   has been investigated in several communities
for $\nrc$ and related languages
\cite{cheneysigmod, xqueryinterp}, in databases going at least as far back as \cite{abiteboulbidoit}.
The connection extends to richer collection types such as multi-sets, which
have been the focus in  using the shredding technique in systems
\cite{grust_aval,cheneysigmod,ulrichphd}. Algorithms for shredding can also
be useful as a technique for lifting optimizations, such as incremental query processing,
from relational languages to nested languages \cite{kochlupei}. And even in the collection of
richer collection types, many of the conservativity properties of $\nrc$ are maintained \cite{wongconservativity}.
But with these additional
type-formers, one needs to move beyond first-order logic in the simulating language. Thus although they are still extremely relevant
to implementation,  reasoning with the resulting representations
becomes problematic.
The thesis \cite{ulrichphd} provides a detailed look at shredding techniques, and also additional
historical background.

  Results of \cite{koch} show
that a $\ptime$ translation of $\nrc$ expressions to interpretations would imply
a collapse of the complexity class $TA[2^{O(n)}, n]$ to $\pspace$, even at Boolean
type. 
The early paper \cite{simulation} proves a  translation of $\nrc$ similar to the
one in the first half of Theorem \ref{thm:coddnrc}
for flat-to-nested queries,
and the nested-to-nested case can be easily obtained from this. 
However \cite{simulation} does not  formalize
the output of the interpretation as an interpretation, and we will need this connection
to obtain our other characterizations.  In the context of the
XML query language XQuery, \cite{xqueryinterp} proves a transformation
to first-order interpretations over trees. As noted in \cite{koch},
there is a very close relationship between XQuery and $\nrc$, and the
translation to interpretations  in \cite{xqueryinterp} can be easily
lifted to $\nrc$.

There is also similarity to results
from the 1960's of  Gandy \cite{gandy}. Gandy
 defines a class of set functions that are similar
to $\nrc$, and shows that they are ``substitutable''. This is the core
of the argument for translating $\nrc$ to interpretations.


\section{Synthesizing interpretations from classical proofs} \label{sec:bethmodeltheoretic}
In Section \ref{sec:effectivebeth} we showed that from an intuitionistic proof that
$\Sigma(\inobj, \ldots, \outobj)$ defines $\outobj$ as a function of $\inobj$,
we could synthesize an $\nrc$ expression that produces $\outobj$ from $\inobj$.
One might believe such a ``witnessing theorem'' to be specific to intuitionistic
calculi. But we will now demonstrate that this result extends to classical proofs,
and that it is actually a general phenomenon connecting implicit definitions
to interpretations.
We will  show that whenever we have a $\deltazero$
specification where there is  a classical proof that the specification
is functional, 
we can generate an  interpretation that realizes the function.
We can then rely
on Theorem \ref{thm:coddnrc} from the previous
section to infer that an $\nrcwget$ expression  realizes
the function as well.
That is, we will prove:

\begin{theorem} \label{thm:bethinterp}
For any
$\deltazero$ formula $\Sigma(\inobj, \outobj, \vec a)$ which
implicitly defines  $\outobj$ as a function of $\inobj$, there is a $\deltazero$ interpretation
$\interp$ such that whenever $\Sigma(\inobj, \outobj, \vec a)$ holds,
then $\interp(\inobj)= \outobj$.
\end{theorem}
In particular, if in addition for each $\inobj$ there is some $\outobj$ and $\vec a$ such that
$\Sigma(\inobj, \outobj, \vec a)$ holds,  then the interpretation and
the formula define the same transformation.

Recall from Section \ref{sec:effectivebeth}
that projective implicit definitions allow extra parameters $\vec a$
while implicit definitions allow only the input and output
variables $\inobj$ and $\outobj$.
From Theorem \ref{thm:bethinterp} we easily get the following characterization:

\begin{corollary} \label{cor:nrcbeth}
The following are equivalent for a transformation $\trans$:
\begin{compactitem}
\item $\trans$ is projectively implicitly definable by a $\deltazero$ formula
\item $\trans$ is implicitly definable by a $\deltazero$ formula
\item $\trans$ is definable via a $\deltazero$ interpretation
\item $\trans$ is $\nrcwget$ definable
\end{compactitem}
\end{corollary}

\myparagraph{Finite instances versus all instances}
In Theorem \ref{thm:bethinterp} and Corollary \ref{cor:nrcbeth} we emphasize that our 
results concern the
class $\funall$ of transformations $\trans$ such  that there is a $\deltazero$ formula 
$\Sigma$  which
defines a functional relationship between  $\inobj$ and $\outobj$
on all instances, finite and infinite, and where the function
agrees with $\trans$. We can consider $\funall$ as a class of transformations on all instances or of
finite instances, but the class  is defined by reference to all instances
for $\inobj$.  Expressed semantically
\[
\Sigma(\inobj, \outobj, \vec a) \wedge  \Sigma(\inobj,\outobj', \vec a') \models \outobj'= \outobj
\]
An  equivalent characterization of $\funall$  is \emph{proof-theoretic}: these are the 
transformations such that there is a  classical proof of functionality in a complete first-order proof system
using some basic axioms about Ur-elements, products and projection functions, and the extensionality axiom
for the membership relation. For example, it is easy to extend the intuitionistic
proof system given in Section \ref{sec:effectivebeth} to be complete for classical
entailment.

Whether one thinks of $\funall$ semantically or proof-theoretically, our results say that 
$\funall$ is identical with the set of transformations  given by $\nrc$ expressions.
But the proof-theoretic perspective is crucial for the synthesis procedure.

It is natural to  ask about the analogous class  $\funfin$ 
of transformations $\trans$ over \emph{finite inputs} for which there is a $\deltazero$ 
$\Sigma_\trans$ which
is functional, when only finite inputs are considered,  and where the corresponding function  agrees with  $\trans$.
It is well-known that $\funfin$ is not identical to $\nrc$ and is not so well-behaved. 
The transformation returning the powerset of a given input relation $\inobj$ is in $\funfin$: the powerset of
a finite input $\inobj$ is the unique collection $\outobj$ of subsets of $\inobj$
  that contains the empty set and such that for each element $e$ of $\inobj$,
if a set $s$ is in $\outobj$ then $s-\{e\}$ and $s \cup \{e\}$ are in $\outobj$.
From this we can see that $\funfin$
contains transformations of high complexity. Indeed, even when considering transformations from flat relations
to flat relations, $\funfin$ contains transformations whose membership in polynomial time  would imply that $\kw{UP} \cap \kw{coUP}$, the class of problems such that both the problem
and its complement  can be solved by an unambiguous non-deterministic polynomial time  machine,
is identical to $\ptime$ \cite{kolaitisimpdef}.
Most importantly for our goals, membership in $\funfin$ is \emph{not} witnessed by proofs in any effective proof system,
since this set is not computably enumerable.

\myparagraph{Total versus partial functions}
When we have a proof  that
$\Sigma(\inobj, \outobj, \vec a)$ defines $\outobj$ as a function of $\inobj$,
the corresponding function may still be partial. 
Our procedure will synthesize an expression $E$ defining a total function
that agrees with the  partial function defined by $\Sigma$.
If $\vec a$ is empty, we can also synthesize a Boolean $\nrc$ expression $\verify_\kw{InDomain}$
that verifies whether a given $\inobj$ is in the domain of the function: that is whether
there is $\outobj$ such that $\Sigma(\inobj, \outobj)$ holds.
$\verify_\kw{InDomain}$ can be taken as:
$$\bigcup \{ \verify_{\Sigma}(\inobj, e) \mid e \in \{ E(\inobj)\}\}$$
where $\verify_\Sigma$ is from Proposition \ref{prop:verify}.

Recall the second transformation from Example \ref{ex:invertibleexp}, where
the domain of the function is the set  of $G$ such that the second component of
each pair is never empty and the value of the second component is determined
by the value of the first component. This property can clearly be described
by a $\deltazero$ formula, and thus by Proposition \ref{prop:verify}
it can be verified in $\nrc$.
\myeat{
\begin{align*}
(\bigcup \{ \pi_1(g) \neq \emptyset \mid g \in G \} \wedge \\
\bigcup \{ \bigcup \{\pi_1(g)=\pi_1('g) \Rightarrow \pi_2(g)=\pi_2(g') \mid g' \in G\} \mid g \in G\}
\end{align*}
}

When $\vec a$ is not empty we cannot generate a domain check, since
the auxiliary parameters might enforce some second-order property of
$\inobj$: for example $\Sigma(\inobj, a, o)$ might state that $a$ is a bijection
from $\pi_1(\inobj)$ to $\pi_2(\inobj)$ and $o=\inobj$. This clearly defines
a functional relationship between $i_0, i_1$ and $o$, but
the domain consists of $i_0, i_1$ that have the same cardinality, which cannot
be expressed in first-order logic.



\myparagraph{Organization of the proof of the theorem}
Our proof of Theorem \ref{thm:bethinterp} will proceed first by some reductions
(Subsection \ref{subsec:reduction}),
showing that it suffices to prove a general result about implicit
definability and definability by interpretations
in  multi-sorted first-order logic, rather than
dealing with higher-order logic  and $\deltazero$ formulas.
In Subsection \ref{subsec:proofmultisorted} we sketch the argument for this
multi-sorted logic theorem.

\subsection{Reduction to a characterization theorem in multi-sorted logic} \label{subsec:reduction}
The first step in the proof of Theorem \ref{thm:bethinterp} is to reduce to a more general
statement relating implicit definitions in multi-sorted logic to
 interpretations.
The first part of this reduction is to argue that we 
can suppress auxiliary parameters $\vec a$ in implicit definitions:

\begin{lemma}
\label{thm:dropprojective}
For any $\deltazero$ formula $\Sigma(\inobj, \outobj, \vec a)$  that 
implicitly defines $\outobj$ as a function of
$\inobj$, there is another $\deltazero$ formula $\Sigma'(\inobj, \outobj)$ which 
implicitly  defines $\outobj$ as a function of $\inobj$,  such that
$\Sigma(\inobj, \outobj, \vec a) \Rightarrow \Sigma'(\inobj, \outobj)$.
\end{lemma}
The lemma is proven using two applications of \emph{classical} $\deltazero$ interpolation.
\begin{proposition}
\label{prop:interpolation-classical}
For any $\deltazero$ formulas $\varphi$ and $\psi$ such that $\varphi \models \psi$, there
exists another $\deltazero$ formula $\theta$ such that $\varphi \models \theta$
and $\theta \models \psi$.
\end{proposition}
This proposition generalizes Proposition~\ref{prop:interpolation}
since we allow classical validity for $\varphi \models \psi$. That being said,
we may prove Proposition~\ref{prop:interpolation-classical} using similar tools, i.e., a complete
cut-free sequent calculus for $\deltazero$ formulas and a standard proof as in~\cite{fittingbook}.
With Lemma~\ref{thm:dropprojective} in hand, from this point on we assume
that we do not have auxiliary parameters $\vec a$ in our implicit definitions.

\myeat{Furthermore, it is also sufficient to consider implicit definitions that
do not mention the constants
$c_i$ (of sort $\ur$). The reason why we can make this assumption is that, for any 
implicit definition $\Sigma(\inobj,\outobj)$ mentioning some 
constants $c_{i_1}, \ldots, c_{i_k}$, we can consider another implicit definition 
$\Sigma'(\inobj',\outobj)$ which does not mention any $c_i$ such that 
$\Sigma'(\<\inobj,c_{i_1}, \ldots\>,\outobj)$ is equivalent to $\Sigma(\inobj, \outobj)$,
where $\< \ldots \>$ denotes iterated tupling.
\begin{myexmp} \label{ex:elimconstants}
We give an example of how we reduce to the case without constants.
Consider $\Sigma(\inobj, \outobj)$ with $\inobj$ of type
$T$ that states only that $\outobj=c$ for some constant $c$ of type $T'$.
Clearly $\Sigma$ implicitly defines $\outobj$ as a function of  $\inobj$.

We let $\Sigma'(\inobj',\outobj)$ with $\inobj'$ of type $T \times T'$  be obtained
by replacing $c$ with $\pi_2(\inobj')$, thus giving a formula
stating that $\outobj=\pi_2(\inobj')$. Then $\Sigma'$ implicitly defines $\outobj$
as a  function of  $\inobj'$.
Once we convert $\Sigma'$ to the desired interpretation or $\nrc$ expression with
input $\inobj'$ we can  regain  an interpretation or expression  for the original
$\Sigma$ by substituting $\<\inobj, c\>$ for $\inobj'$. In this example, the $\nrc$
expression for $\Sigma'$ is just $\pi_2(\inobj')$, and thus the $\nrc$ expression
for $\Sigma$ is $\pi_2(\<\inobj, c\>)=c$.
\end{myexmp}}

\myparagraph{Reduction to Monadic schemas}
A \emph{monadic type} is a type built only using
the atomic type $\ursort$ and the type constructor
$\sett$. To simplify notation we define
$\ursort_0 \eqdef \ursort, \; \ursort_1 \eqdef \sett(\ursort_0), \ldots \ursort_{n+1} \eqdef \sett(\ursort_n)$.
A monadic type is thus a $\ursort_n$ for some $n \in \bbN$.
A nested relational schema is monadic if it contains only monadic types,
and a $\deltazero$ formula is monadic if all of its variables have monadic types.

Restricting to monadic formulas simplifies 
the type system significantly and thus, certain arguments by induction.
It turns out that by the usual ``Kuratowski encoding'' of pairs by sets, we can
reduce all of our questions about implicit versus explicit definability to the case of monadic 
schemas. The following proposition implies that we can derive all of our main results
for arbitrary schemas from their restriction to monadic formulas.
We will thus restrict to 
monadic formulas for the remainder of the argument.

\begin{proposition} \label{prop:reducemonadic}
For any nested relational schema $\aschema$, there is
a monadic nested relational schema $\aschema'$,
 an injection $\convert$ from instances of $\aschema$ to instances
of $\aschema'$ that is definable in $\nrc$, and an $\nrc[\nrcget]$ expression
$\convert^{-1}$
such that $\convert^{-1} \circ \convert$ is the identity transformation from
$\aschema \to \aschema$.

Furthermore, there is a $\deltazero$ formula $\image_\convert$ from $\aschema'$ to $\bool$
such that $\image_\convert(i')$ holds if and only if $i' = \convert(i)$
for some instance $i$ of $\aschema$.

These translations can also be given in terms of $\deltazero$ interpretations
rather than $\nrc$ expressions.
\end{proposition}


Given Proposition \ref{prop:reducemonadic} 
it suffices
to consider only monadic nested relational schemas.
Given a $\deltazero$  implicit definition
$\Sigma(\inobj, \outobj)$ we can form a new
definition that computes the composition of the following
transformations:  $\convert^{-1}_{\inschema}$, a
projection onto the first component, the transformation defined by $\Sigma$,
and $\convert_{\outschema}$. Our new definition  captures this composition by
 a formula $\Sigma'(\inobj', \outobj')$ that
defines $\outobj'$ as a function of $\inobj'$, where the formula
is  over a monadic schema.
Assuming that we have proven the theorem in the monadic case,
we would get an $\nrc$ expression $E'$ from $\inschema'$ to $\outschema'$ agreeing with  this formula
on its domain.
Now we can compose $\convert_{\inschema}$, $E'$, $\convert^{-1}_{\outschema}$,  and the
projection to
get an $\nrc$ expression agreeing with the partial function
defined by $\Sigma(\inobj, \outobj)$ on its domain, as required.

\myparagraph{Reduction to a result in multi-sorted logic}
Now we are ready to give our last reduction,
relating  Theorem \ref{thm:bethinterp} to a  general result
concerning multi-sorted logic.

Let $\sig$ be any multi-sorted signature, $\bigsorts$ be its
sorts and $\smallsorts$ be a subset of $\bigsorts$.
We say that a relation
$R$ is \emph{over $\smallsorts$} if all of its arguments are in
$\smallsorts$. Let $\Sigma$
be a set of sentences in $\sig$.
Given a model
$M$ for $\sig$, let $\smallsorts(M)$ be the union
of the domains of relations over $\smallsorts$,
and let $\bigsorts(M)$ be defined similarly.

We say that $\bigsorts$ is \emph{implicitly interpretable}
over $\smallsorts$ relative to $\Sigma$ if:

\medskip

For any models $M_1$  and $M_2$ of $\Sigma$,  if there is a mapping $m$
from $\smallsorts(M_1)$ to $\smallsorts(M_2)$ that preserves all
relations over $\smallsorts$, then $m$ extends to a unique mapping
from $\bigsorts(M_1)$ to $\bigsorts(M_2)$ which
preserves all relations over $\bigsorts$.

\medskip

Informally, implicit interpretability states  that
the sorts in $\bigsorts$ are semantically determined by
the sorts in $\smallsorts$. The property implies in particular  that
if $M_1$ and $M_2$ agree on the interpretation of sorts
in $\smallsorts$, then the identity mapping on sorts in $\smallsorts$
extends to a mapping that preserves sorts in $\bigsorts$.

We relate this semantic property to a syntactic one.
We say that $\bigsorts$ is \emph{explicitly interpretable}
over $\smallsorts$ relative to $\Sigma$ if
for all $\sort$ in $\bigsorts$ there is a formula $\psi_\sort(\vec x, y)$
where $\vec x$ are variables with sorts in $\smallsorts$, $y$ a variable of
sort $\bigsorts$, such that:

\begin{compactitem}
\item In any model $M$ of $\Sigma$, $\psi_\sort$ defines a partial function 
$F_\sort$
mapping  $\smallsorts$ tuples on to $\sort$.
\item For every relation $R$ of arity $n$ over $\bigsorts$, there is
a formula $\psi_R(\vec x_1, \ldots \vec x_n)$ using
only relations over $\smallsorts$ and
only quantification over $\smallsorts$ such that
in any model $M$ of $\Sigma$, the  pre-image of $R$
under the mappings $F_\sort$ for the different arguments
of $R$ is  defined by $\psi_R(\vec x_1, \ldots \vec x_n)$.
\end{compactitem}

Explicit interpretability states that there is an  interpretation in
the sense of the previous section that produces the structure in $\bigsorts$
from the structure in $\smallsorts$, and in addition  there is a definable
 relationship between
an element $e$ of a sort in $\bigsorts$ and the tuple that codes $e$ in the interpretation.
Note that $\psi_\sort$, the mapping between the elements $y$ in $\sort$  and
the tuples in $\smallsorts$ that interpret them, can use arbitrary relations.
The key property is that when we pull a relation $R$ over $\bigsorts$
back using the mappings $\psi_\sort$, then we obtain something definable
using $\smallsorts$.

With these definitions in hand,
we are  ready to state a result in  multi-sorted logic which
allows us to generate interpretations from classical
proofs of functionality:

\begin{theorem} \label{thm:general}
For any $\Sigma, \smallsorts, \bigsorts$ such that $\Sigma$ entails that a sort
of $\smallsorts$ has at least two elements, $\bigsorts$ is explicitly
interpretable over $\smallsorts$ if and only if it is implicitly interpretable
over $\smallsorts$.
\end{theorem}

This can be thought of as an analog of Beth's theorem \cite{beth,craig57beth} for multi-sorted logic.
The proof is sketched in the next subsection.
For now we explain how it implies Theorem \ref{thm:bethinterp}.
In this explanation we assume a monadic schema for both input and output. Thus every element $e$
in an instance  has sort $\ursort_n$ for some $n \in \bbN$.

Consider a $\deltazero$ formula $\Sigma(\inobj, \outobj)$ over a monadic schema
that implicitly defines
$\outobj$ as a function of  $\inobj$.
$\Sigma$ can be considered as a multi-sorted first-order formula with  sorts for
every subtype occurrence of the input as well as distinct sorts for every subtype 
occurrence of the output other than
$\ursort$.
Because we are dealing with  monadic input and output schema, every sort other than
$\ursort$ will be of the form $\sett(T)$, and these sorts
have only the element relations $\in_T$ connecting them.  We refer
to these as \emph{input sorts} and \emph{output sorts}. We modify $\Sigma$ by asserting
that all elements of  the input sorts lie underneath $\inobj$, and
all elements of the output sorts lie underneath $\outobj$, where an element $e$ is said to lies
underneath an element $e'$ if there is a chain $e=e_1 \in \ldots \in e_n=e'$.
Since $\Sigma$
was $\deltazero$, this does not change the semantics.
We also conjoin to $\Sigma$ the sanity axioms for the schema, including
 the extensionality axiom at the sorts corresponding
to each object type. Let $\Sigma^*$ be the resulting
formula. In this transformation, as was the case with interpretations,
we  change our perspective on inputs and outputs, considering
them as constants rather than as free variables. We do this only 
to match our result in multi-sorted logic, which  deals with
a set of \emph{sentences} in multi-sorted first-order logic, rather than formulas with
free variables.

Given models $M$ and $M'$ of $\Sigma^*$, we define relations $\equiv_i$ connecting
elements of $M$ of depth $i$ with elements of $M'$ of depth $i$. For $i=0$,
$\equiv_i$ is the identity: that is, it connects elements of $\ursort$ if and only if they
are identical. For $i=j+1$, $\equiv_i(x,x')$ holds exactly when for every $y \in x$
there is $y' \in x'$ such that $y \equiv_j y'$, and vice versa.

The fact that $\Sigma$ implicitly defines $\outobj$ as a  function of $\inobj$ tells us that:

\medskip 

Suppose $M \models \Sigma^*$, $M' \models \Sigma^*$ and $M$ and $M'$ are identical on the input sorts.
Then the mapping $m$ taking a $y \in M$ of depth $i$ to a $y' \in M'$ such that $y' \equiv_i y$ is
 an isomorphism of the output sorts that is the identity on $\ursort$. 
Further, any isomorphism  of $\bigsorts(M)$ on to $\bigsorts(M')$ that is the identity 
on $\ursort$ must be equal to $m$: one can show this by induction on the depth $i$ using
the fact that $\Sigma^*$ includes the extensionality axiom.

\medskip

From this, we see that the output sorts are implicitly interpretable over the input sorts
relative to $\Sigma^*$.
Using Theorem \ref{thm:general}, we conclude that the output sorts are explicitly interpretable
in the input sorts relative to $\Sigma^*$.
Applying the conclusion to the formula $x=x$, where $x$ is a
variable of a sort corresponding
to object type $T$ of the
output, we obtain  a first-order formula $\phi^T_{\domainof}(\vec x)$ over the input sorts.
Applying the conclusion to the formula $x=y$ for $x, y$ variables
corresponding to the object type $T$ we get a formula
$\phi_{\equiv_T}(\vec x, \vec x')$ over the input sorts. Finally applying
the conclusion
to the element relation $\epsilon_T$ at every level of the output,  we get a first-order
formula $\phi_{\epsilon_T}(\vec x, \vec x')$ over the input sorts. 
Because $\Sigma^*$ asserts that each element of the input sorts lies beneath a constant for $\inobj$, we can
convert all quantifiers to bind only beneath $\inobj$, giving us $\deltazero$ formulas.
It is easy to verify that these formulas give us the desired interpretation.
This completes the proof of Theorem \ref{thm:bethinterp}, assuming Theorem \ref{thm:general}.


\subsection{Proof of the multi-sorted logic result} \label{subsec:proofmultisorted}
In the previous subsection we reduced our goal result about generating interpretations
from proofs to  a result in multi-sorted first-order logic, Theorem \ref{thm:general}.
We will sketch the proof of Theorem \ref{thm:general}.  
The direction from explicit interpretability to
implicit interpretability is straightforward, so we will be
interested only in the  direction from implicit to explicit.
Although the theorem appears to be new,
each of the components is a variant of arguments that already appear in the model theory literature.

In the body of the paper we make use of  only quite basic results from model theory:
\begin{compactitem}
\item the \emph{compactness theorem} for first-order logic, which states
that for any  theory $\Gamma$, if every finite subcollection of $\Gamma$
is satisfiable, then $\Gamma$ is satisfiable;
\item
the \emph{downward Lowenheim-Skolem theorem}, which
states that if $\Gamma$ is countable and has a model, then
it has a countable model;
\item the \emph{omitting types theorem} for first-order logic.
A first-order theory $\Sigma$ is
said to be \emph{complete} if for
every other first-order sentence $\phi$ in the vocabulary of $\Sigma$, either $\phi$ or $\neg \phi$
is entailed by $\Sigma$. 
Given a set of constants $B$, a \emph{type} over
$B$ is an infinite collection $\tau(\vec x)$ of formulas using
variables $\vec x$ and constants $B$.
A type is \emph{complete} with respect to a theory
$\Sigma$ if every first-order formula  with
variables in $\vec x$ and constants from $B$ is either entailed
or contradicted by $\tau(\vec x)$ and $\Sigma$.
 A type $\tau$ is said to be \emph{realized} in 
a model $M$ if there is a $\vec x_0$ in $M$ satisfying all formulas in
$\tau$. $\tau$ is \emph{non-principal} (with respect to
a first-order theory $\Sigma$)
 if
there is no formula $\gamma_0(\vec x)$ such that $\Sigma  \wedge \gamma_0(\vec x)$
entails all of $\tau(\vec x)$. 
The version of the omitting types theorem that we will
use states  that:

\medskip
 if we have a countable set $\Gamma$ of complete types that are all non-principal relative to a complete
theory $\Sigma$, there
is some model $M$ of $\Sigma$ in which  none of the types in $\Gamma$ are realized.
\end{compactitem}
Each of these results follows from a standard
model construction technique \cite{hodgesbook}. 

We 
 can easily show that to prove the multi-sorted
result, it suffices to consider $\Sigma$ that
is a complete theory.
\begin{proposition} \label{prop:assumecomplete} Theorem \ref{thm:general}
follows from its restriction to $\Sigma$ a complete theory.
\end{proposition}

Recall that our assumption is that  $\Sigma$ yields a function from
$\inobj$ to $\outobj$. Our next step will be to show that the output of this
function is always ``sub-definable'': each element in the output is definable from
the input if we allow ourselves to guess  some parameters. For example, consider
the grouping transformation mentioned in Example \ref{ex:invertibleexp}
and Example \ref{ex:nrc}. Each output is obtained from
grouping input relation $F$ over some Ur-element $a$. So each member of the output is definable
from the input constant $F$ and a ``guessed'' input element $a$.
We will show that this is true in general.

Given a model $M$ of $\Sigma$ and $\vec x_0 \in \bigsorts$ within $M$, the
\emph{type of $\vec x_0$ with parameters
from $\smallsorts$} is the set of all formulas satisfied by 
$\vec x_0$, using any sorts and relations but only constants from $\smallsorts$.

A type $p$ is  \emph{isolated over $\smallsorts$} if
there is a formula $\phi(\vec x, \vec a)$ with parameters $\vec a$
from $\smallsorts$ such that  $M \models \phi(\vec x, \vec a) \rightarrow \gamma(\vec x)$
for each $\gamma \in p$.
The following is a step towards showing that elements in the output are well-behaved:

\begin{lemma} \label{lem:isolation}
Suppose $\bigsorts$ is implicitly interpretable over $\smallsort$ with respect to
$\Sigma$.
Then in any model $M$ of $\Sigma$ the type of any $\vec b$ over $\bigsorts$ with
parameters from $\smallsorts$ is isolated over $\smallsorts$.
\end{lemma}
\begin{proof}
Fix a counterexample $\vec b$, and let $\Gamma$ be the set of formulas in $\bigsorts$ with constants
from $\smallsorts$ satisfied by $\vec b$ in $M$.
We claim that there is a model $M'$ with $\smallsorts(M')$ identical to $\smallsorts(M)$
where there is no tuple  satisfying  $\Gamma$. This follows from the failure of
isolation and 
the omitting types theorem.

Now we have a contradiction of implicit interpretability, since the identity
mapping on $\smallsorts$ cannot extend to an isomorphism of relations over
$\bigsorts$ from $M$ to $M'$.
\end{proof}

The next step is to argue that every element of $\bigsorts$ is definable by a formula using
parameters from
$\smallsorts$. 
\myeat{This property, coupled with the ``reduction property'' defined further below, is referred to as \emph{coordinatisation} \cite{hodgesdugald,hodgesbook}, since it states that every element in
$\bigsorts$ is identified by ``co-ordinates'' from $\smallsorts$.
}

\begin{lemma} \label{lem:nonuniformdefinability}
Assume implicit interpretability of $\bigsorts$ over $\smallsorts$ relative to $\Sigma$.
In any model $M$ of $\Sigma$, for  every element $e$ of a sort $\bigsort$ in $\bigsorts$,
there is a first-order formula $\psi_e(\vec y, x)$  with variables $\vec y$
having sort in $\smallsorts$ and $x$ a variable of sort $\bigsort$,
along with a tuple $\vec a$ in $\smallsorts(M)$
such that $\psi_e(\vec a, x)$ is satisfied only by $e$ in $M$.
\end{lemma}

\myeat{
In a single-sorted setting, this can be found in \cite{hodgesbook} Theorem 12.5.8 where
it is referred to as ``Gaifman's coordinatisation theorem'', credited independently to unpublished work
of Haim Gaifman and Dale Myers.
The multi-sorted version is also a variant of Remark 1.2, part 4 in \cite{groupoids}, which points to a proof
in the appendix of \cite{modeltheorydiff};
the remark assumes that $\bigsorts$ is the set of all sorts. Another variation is Theorem 3.3.4 of
\cite{madarasz}.
}
\begin{proof}
Since a counterexample involves only formulas in a countable
language, by the Lowenheim-Skolem theorem mentioned
above, it is enough to consider the case where $M$ is countable.
By Lemma \ref{lem:isolation}, the type of every $e$ is isolated by
a formula $\phi(\vec x, \vec a)$ with parameters from $\smallsorts$
and relations from $\bigsorts$. We claim that $\phi$  defines $e$:
that is, $e$ is the only satisfier. If not, then there is $e' \neq e$ that satisfies
$\phi$.   Consider the relation $\vec e \equiv \vec e'$ holding
if $\vec e$ and $\vec e'$ satisfy all the same formulas using
relations and variables from $\bigsorts$ and
parameters from $\smallsorts$. 
Isolation implies  that $e \equiv e'$. Further, isolation of types
shows that $\equiv$ has the ``back-and-forth property''
given $\vec  d \equiv \vec d'$, and $\vec e$ we can obtain $\vec e'$
with $\vec d \vec e \equiv \vec d' \vec e'$.
To see this,  fix $\vec d \equiv \vec d'$ and
consider $\vec e$. We have $\gamma(\vec x, \vec y, \vec a)$ isolating the type
of $\vec d,\vec e$, and further $\vec d$ satisfies $\exists \vec y ~ \gamma(\vec x, \vec y, \vec a)$
and thus so does $\vec d'$ with witness $\vec e'$. But then using $\vec d \equiv \vec d'$
again we see that $\vec d, \vec e \equiv \vec d', \vec e'$. Using countability of
$M$ and this property we can inductively create a mapping on $M$ fixing $\smallsorts$ pointwise,  preserving
all relations in $\bigsorts$, and taking $\vec b$ to $\vec b'$. But this contradicts
implicit interpretability.
\end{proof}

\begin{lemma} \label{lem:uniformdefinability}
The formula in Lemma \ref{lem:nonuniformdefinability} can be taken to depend only
on the sort $\sort$.
\end{lemma}
\begin{proof}
Consider the type over the single variable $x$ in $\sort$
consisting of the formulas $\neg \isnonuniformdefiner_\varphi(x)$, taking $\isnonuniformdefiner_\varphi(x)$
to be defined as
\[\exists \vec b ~ [~ \varphi(\vec b, x) \wedge \forall x'~ (\varphi(\vec b, x') \Rightarrow x' = x)]\]
where the tuple $\vec b$ ranges over $\smallsorts$.
By Lemma~\ref{lem:nonuniformdefinability}, this type cannot be satisfied in a model of $\Sigma$.
Since it is unsatisfiable, by compactness, there are finitely many formulas $\varphi_1(\vec b,x) , \dots, \varphi_n(\vec b, x)$ such that $\forall x ~~ \bigvee_{i = 1}^n \isnonuniformdefiner_{\varphi_i}(x)$ is satisfied. Therefore, each $\varphi_i(\vec b, x)$ defines a partial function from tuples of $\smallsort$ to $\sort$
 and every element of $\sort$ is covered by one of the $\phi_i$.
Recall that we
 assumed  that $\Sigma$ enforces that $\smallsorts$ has a sort
with at least two elements.
Thus we can combine the $\phi_i(\vec b, x)$ into a single formula $\psi(\vec b, \vec c, x)$
defining a surjective partial function from $\smallsort$ to $\sort$
where $\vec c$ is an additional parameter in $\smallsorts$ selecting some $i \le n$.
\end{proof}


We now need to go from the ``sub-definability'' or ``element-wise definability''
result above to an interpretation.
Consider the formulas $\psi_\sort$ produced
by Lemma \ref{lem:uniformdefinability}. For a relation $R$ of arity $n$ over
$\bigsorts$,  where the $i^{th}$ argument has sort $\sort_i$, consider the formula
\[\psi_R(\vec x_1 \ldots \vec x_n) ~~=~~ \exists y_1 \ldots y_n ~ R(y_1 \ldots y_n) \wedge 
\bigwedge_i \psi_{\sort_{i}}(\vec x_i, y_i)
\]
where $\vec x_i$ is a tuple of variables of sorts in $\smallsorts$.
The formulas $\psi_\sort$ for each sort $\sort$ and the formulas $\psi_R$ for each
relation $R$ are as required by the definition of
explicitly interpretable, except that they may use quantified variables and relations
of $\bigsorts$, while we only
want to use variables and relations from $\smallsorts$. We take care of this in the following lemma,
which says that  formulas over $\bigsorts$
do not allow us to define any more subsets of $\smallsorts$ than we can
with formulas over $\smallsorts$.

\begin{lemma} \label{lem:uniformreduction} 
Under the assumption of implicit interpretability, 
for  every
formula $\phi(\vec x)$ over $\bigsorts$ with $\vec x$ variables
of sort in $\smallsorts$  there is a formula $\phi^\circ(\vec x)$ over
$\smallsorts$  -- that is, containing only variables, constants, and relations
from $\smallsorts$ -- such that for every model $M$ of $\Sigma$,
\[
M \models \forall \vec x ~ \phi(\vec x) \leftrightarrow \phi^\circ(\vec x)
\]
\end{lemma}
\begin{proof}
Assume not, with $\phi$ as a counterexample. By the compactness and Lowenheim-Skolem
theorems, we know that there is a countable model
$M$ of $\Sigma$ containing $\vec c$, $\vec c'$ that agree on all formulas in $\smallsorts$
but that disagree on $\phi$. As in Lemma \ref{lem:nonuniformdefinability}, we
can obtain a mapping on $M$ preserving $\smallsorts$ but sending $\vec c$ to $\vec c'$.
This contradicts implicit interpretability, since the mapping cannot be extended.
\end{proof}


Above we obtained the formulas $\psi_R$ for each relation symbol $R$
needed for an explicit interpretation.
We can obtain formulas defining the necessary equivalence relations
$\psi_{\equiv}$ and $\psi_{\domainof}$ easily from these.
Thus, putting Lemmas \ref{lem:nonuniformdefinability}, \ref{lem:uniformdefinability},
and \ref{lem:uniformreduction} together yields  a proof
of Theorem \ref{thm:general}.



\subsection{Putting it all together}
We summarize our results on extracting $\nrcwget$  expressions  from classical proofs of
functionality.
We have shown in Subsection \ref{subsec:reduction} how to convert the problem
to one with no extra variables other than input and output
and with only monadic schemas --
and thus no use of products or tupling. We also showed how to convert the resulting
formula  into a theory in multi-sorted first-order logic. That is, we no longer
need to talk about $\deltazero$ formulas.  

In Subsection \ref{subsec:proofmultisorted} we showed that from a theory in multi-sorted
first-order logic we can obtain an interpretation. This first-order
interpretation in a multi-sorted logic can then be converted
back to a $\deltazero$ interpretation, since the background theory forces
each of the input sorts in the multi-sorted structure to
correspond to a level of nesting below one of the constants corresponding
to an input object.
Finally, the results of Section \ref{sec:fointerp} allow us to convert this interpretation
to an $\nrcwget$ expression.
With the exception of the result in multi-sorted logic, all of the constructions are effective.
Further, these effective conversions are all in
polynomial time except for the transformation from an
interpretation to an $\nrcwget$ expression, which is exponential time in the worst case.
Outside of the multi-sorted result, which makes use of infinitary methods, the conversions
are each sound when equivalence over finite input structures is considered as well as the
default case when arbitrary inputs are considered. As explained in Subsection \ref{subsec:reduction},
when equivalence over finite inputs is considered, we cannot hope to get a synthesis result
of this kind.


\section{Conclusion} \label{sec:conc}
We have provided a method taking a proof that a logical formula
defines a functional transformation and generating an expression
in a functional transformation language
that implements it.
In the process we provide a more general synthesis procedure
(Lemma \ref{lem:invar}) that  can generate expressions
interpolating between variables whenever there is a provable containment.
This connection between provably functional formulas and the functional
transformation language
$\nrc$ studied in data management and programming languages is,
to our knowledge, new and non-trivial. 
\myeat{Our effective method relies
on an intuitionistic calculus for proving containments, memberships,
and  equalities among sets. But we also show that the same results
hold for classical logic,  relying on an equivalence between
$\nrc$ expressions and interpretations based on $\deltazero$
formulas that is specific to the
setting of nested sets,  combined with  a more general equivalence in multi-sorted logic, relating  transformations that are provably functional 
and those given by interpretations.
}

We are currently working on an implementation of our effective synthesis  result
in the $\coq$ proof assistant \cite{coq}.
This involves
formalizing the proof calculus, the semantics of $\deltazero$
formulas, the syntax and semantics of $\nrc$, 
in  $\coq$, as well as
the synthesis algorithm.
In addition to giving us a verified proof, we will gain the ability to create
proofs of functionality within a $\coq$ session,
allowing us to build up tactics and definitions on top of the basic rules of the proof calculus.

An open issue is to make the classical interpolation result effective. There is an obvious
extension of our proof system that gains completeness for classical logic: we allow
multiple disjuncts in the consequence, and revise the rules in the obvious way.
For instance, the rule $\in_{\sett(T)}$-\textsc{R} would become
\begin{mathpar}
\small \inferrule*
[left={$\in_{\sett}$-R}]
{\Theta, \; t \in_{\sett(T)} v; \; \Gamma \vdash t =_{\sett(T)} u,\; t_1 \in_{T_1} u_1, \ldots,\; t_k \in_{T_k} u_k}{\Theta, \; t \in_{\sett(T)} v; \; \Gamma \vdash u \in_{\sett(T)} v,\; t_1 \in_{T_1} u_1, \ldots,\; t_k \in_{T_k} u_k}
\end{mathpar}
Theorem \ref{thm:bethinterp} shows that when we have a proof in such a system we can create
an $\nrc$ definition, and we conjecture that it is possible to do this efficiently.
In fact, we can also show
that the higher-type interpolation lemma, Lemma \ref{lem:invar}, holds for classical entailment.
Although our proof of Lemma \ref{lem:invar} is via induction on proofs,
the extension for classical entailment can be done using model-theoretic techniques,
in particular a dichotomy theorem for automorphisms stemming from work
of  Makkai \cite{makkai}. We are investigating  an extension
of our  proof system that will allow us to lift our current inductive argument
for Lemma \ref{lem:invar} to the classical setting.
We conjecture that it will lead us to an efficient procedure for extracting NRC
terms from classical functionality proofs, thereby simultaneously generalizing
Theorem~\ref{thm:betheffective} and Theorem \ref{thm:bethinterp}.

\myeat{
\michael{This well be moved}
One application of our synthesis results is to rewrite
an $\nrc$ transformation $\trans$ in terms of a given set of expressions $V_1 \ldots V_n$.
Corollary \ref{cor:nrcbeth} in this paper already implies that in looking to 
implement  
$\trans$ over $V_1 \ldots V_n$ , it suffices to look for a rewriting in $\nrc$. 
}

In addition to the application areas exhibited  in Examples \ref{ex:invertiblereturn} and \ref{ex:views},
we think that procedures for generating implementations in functional languages
from implicit definitions should have other applications in programming languages
and verification. For example they could be relevant
for generating 
 programs transforming structured data in the context of
more specialized input structures,
such as strings and trees \cite{interpmikolaj}.  

We focused here on a stripped-down setting
where at the base level we have no additional structure, but many of our
results (e.g. Theorem \ref{thm:bethinterp}) generalize in the presence of additional
axiomatizable structure on the base set.  Another important direction is to generalize
the algorithmic development (e.g. Theorem~\ref{thm:betheffective}) to incorporate
specialized decision procedures available on this additional structure.


\begin{acks}
We are very grateful to Szymon Toru\'nczyk,  who outlined a route 
to show
that implicitly definable transformations over nested relations can be defined via interpretations, 
in the process  conjecturing  a  more general result concerning
definability in multi-sorted logic. Szymon also helped in simplifying the mapping of $\nrc$ expressions to interpretations, a basic component in one of our characterizations.
We also thank Ehud Hrushovski, who sketched a proof of the Beth-style
result for multi-sorted logic that serves as another component.
His proof proceeds along very similar lines
to the one we present in this paper, but makes use of a prior Beth-style result in 
classical model
theory \cite{makkai}.
This work was funded by EPSRC grant EP/M005852/1.
\end{acks}




\bibliography{obj}


\begin{thebibliography}{51}


\ifx \showCODEN    \undefined \def \showCODEN     #1{\unskip}     \fi
\ifx \showDOI      \undefined \def \showDOI       #1{#1}\fi
\ifx \showISBNx    \undefined \def \showISBNx     #1{\unskip}     \fi
\ifx \showISBNxiii \undefined \def \showISBNxiii  #1{\unskip}     \fi
\ifx \showISSN     \undefined \def \showISSN      #1{\unskip}     \fi
\ifx \showLCCN     \undefined \def \showLCCN      #1{\unskip}     \fi
\ifx \shownote     \undefined \def \shownote      #1{#1}          \fi
\ifx \showarticletitle \undefined \def \showarticletitle #1{#1}   \fi
\ifx \showURL      \undefined \def \showURL       {\relax}        \fi
\providecommand\bibfield[2]{#2}
\providecommand\bibinfo[2]{#2}
\providecommand\natexlab[1]{#1}
\providecommand\showeprint[2][]{arXiv:#2}

\bibitem[\protect\citeauthoryear{Abiteboul and Beeri}{Abiteboul and
  Beeri}{1995}]%
        {abiteboulbeeri}
\bibfield{author}{\bibinfo{person}{Serge Abiteboul} {and}
  \bibinfo{person}{Catriel Beeri}.} \bibinfo{year}{1995}\natexlab{}.
\newblock \showarticletitle{The Power of Languages for the Manipulation of
  Complex Values}.
\newblock \bibinfo{journal}{\emph{{VLDB} J.}} \bibinfo{volume}{4},
  \bibinfo{number}{4} (\bibinfo{year}{1995}), \bibinfo{pages}{727--794}.
\newblock


\bibitem[\protect\citeauthoryear{Abiteboul and Bidoit}{Abiteboul and
  Bidoit}{1986}]%
        {abiteboulbidoit}
\bibfield{author}{\bibinfo{person}{Serge Abiteboul} {and}
  \bibinfo{person}{Nicole Bidoit}.} \bibinfo{year}{1986}\natexlab{}.
\newblock \showarticletitle{Non First Normal Form Relations: An Algebra
  Allowing Data Restructuring}.
\newblock \bibinfo{journal}{\emph{J. Comput. Syst. Sci.}} \bibinfo{volume}{33},
  \bibinfo{number}{3} (\bibinfo{year}{1986}), \bibinfo{pages}{361--393}.
\newblock


\bibitem[\protect\citeauthoryear{Afrati and Chirkova}{Afrati and
  Chirkova}{2019}]%
        {afratichirkovabook}
\bibfield{author}{\bibinfo{person}{Foto Afrati} {and} \bibinfo{person}{Rada
  Chirkova}.} \bibinfo{year}{2019}\natexlab{}.
\newblock \bibinfo{booktitle}{\emph{Answering Queries Using Views}}.
\newblock \bibinfo{publisher}{Morgan \& Claypool Publishers}.
\newblock


\bibitem[\protect\citeauthoryear{Andr\'eka, Madar\'asz, and N\'emeti}{Andr\'eka
  et~al\mbox{.}}{2008}]%
        {madarasz}
\bibfield{author}{\bibinfo{person}{H. Andr\'eka}, \bibinfo{person}{J.~X.
  Madar\'asz}, {and} \bibinfo{person}{I. N\'emeti}.}
  \bibinfo{year}{2008}\natexlab{}.
\newblock \bibinfo{title}{Definability of New Universes in Many-sorted logic}.
\newblock
\newblock
\newblock
\shownote{manuscript available at
  \url{old.renyi.hu/pub/algebraic-logic/kurzus10/amn-defi.pdf}.}


\bibitem[\protect\citeauthoryear{Benedikt, Cate, Leblay, and Tsamoura}{Benedikt
  et~al\mbox{.}}{2016}]%
        {interpbook}
\bibfield{author}{\bibinfo{person}{Michael Benedikt},
  \bibinfo{person}{Balden~Ten Cate}, \bibinfo{person}{Julien Leblay}, {and}
  \bibinfo{person}{Efthymia Tsamoura}.} \bibinfo{year}{2016}\natexlab{}.
\newblock \bibinfo{booktitle}{\emph{Generating plans from proofs: the
  interpolation-based approach to query reformulation}}.
\newblock \bibinfo{publisher}{Morgan Claypool}.
\newblock


\bibitem[\protect\citeauthoryear{Benedikt and Koch}{Benedikt and Koch}{2009}]%
        {xqueryinterp}
\bibfield{author}{\bibinfo{person}{Michael Benedikt} {and}
  \bibinfo{person}{Christoph Koch}.} \bibinfo{year}{2009}\natexlab{}.
\newblock \showarticletitle{From {XQuery} to {Relational} {Logics}}.
\newblock \bibinfo{journal}{\emph{ACM TODS}} \bibinfo{volume}{34},
  \bibinfo{number}{4} (\bibinfo{year}{2009}), \bibinfo{pages}{25:1--25:48}.
\newblock


\bibitem[\protect\citeauthoryear{Beth}{Beth}{1953}]%
        {beth}
\bibfield{author}{\bibinfo{person}{E.~W. Beth}.}
  \bibinfo{year}{1953}\natexlab{}.
\newblock \showarticletitle{On {Padoa's} Method in the Theory of Definitions}.
\newblock \bibinfo{journal}{\emph{Indagationes Mathematicae}}
  \bibinfo{volume}{15} (\bibinfo{year}{1953}), \bibinfo{pages}{330 -- 339}.
\newblock


\bibitem[\protect\citeauthoryear{Bojanczyk, Daviaud, and Krishna}{Bojanczyk
  et~al\mbox{.}}{2018}]%
        {interpmikolaj}
\bibfield{author}{\bibinfo{person}{Mikolaj Bojanczyk}, \bibinfo{person}{Laure
  Daviaud}, {and} \bibinfo{person}{Shankara~Narayanan Krishna}.}
  \bibinfo{year}{2018}\natexlab{}.
\newblock \showarticletitle{Regular and First-Order List Functions}. In
  \bibinfo{booktitle}{\emph{LICS}}.
\newblock


\bibitem[\protect\citeauthoryear{Buneman, Naqvi, Tannen, and Wong}{Buneman
  et~al\mbox{.}}{1995}]%
        {natj}
\bibfield{author}{\bibinfo{person}{Peter Buneman}, \bibinfo{person}{Shamim~A.
  Naqvi}, \bibinfo{person}{Val Tannen}, {and} \bibinfo{person}{Limsoon Wong}.}
  \bibinfo{year}{1995}\natexlab{}.
\newblock \showarticletitle{Principles of Programming with Complex Objects and
  Collection Types}.
\newblock \bibinfo{journal}{\emph{Theor. Comput. Sci.}} \bibinfo{volume}{149},
  \bibinfo{number}{1} (\bibinfo{year}{1995}), \bibinfo{pages}{3--48}.
\newblock


\bibitem[\protect\citeauthoryear{Cheney, Lindley, and Wadler}{Cheney
  et~al\mbox{.}}{2014}]%
        {cheneysigmod}
\bibfield{author}{\bibinfo{person}{James Cheney}, \bibinfo{person}{Sam
  Lindley}, {and} \bibinfo{person}{Philip Wadler}.}
  \bibinfo{year}{2014}\natexlab{}.
\newblock \showarticletitle{Query shredding: efficient relational evaluation of
  queries over nested multisets}. In \bibinfo{booktitle}{\emph{SIGMOD}}.
\newblock


\bibitem[\protect\citeauthoryear{Colcombet and L{\"{o}}ding}{Colcombet and
  L{\"{o}}ding}{2007}]%
        {interpcolcombet}
\bibfield{author}{\bibinfo{person}{Thomas Colcombet} {and}
  \bibinfo{person}{Christof L{\"{o}}ding}.} \bibinfo{year}{2007}\natexlab{}.
\newblock \showarticletitle{Transforming structures by set interpretations}.
\newblock \bibinfo{journal}{\emph{Logical Methods in Computer Science}}
  \bibinfo{volume}{3}, \bibinfo{number}{2} (\bibinfo{year}{2007}).
\newblock


\bibitem[\protect\citeauthoryear{Cooper}{Cooper}{2009}]%
        {ezra}
\bibfield{author}{\bibinfo{person}{Ezra Cooper}.}
  \bibinfo{year}{2009}\natexlab{}.
\newblock \showarticletitle{The Script-Writer's Dream: How to Write Great {SQL}
  in Your Own Language, and Be Sure It Will Succeed}. In
  \bibinfo{booktitle}{\emph{{DBPL}}}.
\newblock


\bibitem[\protect\citeauthoryear{{Coq}}{{Coq}}{2020}]%
        {coq}
\bibfield{author}{\bibinfo{person}{{Coq}}.} \bibinfo{year}{2020}\natexlab{}.
\newblock \bibinfo{title}{The {Coq} Proof Assistant}.
\newblock
\newblock
\newblock
\shownote{\url{coq.inria.fr}.}


\bibitem[\protect\citeauthoryear{Craig}{Craig}{1957}]%
        {craig57beth}
\bibfield{author}{\bibinfo{person}{William Craig}.}
  \bibinfo{year}{1957}\natexlab{}.
\newblock \showarticletitle{Three Uses of the {Herbrand-Gentzen} Theorem in
  Relating Model Theory and Proof Theory}.
\newblock \bibinfo{journal}{\emph{Journal of Symbolic Logic}}
  \bibinfo{volume}{22}, \bibinfo{number}{3} (\bibinfo{year}{1957}),
  \bibinfo{pages}{269--285}.
\newblock


\bibitem[\protect\citeauthoryear{Fitting}{Fitting}{1996}]%
        {fittingbook}
\bibfield{author}{\bibinfo{person}{Melvin Fitting}.}
  \bibinfo{year}{1996}\natexlab{}.
\newblock \bibinfo{booktitle}{\emph{First-order Logic and Automated Theorem
  Proving}}.
\newblock \bibinfo{publisher}{Springer}.
\newblock


\bibitem[\protect\citeauthoryear{Gandy}{Gandy}{1974}]%
        {gandy}
\bibfield{author}{\bibinfo{person}{R.~O. Gandy}.}
  \bibinfo{year}{1974}\natexlab{}.
\newblock \showarticletitle{Set-theoretic functions for elementary syntax}.
\newblock In \bibinfo{booktitle}{\emph{Proceedings of Symposia in Pure
  Mathematics, 13, Part II}}, \bibfield{editor}{\bibinfo{person}{Thomas Jech}}
  (Ed.). \bibinfo{publisher}{American Mathematical Society},
  \bibinfo{pages}{103--126}.
\newblock


\bibitem[\protect\citeauthoryear{Gentzen}{Gentzen}{1935}]%
        {gentzen1935}
\bibfield{author}{\bibinfo{person}{Gerhard Gentzen}.}
  \bibinfo{year}{1935}\natexlab{}.
\newblock \showarticletitle{Untersuchungen {\"u}ber das logische
  Schlie{\ss}en.}
\newblock \bibinfo{journal}{\emph{Mathematische zeitschrift}}
  \bibinfo{volume}{39}, \bibinfo{number}{1} (\bibinfo{year}{1935}),
  \bibinfo{pages}{176--210, 405--431}.
\newblock


\bibitem[\protect\citeauthoryear{Gibbons}{Gibbons}{2016}]%
        {ringads}
\bibfield{author}{\bibinfo{person}{Jeremy Gibbons}.}
  \bibinfo{year}{2016}\natexlab{}.
\newblock \showarticletitle{{Comprehending Ringads} - For {Phil Wadler}, on the
  Occasion of his 60th Birthday}. In \bibinfo{booktitle}{\emph{A List of
  Successes That Can Change the World - Essays Dedicated to {Philip Wadler} on
  the Occasion of His 60th Birthday}}.
\newblock


\bibitem[\protect\citeauthoryear{Gibbons, Henglein, Hinze, and Wu}{Gibbons
  et~al\mbox{.}}{2018}]%
        {gibbonshenglein}
\bibfield{author}{\bibinfo{person}{Jeremy Gibbons}, \bibinfo{person}{Fritz
  Henglein}, \bibinfo{person}{Ralf Hinze}, {and} \bibinfo{person}{Nicolas Wu}.}
  \bibinfo{year}{2018}\natexlab{}.
\newblock \showarticletitle{Relational algebra by way of adjunctions}.
\newblock \bibinfo{journal}{\emph{{PACMPL}}} \bibinfo{volume}{2},
  \bibinfo{number}{{ICFP}} (\bibinfo{year}{2018}).
\newblock


\bibitem[\protect\citeauthoryear{Grust, Rittinger, and Schreiber}{Grust
  et~al\mbox{.}}{2010}]%
        {grust_aval}
\bibfield{author}{\bibinfo{person}{Torsten Grust}, \bibinfo{person}{Jan
  Rittinger}, {and} \bibinfo{person}{Tom Schreiber}.}
  \bibinfo{year}{2010}\natexlab{}.
\newblock \showarticletitle{Avalanche-Safe {LINQ} Compilation}.
\newblock \bibinfo{journal}{\emph{PVLDB}} \bibinfo{volume}{3},
  \bibinfo{number}{1–2} (\bibinfo{year}{2010}), \bibinfo{pages}{162--–172}.
\newblock


\bibitem[\protect\citeauthoryear{Halevy}{Halevy}{2001}]%
        {alonviews}
\bibfield{author}{\bibinfo{person}{Alon~Y. Halevy}.}
  \bibinfo{year}{2001}\natexlab{}.
\newblock \showarticletitle{Answering queries using views: A survey}.
\newblock \bibinfo{journal}{\emph{{VLDB Journal}}} \bibinfo{volume}{10},
  \bibinfo{number}{4} (\bibinfo{year}{2001}), \bibinfo{pages}{270--294}.
\newblock


\bibitem[\protect\citeauthoryear{Hoder, Kov{\'a}cs, and Voronkov}{Hoder
  et~al\mbox{.}}{2010}]%
        {vampireinterp1}
\bibfield{author}{\bibinfo{person}{Kry{\v{s}}tof Hoder}, \bibinfo{person}{Laura
  Kov{\'a}cs}, {and} \bibinfo{person}{Andrei Voronkov}.}
  \bibinfo{year}{2010}\natexlab{}.
\newblock \bibinfo{booktitle}{\emph{Interpolation and Symbol Elimination in
  {Vampire}}}.
\newblock


\bibitem[\protect\citeauthoryear{Hodges}{Hodges}{1993}]%
        {hodgesbook}
\bibfield{author}{\bibinfo{person}{Wilfrid Hodges}.}
  \bibinfo{year}{1993}\natexlab{}.
\newblock \bibinfo{booktitle}{\emph{Model Theory}}.
\newblock \bibinfo{publisher}{Cambridge University Press}.
\newblock


\bibitem[\protect\citeauthoryear{Hodges, Hodkinson, and Macpherson}{Hodges
  et~al\mbox{.}}{1990}]%
        {hodgesdugald}
\bibfield{author}{\bibinfo{person}{Wilfrid Hodges}, \bibinfo{person}{I.M.
  Hodkinson}, {and} \bibinfo{person}{Dugald Macpherson}.}
  \bibinfo{year}{1990}\natexlab{}.
\newblock \showarticletitle{Omega-categoricity, relative categoricity and
  coordinatisation}.
\newblock \bibinfo{journal}{\emph{Annals of Pure and Applied Logic}}
  \bibinfo{volume}{46}, \bibinfo{number}{2} (\bibinfo{year}{1990}),
  \bibinfo{pages}{169 -- 199}.
\newblock


\bibitem[\protect\citeauthoryear{Hu and D’Antoni}{Hu and D’Antoni}{2017}]%
        {invertloris}
\bibfield{author}{\bibinfo{person}{Qinheping Hu} {and} \bibinfo{person}{Loris
  D’Antoni}.} \bibinfo{year}{2017}\natexlab{}.
\newblock \showarticletitle{Automatic Program Inversion Using Symbolic
  Transducers}. In \bibinfo{booktitle}{\emph{PLDI}}.
\newblock


\bibitem[\protect\citeauthoryear{Jacobs}{Jacobs}{2001}]%
        {jacobsbook}
\bibfield{author}{\bibinfo{person}{Bart Jacobs}.}
  \bibinfo{year}{2001}\natexlab{}.
\newblock \bibinfo{booktitle}{\emph{{Categorical Logic and Type Theory}}}.
\newblock \bibinfo{publisher}{Elsevier}.
\newblock


\bibitem[\protect\citeauthoryear{Jensen}{Jensen}{1972}]%
        {jensen}
\bibfield{author}{\bibinfo{person}{R.~B. Jensen}.}
  \bibinfo{year}{1972}\natexlab{}.
\newblock \showarticletitle{The fine structure of the constructible hierarchy,
  with a section by Jack Silver}.
\newblock \bibinfo{journal}{\emph{Annals of Mathematical Logic}}
  \bibinfo{volume}{4} (\bibinfo{year}{1972}), \bibinfo{pages}{229–308}.
\newblock


\bibitem[\protect\citeauthoryear{Koch}{Koch}{2006}]%
        {koch}
\bibfield{author}{\bibinfo{person}{Christoph Koch}.}
  \bibinfo{year}{2006}\natexlab{}.
\newblock \showarticletitle{{On the Complexity of Non-recursive XQuery and
  Functional Query Languages on Complex Values}}.
\newblock \bibinfo{journal}{\emph{ACM TODS}} \bibinfo{volume}{{\bf 31}},
  \bibinfo{number}{4} (\bibinfo{year}{2006}), \bibinfo{pages}{1215--1256}.
\newblock


\bibitem[\protect\citeauthoryear{Koch, Lupei, and Tannen}{Koch
  et~al\mbox{.}}{2016}]%
        {kochlupei}
\bibfield{author}{\bibinfo{person}{Christoph Koch}, \bibinfo{person}{Daniel
  Lupei}, {and} \bibinfo{person}{Val Tannen}.} \bibinfo{year}{2016}\natexlab{}.
\newblock \showarticletitle{Incremental View Maintenance For Collection
  Programming}. In \bibinfo{booktitle}{\emph{{PODS}}}.
\newblock


\bibitem[\protect\citeauthoryear{Kolaitis}{Kolaitis}{1990}]%
        {kolaitisimpdef}
\bibfield{author}{\bibinfo{person}{Phokion~G. Kolaitis}.}
  \bibinfo{year}{1990}\natexlab{}.
\newblock \showarticletitle{Implicit Definability on Finite Structures and
  Unambiguous Computations}. In \bibinfo{booktitle}{\emph{LICS}}.
\newblock


\bibitem[\protect\citeauthoryear{Lenzerini}{Lenzerini}{2002}]%
        {dataint}
\bibfield{author}{\bibinfo{person}{Maurizio Lenzerini}.}
  \bibinfo{year}{2002}\natexlab{}.
\newblock \showarticletitle{Data Integration: A Theoretical Perspective}. In
  \bibinfo{booktitle}{\emph{PODS}}.
\newblock


\bibitem[\protect\citeauthoryear{Makkai}{Makkai}{1964}]%
        {makkai}
\bibfield{author}{\bibinfo{person}{M. Makkai}.}
  \bibinfo{year}{1964}\natexlab{}.
\newblock \showarticletitle{On a generalization of a theorem of {E. W. Beth}}.
\newblock \bibinfo{journal}{\emph{{Acta Mathematica Academiae Scientiarum
  Hungaricae}}}  \bibinfo{volume}{15} (\bibinfo{year}{1964}),
  \bibinfo{pages}{227–235}.
\newblock


\bibitem[\protect\citeauthoryear{McMillan}{McMillan}{2003}]%
        {mcmillaninterp1}
\bibfield{author}{\bibinfo{person}{K.L. McMillan}.}
  \bibinfo{year}{2003}\natexlab{}.
\newblock \showarticletitle{Interpolation and SAT-Based Model Checking}.
\newblock In \bibinfo{booktitle}{\emph{CAV}}.
\newblock


\bibitem[\protect\citeauthoryear{Meijer, Beckman, and Bierman}{Meijer
  et~al\mbox{.}}{2006}]%
        {linq}
\bibfield{author}{\bibinfo{person}{Erik Meijer}, \bibinfo{person}{Brian
  Beckman}, {and} \bibinfo{person}{Gavin Bierman}.}
  \bibinfo{year}{2006}\natexlab{}.
\newblock \showarticletitle{{LINQ}: Reconciling Object, Relations and {XML} in
  the {.NET} Framework}. In \bibinfo{booktitle}{\emph{SIGMOD}}.
\newblock


\bibitem[\protect\citeauthoryear{Melnik, Gubarev, Long, Romer, Shivakumar,
  Tolton, and Vassilakis}{Melnik et~al\mbox{.}}{2010}]%
        {dremel}
\bibfield{author}{\bibinfo{person}{Sergey Melnik}, \bibinfo{person}{Andrey
  Gubarev}, \bibinfo{person}{Jing~Jing Long}, \bibinfo{person}{Geoffrey Romer},
  \bibinfo{person}{Shiva Shivakumar}, \bibinfo{person}{Matt Tolton}, {and}
  \bibinfo{person}{Theo Vassilakis}.} \bibinfo{year}{2010}\natexlab{}.
\newblock \showarticletitle{{Dremel: Interactive Analysis of Web-Scale
  Datasets}}.
\newblock \bibinfo{journal}{\emph{{PVLDB}}} \bibinfo{volume}{3},
  \bibinfo{number}{1-2} (\bibinfo{year}{2010}), \bibinfo{pages}{330--339}.
\newblock


\bibitem[\protect\citeauthoryear{Mostowski}{Mostowski}{1949}]%
        {mostowski}
\bibfield{author}{\bibinfo{person}{Andrzej Mostowski}.}
  \bibinfo{year}{1949}\natexlab{}.
\newblock \showarticletitle{An undecidable arithmetical statement}.
\newblock \bibinfo{journal}{\emph{Fundamenta Mathematicae}}
  \bibinfo{volume}{36}, \bibinfo{number}{1} (\bibinfo{year}{1949}),
  \bibinfo{pages}{143–164}.
\newblock


\bibitem[\protect\citeauthoryear{Nash, Segoufin, and Vianu}{Nash
  et~al\mbox{.}}{2010}]%
        {NSV}
\bibfield{author}{\bibinfo{person}{Alan Nash}, \bibinfo{person}{Luc Segoufin},
  {and} \bibinfo{person}{Victor Vianu}.} \bibinfo{year}{2010}\natexlab{}.
\newblock \showarticletitle{Views and queries: Determinacy and rewriting}.
\newblock \bibinfo{journal}{\emph{ACM TODS}} \bibinfo{volume}{35},
  \bibinfo{number}{3} (\bibinfo{year}{2010}).
\newblock


\bibitem[\protect\citeauthoryear{Otto}{Otto}{2000}]%
        {otto}
\bibfield{author}{\bibinfo{person}{Martin Otto}.}
  \bibinfo{year}{2000}\natexlab{}.
\newblock \showarticletitle{An interpolation theorem}.
\newblock \bibinfo{journal}{\emph{Bulletin of Symbolic Logic}}
  \bibinfo{volume}{6}, \bibinfo{number}{4} (\bibinfo{year}{2000}),
  \bibinfo{pages}{447--462}.
\newblock


\bibitem[\protect\citeauthoryear{Paredaens and Van~Gucht}{Paredaens and
  Van~Gucht}{1992}]%
        {conservativity}
\bibfield{author}{\bibinfo{person}{Jan Paredaens} {and} \bibinfo{person}{Dirk
  Van~Gucht}.} \bibinfo{year}{1992}\natexlab{}.
\newblock \showarticletitle{Converting Nested Algebra Expressions into Flat
  Algebra Expressions}.
\newblock \bibinfo{journal}{\emph{ACM TODS}} \bibinfo{volume}{17},
  \bibinfo{number}{1} (\bibinfo{year}{1992}), \bibinfo{pages}{65--93}.
\newblock


\bibitem[\protect\citeauthoryear{Sazonov}{Sazonov}{1985}]%
        {sazonovcoll}
\bibfield{author}{\bibinfo{person}{Vladimir~Yu. Sazonov}.}
  \bibinfo{year}{1985}\natexlab{}.
\newblock \showarticletitle{Collection principle and existential quantifier}.
\newblock \bibinfo{journal}{\emph{Vychislitel'nye sistemy}}
  \bibinfo{volume}{107} (\bibinfo{year}{1985}), \bibinfo{pages}{30--39}.
\newblock


\bibitem[\protect\citeauthoryear{Segoufin and Vianu}{Segoufin and
  Vianu}{2005}]%
        {SVconf}
\bibfield{author}{\bibinfo{person}{Luc Segoufin} {and} \bibinfo{person}{Victor
  Vianu}.} \bibinfo{year}{2005}\natexlab{}.
\newblock \showarticletitle{Views and queries: determinacy and rewriting}. In
  \bibinfo{booktitle}{\emph{PODS}}.
\newblock


\bibitem[\protect\citeauthoryear{S{\o}rensen and Urzyczyn}{S{\o}rensen and
  Urzyczyn}{2006}]%
        {su06book}
\bibfield{author}{\bibinfo{person}{M.~H. S{\o}rensen} {and} \bibinfo{person}{P.
  Urzyczyn}.} \bibinfo{year}{2006}\natexlab{}.
\newblock \bibinfo{booktitle}{\emph{{Lectures on the Curry-Howard
  Isomorphism}}}.
\newblock \bibinfo{publisher}{Elsevier}.
\newblock


\bibitem[\protect\citeauthoryear{Srivastava, Gulwani, Chaudhuri, and
  Foster}{Srivastava et~al\mbox{.}}{2011}]%
        {invertswarat}
\bibfield{author}{\bibinfo{person}{Saurabh Srivastava}, \bibinfo{person}{Sumit
  Gulwani}, \bibinfo{person}{Swarat Chaudhuri}, {and}
  \bibinfo{person}{Jeffrey~S. Foster}.} \bibinfo{year}{2011}\natexlab{}.
\newblock \showarticletitle{{Path-Based Inductive Synthesis for Program
  Inversion}}. In \bibinfo{booktitle}{\emph{{PLDI}}}.
\newblock


\bibitem[\protect\citeauthoryear{Suciu}{Suciu}{1995}]%
        {suciuthesis}
\bibfield{author}{\bibinfo{person}{Dan Suciu}.}
  \bibinfo{year}{1995}\natexlab{}.
\newblock \emph{\bibinfo{title}{Parallel Programming Languages for
  Collections}}.
\newblock \bibinfo{thesistype}{Ph.D. Dissertation}. \bibinfo{school}{Univ.
  Pennsylvania}.
\newblock


\bibitem[\protect\citeauthoryear{Toman and Weddell}{Toman and Weddell}{2011}]%
        {tomanweddell}
\bibfield{author}{\bibinfo{person}{David Toman} {and} \bibinfo{person}{Grant
  Weddell}.} \bibinfo{year}{2011}\natexlab{}.
\newblock \bibinfo{booktitle}{\emph{Fundamentals of Physical Design and Query
  Compilation}}.
\newblock \bibinfo{publisher}{Morgan Claypool}.
\newblock


\bibitem[\protect\citeauthoryear{Troelstra and Schwichtenberg}{Troelstra and
  Schwichtenberg}{2000}]%
        {ts00book}
\bibfield{author}{\bibinfo{person}{A.~S. Troelstra} {and} \bibinfo{person}{H.
  Schwichtenberg}.} \bibinfo{year}{2000}\natexlab{}.
\newblock \bibinfo{booktitle}{\emph{Basic Proof Theory}}.
\newblock \bibinfo{publisher}{Cambridge University Press}.
\newblock


\bibitem[\protect\citeauthoryear{Ulrich}{Ulrich}{2019}]%
        {ulrichphd}
\bibfield{author}{\bibinfo{person}{Alexander Ulrich}.}
  \bibinfo{year}{2019}\natexlab{}.
\newblock \emph{\bibinfo{title}{Query Flattening and the Nested Data
  Parallelism Paradigm}}.
\newblock \bibinfo{thesistype}{Ph.D. Dissertation}. \bibinfo{school}{University
  of T{\"{u}}bingen, Germany}.
\newblock
\urldef\tempurl%
\url{https://publikationen.uni-tuebingen.de/xmlui/handle/10900/87698/}
\showURL{%
\tempurl}


\bibitem[\protect\citeauthoryear{{Van den Bussche}}{{Van den Bussche}}{2001}]%
        {simulation}
\bibfield{author}{\bibinfo{person}{Jan {Van den Bussche}}.}
  \bibinfo{year}{2001}\natexlab{}.
\newblock \showarticletitle{{Simulation of the Nested Relational Algebra by the
  Flat Relational Algebra, with an Application to the Complexity of Evaluating
  Powerset Algebra Expressions}}.
\newblock \bibinfo{journal}{\emph{Theoretical Computer Science}}
  \bibinfo{volume}{{\bf 254}}, \bibinfo{number}{1--2} (\bibinfo{year}{2001}),
  \bibinfo{pages}{363--377}.
\newblock


\bibitem[\protect\citeauthoryear{Wernhard}{Wernhard}{2018}]%
        {wernhard}
\bibfield{author}{\bibinfo{person}{Christoph Wernhard}.}
  \bibinfo{year}{2018}\natexlab{}.
\newblock \showarticletitle{Craig Interpolation and Access Interpolation with
  Clausal First-Order Tableaux}.
\newblock \bibinfo{journal}{\emph{CoRR}}  \bibinfo{volume}{abs/1802.04982}
  (\bibinfo{year}{2018}).
\newblock


\bibitem[\protect\citeauthoryear{Wong}{Wong}{1994}]%
        {limsoonthesis}
\bibfield{author}{\bibinfo{person}{Limsoon Wong}.}
  \bibinfo{year}{1994}\natexlab{}.
\newblock \emph{\bibinfo{title}{Querying Nested Collections}}.
\newblock \bibinfo{thesistype}{Ph.D. Dissertation}. \bibinfo{school}{Univ.
  Pennsylvania}.
\newblock


\bibitem[\protect\citeauthoryear{Wong}{Wong}{1996}]%
        {wongconservativity}
\bibfield{author}{\bibinfo{person}{Limsoon Wong}.}
  \bibinfo{year}{1996}\natexlab{}.
\newblock \showarticletitle{Normal Forms and Conservative Extension Properties
  for Query Languages over Collection Types}.
\newblock \bibinfo{journal}{\emph{J. Comput. Syst. Sci.}} \bibinfo{volume}{52},
  \bibinfo{number}{3} (\bibinfo{year}{1996}), \bibinfo{pages}{495--505}.
\newblock


\end{thebibliography}

\section*{Supplementary materials}
\appendix
\section{Proofs for Section 3}
\subsection{Proof that we can obtain $\nrc$ expressions that verify $\deltazero$ formulas}
Recall that in the body of the paper, we claimed the following statement, concerning
the equivalence of $\nrc$ expressions of Boolean type and $\deltazero$ formulas:

\medskip

There is a polynomial time
function  taking a $\deltazero$ formula $\phi(\vec x)$ and producing
 an $\nrc$ expression $\verify_\phi(\vec x)$, where the expression  takes as input
$\vec x$ and returns true if and only if $\phi$ holds.

\medskip

We refer to this as the ``Verification Proposition'' later on in these supplementary
materials.

\begin{proof}
First, one should note that every term in the logic can be translated to
a suitable $\nrc$ expression of the same sort. For example, a variable in
the logic corresponds to a variable in $\nrc$. 

We prove the proposition by induction over the formula $\phi(\vec x)$.
\begin{compactitem}
\item If $\phi(\vec x)$ is an equality $t = t'$ or a membership $t \in t'$,
it is straightforward to write out $\nrc$ expressions that verify them by
simultaneous induction on the type.
For equality, the expression verifies two containments, with a containment
$t \subseteq t'$ verified as $\bigcup \{x \in t | E'(x ,t') \}$, where
$E'(x,t')$ is the expression obtained for membership inductively.

\item If $\phi(\vec x)$ is a disjunction $\phi_1(\vec x) \vee \phi_2(\vec x)$,
we take $\verify_\phi(\vec x) = \verify_{\phi_1} \cup \verify_{\phi_2}$.
We proceed similarly for disjunction thanks to $\cap$.
\item If $\phi(\vec x)$ is a negation, we use the definability of negation
in $\nrc$.
\item If $\phi(\vec x)$ begins with a bounded existential quantification $\exists z \in y \; \psi(\vec x, y, z)$,
we simply set $\verify_\phi(\vec x, y) = \bigcup \{ \verify_{\psi(\vec x, y, z)} \mid z \in y \}$.
Universal quantification is then treated similarly by using  negation in $\nrc$.
\end{compactitem}
\end{proof}

Note that the converse (without the polynomial time bound) also holds; this will follow
from the more general result on moving from $\nrc$ to interpretations
that is proven later in the supplementary materials.



\section{Proofs for Section 4: properties  of the proof system, details of the
 synthesis results}
\subsection{Strength of the proof system} 

In the body of the paper we claimed
that although our proof system does not derive
every classically valid $\deltazero$ sequent, we can
show that it derives all sequents of the shape we consider that
are constructively derivable  in the
sense of intuitionistic logic.
In this subsection we  present variants of prior intuitionistic calculi formally,
and detail the argument for their equivalence with our system.

\begin{figure}[h]
\begin{mathpar}
\inferrule* [left=\rnax] { }{\Gamma, \; \phi \vdash \phi}
\\
\inferrule* [left=\rnandL] {\Gamma, \; \phi_1, \; \phi_2 \vdash \psi}{\Gamma,\; \phi_1 \wedge \phi_2 \vdash \psi}
\and
\inferrule* [left=\rnandR] {\Gamma \vdash \phi \and
\Gamma \vdash \psi
 }{\Gamma \vdash \phi \wedge \psi}
\and
\inferrule* [left=\rntop]{ }{\Gamma \vdash \top}
\\
\inferrule* [left=\rnbot] { }{\Gamma,  \bot \vdash \phi}
\and
\inferrule* [left=\rnorL] {\Gamma, \; \phi_1 \vdash \psi \qquad
\phi_2, \Gamma \vdash \psi}{\Gamma, \; \phi_1 \vee \phi_2 \vdash \psi}
\and
\inferrule* [left=\rnorR] {\Gamma \vdash \phi_i \and i \in \{1,2\}}{\Gamma \vdash \phi_1 \vee \phi_2}
\\
\inferrule* [left=\rnimplL] {\Gamma \vdash \phi \and \Gamma, \; \psi \vdash \theta}{\Gamma, \; \phi \Rightarrow \psi \vdash \theta}
\and
\inferrule* [left=\rnimplR] {\Gamma, \; \phi \vdash \psi}{\Gamma \vdash \phi \Rightarrow \psi}
\\
\inferrule* [left=\rnallL] {\Gamma, \; \phi[t/y] \vdash \psi}{\Gamma, \; \forall y ~~ \phi \vdash \psi}
\and
\inferrule* [left=\rnallR] {\Gamma \vdash \phi \and z \notin \freevars(\Gamma))}{\Gamma \vdash \forall z \; \phi}
\\
\inferrule* [left=\rnexL] {\Gamma, \; \phi \vdash \psi \qquad x \notin \freevars(\Gamma, \psi)}{\Gamma, \; \exists x ~~\phi \vdash \psi}
\and
\inferrule* [left=\rnexR] {\Gamma \vdash \phi[t/x]}{\Gamma \vdash \exists x~~ \phi}
\\
\inferrule* [left=\rneqL] {\Gamma[s/x,t/y] \vdash \phi[s/x,t/y]}{\Gamma[t/x,s/y], \; t = s \vdash \phi[t/x,s/y]}
\and
\inferrule* [left=\rneqR] { }{\Gamma \vdash t = t}
\\
\inferrule* [left=\rnpaireta] {\Gamma[\<x_1,x_2\>/x] \vdash \phi[\<x_1,x_2\>/x] \and x_1, x_2 \notin \freevars(\Gamma,\phi)}{\Gamma \vdash \phi}
\and
\inferrule* [left=\rnpairbeta]{\Gamma[t_i/x] \vdash \phi[t_i/x] \and i \in \{1,2\}}{\Gamma[\pi_i(\<t_1,t_2\>)/x] \vdash \phi[\pi_i(\<t_1,t_2\>)/x]}
\and
\inferrule* [left=\rnuniteta]{\Gamma[()/x] \vdash \phi[()/x]}{\Gamma \vdash \phi}
\\
\end{mathpar}

\caption{The intuitionistic sequent calculus (LJ) for multi-sorted first-order logic with equality and pairs}
\label{fig:LJsequentcalculus}
\end{figure}


Let us first recall the syntax of multi-sorted first-order logic,
with equality at every sort and a predicate $- \in_T -$
for every sort $T$ representing \emph{membership}.
$$\varphi, \; \psi \bnfeq t \in_T u \bnfalt t =_T u \bnfalt \top \bnfalt \bot \bnfalt \phi \wedge \psi \bnfalt \phi \vee \psi \bnfalt \phi \Rightarrow \psi \bnfalt \forall x^T \; \phi \bnfalt \exists x^T \; \phi$$
We will deal with the case where the terms are built up using Ur-element constants,  the unit constant,
the pairing function
and the projection functions.
The intuitionistic sequent calculus we adopt for first-order logic with equality, projection, and pairing
is shown in Figure~\ref{fig:LJsequentcalculus}, with the structural rules (weakening and contraction) omitted.
It is a straightforward extension of the textbook definition of the sequent calculus
LJ for intuitionistic first-order logic
(see e.g.~\cite[Sections 7.2 and 9.3]{su06book} and~\cite[Chapter 3]{ts00book}) due to Gentzen~\cite{gentzen1935}
to accommodate our typing discipline and additional rules concerning equalities, projection and pairing.
The main technical distinction between LJ and the sequent calculus for classical logic LK is that
there is a single conclusion formula on the right, rather than a list of formulas.
This prevents one from deriving the law of excluded middle $\vdash \phi \vee \neg \phi$ for arbitrary $\phi$ in LJ. 
Note that this does not imply that the calculus is incomplete for  (translations of) the restricted sequents that
we deal with in our calculus.

The extensions of LJ to accommodate typed terms, equality, and the
projection and pairing functions are straightforward. Although
we are not aware of a source describing exactly  the  proof system above, \cite[Chapter 4]{jacobsbook}
describes an equivalent system based on natural deduction
and~\cite[Section 4.7]{ts00book} extends LJ with rules for equality without types.

In this section, we define a translation of the sequents $\Theta; \; \Gamma \vdash \phi$
of our restricted proof system into sequents $\widetilde\Theta, \; \widetilde\Gamma \vdash \widetilde\phi$ of 
the calculus displayed in Figure~\ref{fig:LJsequentcalculus},
which we refer to as LJ from now on.

As is customary for two-sided sequent calculi, rules introducing logical connectives
can be split into \emph{left}(-hand side) and \emph{right}(-hand side) rules. We make 
this distinction
in our naming conventions, using $\textsc{L}$ and $\textsc{R}$ in rule names to indicate left and right rules.
Informally speaking, a rule is left if the right-hand side formula
stays the same in the premises and the conclusion and the corresponding connective occurs
in the left-hand side of the conclusion. Right rules can be similarly characterized.
Some rules are neither right nor left. 
For LJ, these would be the axiom rule $\textsc{AX}$ and the rules $\times_\eta$, $\times_\beta$ and $\unit_\eta$.

\myeat{We will first discuss the relationship of our restricted syntax for $\deltazero$ formula
(which does not use equality predicates for sorts other than $\ur$ and restrict uses of the atomic
$\in_T$ predicate) with the more standard notion \michael{the more standard
notion of what? If you mean $\deltazero$ 
formulas, then I am not sure which is more standard. Even then I can not
see the relation between this and what comes next},
  and explain why we do not require the axiom of extensionality.
We will then argue that proofs of $\Theta; \; \Gamma \vdash \psi$ in our restricted
system to proofs of $\widetilde\Theta, \; \widetilde\Gamma \vdash \widetilde\psi$ in LJ, and that, conversely
proofs of $\widetilde\Theta, \; \widetilde\Gamma \vdash \widetilde\psi$ may be turned into proofs of $\Theta; \; \Gamma \vdash \phi$
in our restricted system.}

\myparagraph{Translation to LJ sequents}
We will need to perform some translations from the membership contexts and $\deltazero$ formulas used
in our context to the multi-sorted first-order formulas used in LJ.
$\deltazero$ formulas $\phi$ as defined in the paper can be regarded as 
a particular case of general formulas with an abbreviated syntax. Formally, for
each $\deltazero$ formula $\phi$ we have a corresponding first-order formula $\phi^*$ defined
in the usual way
$$
\begin{array}{lcl !\qquad lcl}
(t =_\ur u)^* &\eqdef& t =_\ur u &
(t \neq_\ur u)^* &\eqdef& t =_\ur u \Rightarrow \bot
\\
\top^* &\eqdef& \top &
\bot^* &\eqdef& \bot 
\\
(\phi \wedge \psi)^* &\eqdef& \phi^* \wedge \psi^* &
(\phi \vee \psi)^* &\eqdef& \phi^* \vee \psi^*
\\ 
(\forall x \in_T t ~~\phi)^* &\eqdef& \forall x^T~~ (x \in_T t \Rightarrow \phi^*) &
(\exists x \in_T t ~~\phi)^* &\eqdef& \exists x^T~~ (x \in_T t \wedge \phi^*)
\\ 
\end{array}
$$
Recall that sequents in our restricted system are of the shape
$\Theta; \; \Gamma \vdash \psi$
where $\Theta$ is a multiset of pairs of formulas $t \in_T u$,
$\Gamma$ a list of $\deltazero$ formulas and $\psi$ a special
right-hand side formula of shape either $t \in_T u$, $t \subseteq_T u$ or $t =_T u$.
Given such contexts, we write $\widetilde{\Gamma}$ for the multiset of formulas
$\{\phi^* \mid \phi \in \Gamma\}$ and $\widetilde{\Theta}$ for the multiset $\{ t \in_T u \mid  (t \in_T u) \in \Theta\}$.
As for right-hand side formulas $\psi$, we define the notation $\widetilde{\psi}$ by recursion on the type of the main connective of $\psi$ as follows:
$$
\begin{array}{rcl !\qquad rcl}
t \macroin_T u &\eqdef& \exists z'~(z' \in u \wedge t \macroeq_T z')
&
t \macrosubseteq_T u &\eqdef& \forall z ~~(z \in_T t \Rightarrow t \macroin_T u)
\\
t \macroeq_{\sett(T)} u &\eqdef& {t \macrosubseteq_T u} ~~\wedge ~~{u \macrosubseteq_T t}&
t \macroeq_\unit u &\eqdef& \top \\
\multicolumn{6}{c}{t \macroeq_{T_1 \times T_2} u ~~\eqdef~~ {\pi_1(t) \macroeq_{T_1} \pi_1(u)}~~ \wedge~~ {\pi_2(t) \macroeq_{T_2} \pi_2(u)} \qquad t \macroeq_\ursort u ~~\eqdef~~ t =_\ur u}\\
\end{array}
$$



\myparagraph{Translating proofs to LJ}
We are now ready to state the first direction concerning the equivalence between LJ and our proof system.
\begin{lemma}
\label{lem:tolj}
If $\Theta; \; \Gamma \vdash \phi$ is derivable in our restricted system, then
LJ derives $\widetilde\Theta, \; \widetilde\Gamma \vdash \widetilde\phi$.
\end{lemma}

Towards a proof of Lemma~\ref{lem:tolj}, first notice that for every rule
$$
\dfrac{\Theta; \; \Gamma \vdash \psi \qquad \ldots}{\Theta'; \; \Gamma' \vdash \psi'}
$$
of our restricted system, the rule
$$
\dfrac{\widetilde\Theta, \; \widetilde\Gamma \vdash \widetilde\psi \qquad \ldots}{\widetilde\Theta', \; \widetilde\Gamma' \vdash \widetilde\psi'}
$$
is easily seen to be admissible in LJ, save for one:
$$
\dfrac{\Theta, \; t \in_\ursort  z ; \; \Gamma \vdash u \in_\ursort z \qquad z \notin \freevars(\Theta, \Gamma, t, u)}{\Gamma \vdash t =_\ur u}$$
It is helpful to treat the sequents of the type $\Theta, \; t \in_\ur z; \; \Gamma \vdash u \in_\ur z$ with $z \notin \freevars(\Theta, \Gamma, t, u)$ as a special case.

\begin{proposition}
\label{prop:impredureq-tolj}
For every contexts $\Theta$, $\Gamma$ and terms $t$ and $u$ of type $\ur$ whose free variables
do not include $z$, if the sequent
$\Theta, \; t \in_\ursort z; \; \Gamma \vdash u \in_\ursort z$
is derivable in the restricted system, then LJ derives $\widetilde\Theta, \; \widetilde\Gamma \vdash t =_\ur u$.
\end{proposition}
\begin{proof}
The proof goes by induction on the proof in the restricted system.
For most cases, the induction hypothesis is used in a very simple way.
We focus on one representative subcase.
\begin{itemize}
\item If the last rule applied is a $\forall$ rule, with $\Gamma = \Gamma', \; \forall x \in_T y ~~\phi$
$$\dfrac{\Theta,\; t \in_\ur z; \; \Gamma', \; \phi[v/x] \vdash u \in_\ursort z}{\Theta, \; t \in_\ur z; \; \Gamma', \; \forall x \in_T y ~~\phi \vdash u \in_\ur z}$$
then we must have $v \in_T y$ occurring in $\Theta, \; t \in_\ur z$.
By assumption, $z$ does not occur freely in $\Theta$, so we have necessarily that $v$ does not have $z$ as a free variable. Therefore $z$ does not occur free in either $\Theta$, $\Gamma'$ or $\phi[v/x]$, so we can conclude by applying the inductive hypothesis and using the rule \rnallL~of LJ.
\[
\text{
\AXC{ Induction hypothesis}
\UIC{$\widetilde{\Theta}, \; t \in_\ur z, \; \widetilde{\Gamma'}, \;  \phi^*[v/x] \vdash t =_\ur u$}
\AXC{ }
\UIC{$\widetilde{\Theta}, \; t \in_\ur z, \; \widetilde{\Gamma'} \vdash v \in y$}
\BIC{$\widetilde{\Theta}, \; t \in_\ur z, \; \widetilde{\Gamma'}, \;  v \in y \Rightarrow \phi^*[v/x] \vdash t =_\ur u$}
\UIC{$\widetilde{\Theta}, \; t \in_\ur z, \; \widetilde{\Gamma'}, \;  (\forall x \in y~~ \phi)^* \vdash t =_\ur u$}
\DisplayProof}
\]
\end{itemize}
\end{proof}

\begin{proof}[Proof of Lemma~\ref{lem:tolj}]
The proof goes by induction over the proof of $\Theta; \; \Gamma \vdash \psi$ in the
restricted system. Now that we have proven Proposition~\ref{prop:impredureq-tolj}, all
the cases are straightforward. We only outline a few.
\begin{itemize}
\item If the last rule applied is $=_\ur$-\textsc{R}
$$
\dfrac{\Theta, \; t \in_\ursort  z ; \; \Gamma \vdash u \in_\ursort z \qquad z \notin \freevars(\Theta, \Gamma, t, u)}{\Gamma \vdash t =_\ur u}$$
then we may use the induction hypothesis together with Proposition~\ref{prop:impredureq-tolj}.
\item If the last rule applied is $\in_\sett$-\textsc{R}
$$
\dfrac{\Theta, \; t \in_T u; \; \Gamma \vdash t =_T t'}{\Theta, \; t \in_T u; \; \Gamma \vdash t' \in_T u}
$$
recalling that $t' \macroin_T u$ is defined as $\exists x~~ (x \in_T u \wedge x \macroeq_T t')$, 
we give the following derivation in LJ
\[
\text{
\AXC{ }
\UIC{$\widetilde{\Theta}, \; t \macroin_T u, \; \widetilde\Gamma \vdash t \macroin_T u$}
\AXC{Induction hypothesis }
\UIC{$\widetilde{\Theta}, \; t \macroin_T u, \; \widetilde\Gamma \vdash t \macroeq_T t'$}
\BIC{$\widetilde{\Theta}, \; t \macroin_T u, \; \widetilde\Gamma \vdash t \macroin_T u \wedge t \macroeq_T t'$}
\UIC{$\widetilde{\Theta}, \; t \macroin_T u, \; \widetilde\Gamma \vdash t' \macroin_T u$}
\DisplayProof}
\]
\end{itemize}
\end{proof}



\myparagraph{From LJ to our restricted calculus}
Now, we prove the converse of Lemma~\ref{lem:tolj}.

\begin{lemma}
\label{lem:fromlj}
If the sequent $\widetilde\Theta, \; \widetilde\Gamma \vdash \widetilde\psi$ is derivable in LJ, then
$\Theta; \; \Gamma \vdash \psi$ is derivable in the restricted system.
\end{lemma}

This direction is harder to prove than Lemma~\ref{lem:tolj}, so we will decompose this result in multiple steps:
\begin{enumerate}
\item First, we note that we have the subformula property for LJ:
any formula $\phi$ occurring in a LJ-proof tree is necessarily a subformula of some 
formula occurring at the root, up to substitution of terms.
This allows us to distinguish a special class of formulas which we call \emph{sub$\deltazero$ formulas} and consider
LJ sequents containing only such formulas.
\item For sequents containing only sub$\deltazero$ formulas, we note that if we replace the rules
\rnexL, \rnallL, \rnexR ~and \rnallR~by the bounded variants
\begin{mathpar}
\\
\inferrule* [left=\rnallLBV] {\Gamma, \; t \in y, \; \phi[t/x] \vdash \psi}{\Gamma, \; t \in u, \; \forall x ~(x \in y \Rightarrow \phi) \vdash \psi}
\and
\inferrule* [left=\rnallRBV] {\Gamma, \; z \in y \vdash \phi \and z \notin \freevars(\Gamma))}{\Gamma \vdash \forall z~ (z \in y \Rightarrow \phi)}
\\
\inferrule* [left=\rnexLBV] {\Gamma, \; x \in y, \; \phi \vdash \psi \qquad x \notin \freevars(\Gamma, \psi, y)}{\Gamma, \; \exists x ~(x \in y \wedge \phi) \vdash \psi}
\and
\inferrule* [left=\rnexRBV] {\Gamma, \; t \in y \vdash \phi[t/x]}{\Gamma, \; t \in y \vdash \exists x~ (x \in y \wedge \phi)}
\end{mathpar}
while deriving the same sequents as LJ, while retaining the constraint that the right-hand side formula be neither a conjunct, universal quantification or implication when left-hand side rules are applied. We will call the corresponding system $\ljboundedvarq$.
\item Then, we note that $\ljboundedvarq$ is equivalent to its restriction where left rules
cannot be applied if the
right-hand side formula under consideration is a conjunction, an implication or a universal quantification.
\item Finally, the translation can go by induction on such restricted proofs.
\end{enumerate}

We now go through  these steps in more detail.

\myparagraph{Step 1}
That LJ has the subformula property is obvious from inspection of the proof rules. We identify
the set of subformulas of (translation of) $\deltazero$ formulas, that we call \emph{sub$\deltazero$ formulas}.
\begin{definition}
A sub$\deltazero$ formula is a formula of LJ which is either of the shape $t \in_T u$, $t \in_T u \wedge \phi^*$,
$t \in_T u \Rightarrow \phi^*$ or $\phi^*$, where $\phi$ is a $\deltazero$ formula.
\end{definition}

From now on, we will suppose that all sequents under consideration exclusively contain sub$\deltazero$ formulas.
We call $\ljdeltazero$ the subsystem of LJ where all sequents contain exclusively sub$\deltazero$ formulas.

\myparagraph{Step 2}
Now we need to show that replacing the rules $Q$-\textsc{D} by their counterpart $Q$-\textsc{DBV}, with $Q \in \{ \forall, \exists\}$ and $D \in \{L,R\}$ does not limit LJ's power, \emph{as far as sub$\deltazero$ formulas are concerned}.
It is actually more convenient to do this in multiple steps, which are all proven by straightforward (if lengthy) induction on the proofs. To this end, we consider the following
three set of rules

\begin{mathpar}
\inferrule* [left=\rnimplLB] {\Gamma, \; t \in u, \; \phi \vdash \psi}{\Gamma, \; t \in u, \; t \in u \Rightarrow \phi \vdash \psi}
\and
\inferrule* [left=\rnandRB] {\Gamma, \; t \in u \vdash \phi}{\Gamma, \; t \in u \vdash t\in u \wedge \phi}
\end{mathpar}
\begin{center} \rule{0.5\textwidth}{0.5pt} \end{center}
\begin{mathpar}
\inferrule* [left=\rnallLB] {\Gamma, \; t \in u, \; \phi[t/x] \vdash \psi}{\Gamma, \; t \in u, \; \forall x ~(x \in u \Rightarrow \phi) \vdash \psi}
\and
\inferrule* [left=\rnallRB] {\Gamma, \; z \in u \vdash \phi \and z \notin \freevars(\Gamma,u)}{\Gamma \vdash \forall z~ (z \in u \Rightarrow \phi)}
\\
\inferrule* [left=\rnexLB] {\Gamma, \; x \in u, \; \phi \vdash \psi \qquad x \notin \freevars(\Gamma, \psi, u)}{\Gamma, \; \exists x ~(x \in u \wedge \phi) \vdash \psi}
\and
\inferrule* [left=\rnexRB] {\Gamma, \; t \in u \vdash \phi[t/x]}{\Gamma, \; t \in y \vdash \exists x~ (x \in u \wedge \phi)}
\end{mathpar}
\begin{center} \rule{0.5\textwidth}{0.5pt} \end{center}
\begin{mathpar}
\inferrule* [left=\rnallLBV] {\Gamma, \; t \in y, \; \phi[t/x] \vdash \psi}{\Gamma, \; t \in u, \; \forall x ~(x \in y \Rightarrow \phi) \vdash \psi}
\and
\inferrule* [left=\rnallRBV] {\Gamma, \; z \in y \vdash \phi \and z \notin \freevars(\Gamma))}{\Gamma \vdash \forall z~ (z \in y \Rightarrow \phi)}
\\
\inferrule* [left=\rnexLBV] {\Gamma, \; x \in y, \; \phi \vdash \psi \qquad x \notin \freevars(\Gamma, \psi, y)}{\Gamma, \; \exists x ~(x \in y \wedge \phi) \vdash \psi}
\and
\inferrule* [left=\rnexRBV] {\Gamma, \; t \in y \vdash \phi[t/x]}{\Gamma, \; t \in y \vdash \exists x~ (x \in y \wedge \phi)}
\end{mathpar}
and the corresponding proof systems:
\begin{itemize}
\item We call $\ljboundedconn$ the system $\ljdeltazero$ with the addition of the rules \rnimplLB~and \rnandRB~but omitting the rules \rnimplL~and the following instances of \rnandR
\begin{mathpar}
\inferrule* {\Gamma \vdash t \in u \qquad \Gamma \vdash \phi}{\Gamma \vdash t \in u \wedge \phi}
\end{mathpar}
\item We call $\ljboundedq$ the system $\ljboundedconn$ with the addition of the rules \rnallLB, \rnallRB, \rnexLB~and \rnexRB, but omitting the rules \rnallL, \rnallR, \rnexL~and \rnexR.
\item We call $\ljboundedvarq$ the system LJB2 with the addition of the rules \rnallLBV, \rnallRBV, \rnexLBV~and \rnexRBV, but omitting the rules \rnallLB, \rnallRB, \rnexLB~and \rnexRB.
\end{itemize}

We can now show that all those systems derive the same sequents thanks 
to a series of lemmas 
stating that when moving from $\ljdeltazero$ to $\ljboundedconn$
to $\ljboundedq$ to $\ljboundedvarq$, in each step 
the rules  we have removed 
remain admissible using the rules we have added. The admissibility of each individual rule mentioned in the lemmas
 can be shown by a lengthy induction.
\begin{lemma}
\label{lem:ljboundedconn}
The rules \rnimplL~and \rnandR ~are admissible in $\ljboundedconn$.
\end{lemma}
\begin{proof}
Let us first focus on the admissibility of {\rnandR}.
By induction on the depth of a $\ljboundedconn$ proof of
$$\Gamma \vdash t \in_T u \qquad \text{and} \qquad \Gamma \vdash \psi$$
we want to show that $\Gamma \vdash t \in_T u \wedge \psi$ is derivable in 
$\ljboundedconn$.
 Note that if the first
conjunct is not a formula of the shape $t \in_T u$, we may conclude using an instance of \rnandR~of $\ljboundedconn$,
To this end, we make a case analysis according to the last $\ljboundedconn$ rule applied to derive $\Gamma \vdash t \in_T u$.
As they are many cases, we only outline a few representative ones. 
Most cases are easy because
it cannot be the case that a right-hand side rule of $\ljboundedconn$  may be applied, since $t \in_T u$ is an atomic
formula.
\begin{itemize}
\item If the last rule applied was an axiom, this means that $t \in_T u$ was part of $\Gamma$.
In this case
$$\dfrac{\Gamma \vdash \psi}{\Gamma \vdash t \in_T u \wedge \psi}$$
is an instance of \rnandRB, the designated replacement of \rnandR.
\item If the last rule applied was \rnandL, assuming that $\Gamma = \Gamma', \phi_1 \wedge \phi_2$
$$\dfrac{\Gamma',\; \phi_1, \; \phi_2 \vdash t \in_T u}{\Gamma', \; \phi_1 \wedge \phi_2 \vdash t \in_T u}$$
then the induction hypothesis gives us a proof of $\Gamma', \; \phi_1, \; \phi_2 \vdash t \in_T u \wedge \psi$,
so we may build the tree
$$
\dfrac{
\dfrac{\text{Induction hypothesis}}{
\Gamma', \; \phi_1, \; \phi_2 \vdash t \in_T u \wedge \psi}}
{\Gamma \vdash \psi}$$
by applying the rule \rnandL.
\end{itemize}

The admissibility of \rnimplL ~ is handled similarly, noticing that, since we are dealing
with sub$\deltazero$ formulas, the antecedent of an implication in such a rule is also an
atomic formula $t \in_T u$. 
\end{proof}
\begin{corollary}
\label{cor:ljboundedconn}
$\ljdeltazero$ and $\ljboundedconn$ derive the same sequents.
\end{corollary}
\begin{proof}
Thanks to Lemma~\ref{lem:ljboundedconn}, it is then obvious that all the rules of $\ljdeltazero$
are admissible in $\ljboundedconn$, so every sequent derivable in $\ljdeltazero$ is derivable in $\ljboundedconn$.
The converse is obvious.
\end{proof}

\begin{lemma}
\label{lem:ljboundedq}
The rules \rnallL, \rnallR, \rnexL~and \rnexR~are admissible in $\ljboundedq$.
\end{lemma}
\begin{proof}
Let us focus on \rnallL. We assume that we have a $\ljboundedq$ derivation of
$$\Gamma, \; t \in_T u \Rightarrow \phi[t/x] \vdash \psi$$
and we show, by induction on its depth, that we way obtain a $\ljboundedq$ derivation
of $\Gamma, \; \forall x \; (x \in_T u \Rightarrow \psi)$. As usual, one should proceed by
case analysis on the last rule applied to get $\Gamma, \; t \in_T u \Rightarrow \phi[t/x] \vdash \psi$.
In all but one case, the main formula under consideration is not $t \in_T u \Rightarrow \phi[t/x]$ and it
is easy to use the induction hypothesis. The only interesting case thus occurs when the last rule applied was
the \rnimplLB~rule
$$
\dfrac{\Gamma, \; \phi \vdash \psi}{\Gamma, \; t \in_T u \Rightarrow \phi[t/x] \vdash \psi}
$$
In such a case, we know that $t \in_T u$ is a formula occurring in $\Gamma$, so we replace the application of
this rule with the new rule \rnallLB~of $\ljboundedq$ to conclude.
$$
\dfrac{\Gamma, \; \phi \vdash \psi}{\Gamma, \; \forall x \; (x \in_T u \Rightarrow \phi) \vdash \psi}
$$

The reasoning for the other rule \rnexR~is extremely similar, where the only interesting case occurs upon
applying a rule \rnandRB. The last two rules are also handled similarly, the interesting case for the
admissibility of \rnallR~(respectively \rnexL) being \rnimplR~(respectively \rnandL).
\end{proof}
\begin{corollary}
\label{cor:ljboundedq}
$\ljboundedconn$ and $\ljboundedq$ derive the same sequents.
\end{corollary}

\begin{lemma}
\label{lem:ljboundedvarq}
The rules \rnallLB, \rnallRB, \rnexLB~and \rnexRB~are admissible in $\ljboundedvarq$.
\end{lemma}
\begin{proof}
All four cases are proven in a similar manner. Exceptionally, the induction this time
is not over the size of the proofs, but rather on a quantity computed from the bounding
term occurring in the main quantifier of the rule. For instance, this would be $t$ in the
following instance of \rnexLB:
\begin{mathpar}
\inferrule* [left=\rnexLB] {\Gamma, \; x \in t, \; \phi \vdash \psi}{\Gamma, \; \exists x ~(x \in t \wedge \phi) \vdash \psi}
\end{mathpar}
The ``size'' of such a term $t$ is the pair $\<v_t, r_t\>$ computed as follows:
\begin{itemize}
\item There is an intuitive notion of size for types defined by induction:
$$
\begin{array}{lcl}
s(\ur) &=& 1 \\
s(\unit) &=& 1 \\
s(T_1 \times T_2) &=& 1 + s(T_1) + s(T_2) \\
s(\sett(T)) &=& 1 + s(T) \\
s(T_1 \times T_2) &=& 1 + s(T_1) + s(T_2) \\
\end{array}
$$
From this we can define the ``variable size'' of a term $t$, denoted
  $v_t$, to be the sum of the size of the free variables of t.
$$v_t = \sum_{\substack{x \in \freevars(t) \\ \text{$x$ of type $T$}}} s(T)$$
\item $r_t$ is the intuitive notion of size for terms, computed by induction over $t$:
$$
\begin{array}{lcl}
r_{c_i} &=& 1 \\ r_{(t,u)} &=& 1 + r_{t} + r_u \\
r_{()} &=& 1 \\ r_{\pi_i(t)} &=& 1 + r_t \\
\end{array}
$$
\end{itemize}

Then we can use the fact that the lexicographic product of $\mathbb{N}$ with itself is well-founded
to run induction over the pair $(v_t, r_t)$. Let us do so for the rule \rnexLB.
To this end, suppose that $t$ is a term such that the rule
\begin{mathpar}
\inferrule* {\Gamma, \; x \in u, \; \phi \vdash \psi}{\Gamma, \; \exists x ~(x \in u \wedge \phi) \vdash \psi}
\end{mathpar}
is admissible in $\ljboundedvarq$ for every $u$ such that either $v_u < v_t$ or $v_u = v_t$ and $r_u < r_t$.
We proceed with a case analysis to show that the same rule with $t$ instead of $u$ is admissible.
\begin{itemize}
\item If $t$ is a variable, then this is an instance of the rule \rnexLBV~of $\ljboundedvarq$.
\item Otherwise, if $t$ has a free variable $z$ of type $T_1 \times T_2$, one may apply the rule \rnpaireta
$$
\dfrac{\Gamma[\<z_1,z_2\>/z], \; \exists x \in t[\<z_1,z_2\>/z]~~ \phi[\<z_1,z_2\>/z] \vdash 
\psi[\<z_1,z_2\>/z]}{\Gamma, \; \exists x \in t ~~ \phi \vdash \psi}$$
and conclude using our induction hypothesis since $v_{t[\<z_1,z_2\>/z]} < v_{t}$.
\item Otherwise, if $t$ has no such free variable, but is itself not a free variable, then it is
necessarily of the shape $\pi_i(\<t_1,t_2\>)$ for some $i \in \{1,2\}$, so we may apply the rule \rnpairbeta
$$
\dfrac{\Gamma, \; \exists x \in t_i~~ \phi \vdash \psi}{\Gamma, \; \exists x \in t ~~ \phi \vdash \psi}$$
and conclude using our induction hypothesis as we have $v_{t_i} \le v_t$ and $r_{t_i} < r_t$.
\end{itemize}
\end{proof}
\begin{corollary}
\label{cor:ljboundedvarq}
$\ljboundedq$ and $\ljboundedvarq$ derive the same sequents.
\end{corollary}

\begin{lemma}
\label{lem:ljtoboundedvarq}
$\ljdeltazero$ and $\ljboundedvarq$ derive the same sequents.
\end{lemma}
\begin{proof}
Combine Corollaries~\ref{cor:ljboundedconn},~\ref{cor:ljboundedq} and~\ref{cor:ljboundedvarq}.
\end{proof}

\myparagraph{Step 3}
Recall that a right-hand side rule is one that changes the right-hand side formula. Among the
rules of $\ljboundedvarq$, these are the rules \rneqR, \rnandR, \rnandRB, \rnimplR, \rnallRBV~and \rnexRBV.
We call a proof tree \emph{right-focused} if every occurrence of sequent $\Gamma \vdash \psi$ in
the tree such that the top-level connective of $\psi$ is either $\forall$, $\Rightarrow$ or $\wedge$
is necessarily the conclusion of a right-hand side rule.

The rationale behind this choice is that the rules \rnandR, \rnandRB, \rnimplR~and \rnallRBV~are \emph{invertible}
(if their conclusion is true, so are all the premises), so they may be safely applied eagerly.

\begin{lemma}
\label{lem:rightfocus}
If $\Gamma \vdash \phi$ is derivable in $\ljboundedvarq$, then there is a right-focused 
$\ljboundedvarq$ proof tree of deriving $\Gamma \vdash \phi$.
\end{lemma}
\begin{proof}
The result is proven by induction over the depth of the proof-tree, and is straightforward.
We sketch one of the case: if the last rule applied is \rnorL~and the right-hand side formula
is an implication
$$\dfrac{\Gamma, \phi_1 \vdash \psi \Rightarrow \theta \qquad \Gamma, \phi_2 \vdash \psi \Rightarrow \theta}{\Gamma, \phi_1 \vee \phi_2 \vdash \psi \Rightarrow \theta}$$
by the induction hypotheses, we have right-focused proofs $\pi_i$ with conclusion $\Gamma, \phi_i, \psi \vdash \theta$ for $i \in \{1, 2\}$. We may then build the tree
$$
\dfrac{\dfrac{\dfrac{\pi_1}{\Gamma, \phi_1, \psi \vdash \theta} \qquad \dfrac{\pi_2}{\Gamma, \phi_2, \psi \vdash \theta}}{\Gamma, \phi_1 \vee \phi_2, \psi \vdash \theta}}{\Gamma, \phi_1 \vee \phi_2 \vdash \psi \Rightarrow \theta}$$
which is right-focused.
\end{proof}

\myparagraph{Step 4}
First, we observe that $\ljboundedvarq$ has a stronger variant of the subformula property:
if all formulas in the conclusion sequent $\Gamma \vdash \psi$ is the translation
of some $\deltazero$ formula, then all formulas occurring in a proof tree are actually
$\deltazero$ formulas.

\begin{lemma}
\label{lem:ljboundedvarq-to-restr}
If $\widetilde\Theta, \widetilde\Gamma \vdash \tilde\psi$ has a right-focused proof tree in 
$\ljboundedvarq$,
then there is a proof of $\Theta; \; \Gamma \vdash \psi$ in our restricted system.
\end{lemma}

The proof goes by induction over the right-focused $\ljboundedvarq$ proof tree.
All cases are immediate, except for the case of the congruence rule
$$
\dfrac{\Gamma[s/x,t/y] \vdash \psi[s/x,t/y]}{\Gamma[t/x,s/y], \; t =_\ur s \vdash \psi[t/x,s/y]}
$$
This particular case can be treated by showing that the obvious counterpart to this rule is
admissible in the restricted system before embarking on the proof of Lemma~\ref{lem:ljboundedvarq-to-restr}.

\begin{proposition}
\label{prop:eqcong-adm}
The following rule is admissible in our restricted proof system
$$
\dfrac{\Theta[s/x,t/y];\; \Gamma[s/x,t/y] \vdash \psi[s/x,t/y]}{\Theta[t/x,s/y]; \; \Gamma[t/x,s/y], \; t =_\ur s \vdash \psi[t/x,s/y]}
$$
\end{proposition}
Proposition~\ref{prop:eqcong-adm} can
be proven in a similar way as Lemma~\ref{lem:ljboundedvarq}, by reducing to the case where $s$ and $t$ are variables
using the rules $\times_\beta$ and $\times_\eta$.
Then, similarly to Lemma~\ref{lem:tolj}, Lemma~\ref{lem:ljboundedvarq-to-restr} is 
proven by a routine induction on the proof of the desired sequent in $\ljboundedvarq$, which allows 
us to complete the proof of Lemma~\ref{lem:fromlj}. 

\begin{proof}[Proof of Lemma~\ref{lem:fromlj}]
Assume $\widetilde\Theta, \; \widetilde\Gamma \vdash \tilde\psi$ is derivable in LJ. Because of the
subformula property, it is also derivable in $\ljdeltazero$ and thus, by Lemma~\ref{lem:ljtoboundedvarq},
it is also derivable in $\ljboundedvarq$. Then, Lemma~\ref{lem:rightfocus} shows that it can be done using
a right-focused proof, and then Lemma~\ref{lem:ljboundedvarq-to-restr} allows us to conclude that
$\Theta; \; \Gamma \vdash \psi$ is derivable in the restricted system.
\end{proof}




\subsection{Proof of interpolation for $\deltazero$ formulas in the intuitionistic proof system}

\myeat{
\michael{It would better to state both a classical and an intuistionistic
result}
\pierre{Probably better off as a remark at the end not to have to revamp the text?
``This result can easily be adapted to classical entailments. To do so one, it would suffice
the proof system in the body of the paper to allow for multiple conclusions on the right hand-side
and add in the rule of contraction on the right.''
Otherwise, you can explain the switch replace the right hand-side $\psi$ by some $\Delta_R$
and little else should change.
}
}

Recall that in the body of the paper we made use of a Craig interpolation result
for $\deltazero$ formulas, both for classical validity and intuitionistic provability.
Both may be proven in similar way, but 
we only give the proof for the intuitionistic case here. The classical result is obtained by taking a system with
multiple conclusions. With this caveat, the inductive proof is essentially the same. The
The precise rule can be found in the conclusion of the  body of the paper.


We restate the result, abusing notation by
eliding the difference between membership contexts and $\deltazero$ formulas:

\medskip

Let $\Lambda_L$ and $\Lambda_R$ be multi-sets each consisting possibly of formulas and membership contexts
and $\psi$ a formula. Let $\vec{i}$ be the collection of variables that occur in $\Lambda_L$ and which 
also occur in $\Lambda_R \msu \psi$.
Then for every derivation
$$\Lambda_L \msu \; \Lambda_R \vdash \psi$$
there exists a $\deltazero$ formula $\theta$ with free variables $\vec i$ such that the following holds
\[
\Lambda_L \models \theta \qquad \qquad \text{and} \qquad \qquad \Lambda_R \msu \; \theta \models \psi
\]
Further, there is a polynomial-time algorithm which outputs $\Theta$ when given as input a formal derivation of $\Lambda_L, \; \Lambda_R \vdash \psi$.

\medskip

We use induction on the complexity of the proofs, following the template
presented in  Fitting's textbook \cite{fittingbook}, see also the expositions of 
this  method
in \cite{wernhard, tomanweddell}.
We present here further representative cases of the rules, omitting many cases
that are either trivial or similar to rules that are already covered below.

In order for the inductive argument to go through, we assume that if we have
$t \in u$ in a $\in$-context, then $t$ does not contain a projection $\pi_i$
as a subterm. This can be guaranteed by transforming the proof so that
the initial steps consist of application of the rules $\times_\beta$ and $\times_\eta$,
which are invertible.

The base case consists of rules with no hypotheses.

Consider first the case of a proof consisting only of an application of the rule:

\[
\inferrule{ }{\Lambda \msu  \;   t \neq_\ursort t \vdash u \in_T v}
\]

Note that $t \neq_\ursort t$ is a $\deltazero$ formula representing $\false$, just
as $t =_\ursort  t$ represents $\true$.

If  $t \neq_\ursort t$ is in $\Lambda_L$ we generate $t \neq_\ursort t$, while
if it is in $\Lambda_R$ we generate $t =_\ursort  t$.

For the hypothesis-free rule:
\[
\inferrule{ }{\Theta \msu \; t_0 \in_\ursort u  \msu \; \Gamma \msu  \; t_0 =_\ursort t_1 \msu  \; \ldots \msu  \; t_{k-1} =_\ursort t_k \vdash t_k \in_\ursort u}
\]

we will generate $t \in_\ursort u$ if $t \in_\ursort u$ is in $\Lambda_L$,
and otherwise $\neg (t \in_\ursort u)$.

We now consider the case where the final rule applied is:
\[\inferrule{\Lambda \msu  \; t \in_{\sett(T)} v \vdash t =_{\sett(T)} u}{\Lambda \msu \;  t \in_{\sett(T)} v \vdash u \in_{\sett(T)} v}\]

First consider the subcase where $t \in_{\sett(T)} v$ is in $\Lambda_L$  within  the bottom sequent.
Thus our goal is to find an interpolant $\theta'$ which contains only variables common
to $\Lambda_L \msu \;   t \in_{\sett(T)} v $ and $\Lambda_R \msu \; u \in_{\sett(T)} v$.

We apply the induction hypothesis with the same decomposition of the left side into $L$ and $R$.
It gives us a $\theta$ such that 
$ \Lambda_L\msu \;  t \in_{\sett(T)} v  \vdash \theta$ and
$\Lambda_R \msu \; \theta \vdash  x =_{\sett(T)} u$,
and $\theta$ includes only variables that are common
to $\Lambda_L \msu  \; t \in_{\sett(T)} v $ and  $\Lambda_R \msu \; x =_{\sett(T)} u$.
Thus all the variables in $\theta$ meet the criteria for $\theta'$ except possibly for $t$.

We set $\theta' = \exists t \in v ~ \theta$.
The free variables in $\theta'$ are those of $\theta$ other than $t$, and also $v$, and
thus they meet the desired criteria.

It is easy to see using the properties of $\theta'$ that $\Lambda_L \msu \;   t \in_{\sett(T)} v \models \theta'$ and
$\Lambda_R \msu \; \theta' \models  u \in_{\sett(T)} v$ as required.

In the other subcase, where $t \in_{\sett(T)} v$ is in  $\Lambda_R$,  we can apply the induction hypothesis as above and set
$\theta'=\theta$.

We now turn to the case where the last proof rule is:
\[
\inferrule{\Lambda \msu \; z \in_T \msu \; t \vdash z \in_T u \and z \notin \freevars(\Lambda, t, u)}{\Lambda \vdash t \subseteq_T u}
\]
We call the induction hypothesis on the top sequent, splitting the formulas the same
way but putting $ z \in_T t$ in $\Lambda_R$. We can use the inductively formed interpolant  directly.

Let us turn to the case where the last rule applied is:
\[
\inferrule{\Lambda \msu \; t \in_{T} z \msu \;  \phi[t/y] \vdash v \in_{T'} w}{\Lambda \msu \; t \in_T z \msu \; \forall y \in_T z ~~ \phi\vdash v \in_{T'} w}
\]

To simplify matters, let us assume that $t$ is a single variable.
We first consider the subcase where $\forall y \in_T z ~~ \phi$ is in $\Lambda_R$ in the bottom.
We can apply the induction hypothesis to the top sequent with the partition of formulas being the one induced from 
the partition on the bottom.
The induction gives us a $\theta$ that may use the variable $t$, which may not occur in any formula within $\Lambda_R$
in the bottom sequent, and hence is not allowed in our interpolant for the bottom.
If this happens, then  this implies that $t \in_{T} z$ is in $\Lambda_L$ on the bottom.
 In this case we  set $\theta'= \exists y \in_{T} z~ \theta$.
It is clear that $\Lambda_L \msu \;  t \in_{T} z \models \theta'$.
Since $t$ does not occur in $\Lambda_L$  and
$\Lambda_L\ \msu \;  \phi[t/y] \msu \; \theta \models v \in_{T'} w$ by induction,
we conclude that $\Lambda_L \msu \;  \phi \msu \; \theta' \models v \in_{T'} w$ as required.

Now consider the subcase where  $\forall y \in z ~~ \phi$ is in  $\Lambda_L$ in the bottom sequent.
We apply induction in the same
way, to obtain $\theta$ as above. The only difficult case is when $t$  only occurs in formulas
within $\Lambda_R$ on the bottom. In this case we can check that $\theta' = \forall y \in z ~ ~ \phi$ can
be used as the desired interpolant.

\myeat{
The  case where the last rule applied is:
\[
\inferrule{\Theta \msu \; x \in z \msu \; \phi \msu \; \Lambda \vdash t \in x' \qquad x \notin \freevars(\Lambda \msu \Theta \msu t \msu x')}{\Theta \msu \; \exists x \in z~~\phi \msu \; \Lambda \vdash t \in x'}
\]
is handled similarly to the case immediately above.
}


\subsection{Proof of the higher-type interpolation lemma}
Recall the higher-type interpolation lemma from the body of the paper, which gives the inductive invariant used in the
 synthesis of $\nrcwget$ expressions from proof:

\medskip
Let $\Theta = \Theta_L, \Theta_R$ be a $\in$-context and $\Gamma = \Gamma_L, \Gamma_R$ a context.
Call $L = \freevars(\Theta_L, \Gamma_L)$ the set of left-hand side variables, $R = \freevars(\Theta_R, \Gamma_R)$ the set of right-hand side variables, and $C = \freevars(\Theta_L, \Gamma_L) \cap \freevars(\Theta_R, \Gamma_R)$ the set of common free variables.
Suppose that $t$ and $u$ are terms of suitable types such that $\freevars(t) \subseteq L$
and $\freevars(u) \subseteq R$ and
Then we have:
\begin{compactitem}
\item If $\Theta; \; \Gamma \;\vdash; t =_{T} u$ is derivable, there is an $\nrc$ expression $E$ of type $T$ such that $\Theta; \; \Gamma \models t = E = u$ and $\freevars(E) \subseteq C$.
\item If $\Theta; \; \Gamma \;\vdash\; t \subseteq_T u$ is derivable, there is an $\nrc$ expression $E$ of type $\sett(T)$ such that $\Theta; \; \Gamma \models t \subseteq E \subseteq u$ and $\freevars(E) \subseteq C$.
\item If $\Theta;\; \Gamma \;\vdash\; t \in_T u$ is derivable, then there is an $\nrc$ expression $E$ of type $\sett(T)$ such that $\Theta; \; \Gamma \models t \in E$ and $\freevars(E) \subseteq C$.
\end{compactitem}

Further the desired expressions can be constructed in time polynomial in the size of the proof (e.g. measured
in terms of the number of steps and the maximal size of a sequent in each step).

\medskip

\begin{proof}
First, we assume that if we have $t \in u$ in $\Theta_L, \; \Theta_R$,
then $t$ does not contain a projection $\pi_i$
as a subterm. This can be guaranteed by transforming the proof so that
the initial steps consist of application of the rules $\times_\beta$ and $\times_\eta$,
which are invertible.

We proceed by induction over the proof tree, calling $E$ the  desired 
expression that we want to create in the
inductive step. In each case we will prove the result
for the bottom sequent of a proof rule by making a single call
to the induction hypothesis for each sequent on top of the proof rule.
We will require a partition of the symbols in the top sequent, but it will
always be clear from the bottom sequent.
\begin{itemize}
\item If the last proof rule used is contraction, we directly use the induction hypothesis.
\item If the last proof rule used is $=_\sett$-\textsc{R} then we directly use the induction hypothesis as well.
$$
\dfrac{\Theta; \; \Gamma \vdash t \subseteq u \qquad \qquad \Theta; \; \Gamma \vdash u \subseteq t}{\Theta; \; \Gamma \vdash t = u}
$$
then one has a transformation  $E'$ such that $\Gamma \models t \subseteq E' \subseteq u$ by applying the induction
hypothesis on the first subproof. Since the system is sound, we do have $\Gamma \models t = u$, so
$\Gamma \models t = E' = u$. We can thus take $E = E'$.
\item If the last proof rule used is $=_\times$-\textsc{R}
$$
\dfrac{\Theta; \; \Gamma \vdash \pi_1(t) =_{T_1} \pi_1(u) \qquad \qquad \Theta; \; \Gamma \vdash \pi_2(t) =_{T_2} \pi_2(t)}{\Theta; \; \Gamma \vdash t =_{T_1 \times T_2} u}
$$
The induction hypothesis yields $\nrc$ expressions $E_1$ and $E_2$ such that
$$\Theta; \; \Gamma \models \pi_1(t) = E_1 = \pi_1(u) \qquad \qquad \text{and} \qquad \qquad
\Theta; \; \Gamma \models \pi_2(t) = E_2 = \pi_2(u)$$
It suffices to take $E = (E_1,E_2)$.

\item If the last proof rule used is $=_\unit$-\textsc{R}
$$
\dfrac{}{\Theta; \; \Gamma \vdash t =_\unit u}
$$
Then the expression   returning the unique element of $\unit$ works.

\item If the last proof used is $=_\ur$-\textsc{R}
$$
\dfrac{\Theta, \; t \in_\ursort z; \; \Gamma \vdash u \in_\ursort z \qquad \qquad z \notin \freevars(\Theta, \Gamma,t,u)}{\Theta; \; \Gamma \vdash t =_\ursort u}$$
The induction hypothesis gives us an expression $E'$ of type $\sett(\ursort)$ such that
$$\Theta; \; \Gamma, \; t \in_\ursort z \models u \in E'$$
Note that since $z$ is fresh, we must actually have
$$
\Theta; \; \Gamma \models u \in E'
$$
Applying interpolation, there is a $\deltazero$ formula $\theta(\vec i, z)$ such that
$$\Theta_I, \; \Theta_L; \; \Gamma_L, \; t \in_\ursort z \models \theta(\vec i, z)
\qquad
\text{and}
\qquad
\Theta_R; \; \Gamma_R, \; \theta(\vec i, z) \models u \in_\ursort z
$$
This means that we have
$$\Theta; \; \Gamma, \models \theta(\vec i, z) \leftrightarrow t \in z
$$
In particular $\Theta; \Gamma$
entails that $\{t\}$ is the unique singleton set $z$ satisfying $\theta(vec i,z)$.

So we may take $E$ to be the unique element of $\{ x \in E' \mid \theta(\vec i, \{x\})\}$, which can be
formally defined in $\nrc$ as
$$E = \nrcget \left( \bigcup \{ \case(\verify_\theta(\vec i, \{x\}), \{x\}, \emptyset) \mid x \in E'\} \right)$$

\item If the last proof rule used is $\subseteq$-\textsc{R}
$$
\dfrac{\Theta, \; z \in_T t; \; \Gamma \vdash z \in_T u \qquad z \notin \freevars(\Theta; \; \Gamma, t, u)}{\Theta; \; \Gamma \vdash t \subseteq_T u}
$$
then the inductive hypothesis gives us an expression $E'(\vec i)$ such that
$$\Theta; \; \Gamma \models  z \in E'$$
Apply interpolation to the premise so as to obtain a $\deltazero$ formula $\theta(\vec i, z)$ with
$$\Theta_I, \; \Theta_L; \; \Gamma_L, \; z \in t \models \theta(\vec i, z) \qquad
\text{and}
\qquad
\Theta_R; \; \Gamma_R, \; \theta(\vec i, z) \models z \in u$$
In this case, we take
$$E(\vec i) = \{ z \in E'(\vec i) \mid \theta(\vec i, z) \}$$
which is $\nrc$-definable as
$$\bigcup \{ \case(\verify_\theta(\vec i, z), \{ z \}, \emptyset) \mid z \in E'(\vec i) \}$$
Now, let us assume that $\Gamma$ is valid and show that $t \subseteq E$ and $E \subseteq u$.
\begin{itemize}
\item Suppose that $z \in t$. By the induction hypothesis, we know that $z \in E'$.
But we also know that $\Gamma_L$ is valid, so that $\theta(\vec i, z)$ holds.
By definition, we thus have $z \in E$.
\item Now suppose that $z \in E$, that is, that $z \in E'$ and $\theta(\vec i, z)$ holds.
The latter directly implies that $z \in u$ since $\Gamma_R$ is valid.
\end{itemize}

\item If the last proof rule used is $\in_\sett$-\textsc{R}
$$
\dfrac{\Theta, \; t \in_{\sett(T)} v; \; \Gamma \vdash t =_{\sett(T)} u}{
\Theta, \; t \in_{\sett(T)} v; \; \Gamma \vdash u \in_{\sett(T)} v}
$$
then, by using the induction hypothesis on the premise,
we get an expression $E'$ which is equal to $u$ assuming $\Theta, t \in_{\sett(T)} v; \; \Gamma$.
So we may take $E = \{E'\}$.

\item
If the last proof rule used is $\in_\ur$-\textsc{R}
$$\dfrac{}{\Theta, \; t \in_\ursort u; \; \Gamma \vdash t \in_\ursort u}$$
then it means that $\freevars(t) \subseteq C$, so we may take the expression $\{t\}$.
\item
If the last rule used is $\times_\beta$ or $=$-\textsc{subst} 
\begin{mathpar}
\inferrule{\Theta[t_i/y]; \; \Gamma[t_i/y] \vdash (t \in_T u)[t_i/y] \and i \in \{1,2\}}{\Theta; \; \Gamma[\pi_i(\<t_1,t_2\>)/y] \vdash (t \in_T u)[\pi_i(\<t_1,t_2\>)/y]}
\and
\inferrule{\Theta[y/x]; \; \Gamma[y/x] \vdash v[y/x] \in_T w[y/x]}{\Theta; \; \Gamma, \; x =_\ur y \vdash w \in_T v}
\end{mathpar}
the expression obtained using the induction hypothesis allows to reach our conclusion.
\item If the last rule used is $\times_\eta$
\begin{mathpar}
\inferrule{\Theta[\<x_1, x_2\>/x]; \; \Gamma[\<x_1,x_2\>/x] \vdash (t \in_T u)[\<x_1,x_2\>/x] \and x_1,x_2 \notin \freevars(\Theta; \; \Gamma, t, u)}{\Theta; \; \Gamma \vdash t \in_T u}
\end{mathpar}
then the induction hypothesis yields an expression $E'$. If $x \notin L \cap R$, then we also have that $x_1, x_2 \notin L \cap R$,
so $E'$ has the expected free variables and we may set $E = E'$.
Otherwise, $x_1$ and $x_2$ are among the free variables of $E'$ and $x \in L \cap R$. Writing $E'(\vec z, x_1, x_2)$ to
clarify the free variables, it suffices to set
$$E(\vec z, x) ~~=~~ E'(\vec z, \pi_1(x), \pi_2(x))$$

\item
If the last proof rule is \rnbot
\begin{mathpar}
\inferrule{ }{\Theta; \; \Gamma,\;  \bot \vdash t \in_T u}
\end{mathpar}
then, any expression can be used since the premise is contradictory.
This is also the case for the rule $\neq-\textsc{L}$.

\item
If the last proof rule is \rnandL
\begin{mathpar}
\inferrule{\Theta; \; \Gamma, \;  \phi,\;  \psi  \vdash t \in_T u}{\Theta; \; \Gamma, \; \phi \wedge \psi \vdash t \in_T u}
\end{mathpar}
one may directly take the expression given by the induction hypothesis.

\item
If the last proof rule used is \rnorL
\begin{mathpar}
\inferrule{\Theta; \; \Gamma, \;  \phi \vdash t \in_T u \qquad
\Theta; \; \Gamma, \;  \psi \vdash t \in_T u
}{\Theta; \; \Gamma, \;  \phi \vee \psi \vdash t \in_T u}
\end{mathpar}
the induction hypothesis yields expressions $E_1$ and $E_2$ of sort $\sett(T)$ such that
$$ \Theta; \; \Gamma, \; \phi \models t \in E_1 \qquad \text{and} \qquad \Theta; \; \Gamma, \; \psi \models t \in E_2$$
So we may take $E = E_1 \cup E_2$.

\item
Suppose the last proof rule used is \rnallL
\begin{mathpar}
\inferrule{\Theta, \; t \in_{T} z; \; \Gamma, \; \phi[t/y] \vdash v \in_{T'} w}{\Theta, \; t \in_T z ; \; \Gamma , \; \forall y \in_T z ~~ \phi\vdash v \in_{T'} w}
\end{mathpar}
If $t \in_T z$ and $\forall y \in_T z ~~ \phi$ are both part of the left-hand side or right-hand side,
then we may directly use the inductive hypothesis to obtain an expression $E'$, and we may check that $E = E'$ satisfies the
inductive invariant.
Otherwise, it might be the case that $E'$ contains some additional variables $x_1, \ldots, x_k$ from the term $t$ and that $z \in L \cap R$.
Recall that our preliminary assumption means that $t$ does not contain any projection, so that we have terms $p_1, \ldots, p_k$ with
a single variable $u$ such that $p_i[t/u]$ is semantically equivalent to $x_i$.
Then, we may show that
\[ E = \bigcup \{ E'\left[p_1/x_1, \ldots, p_k/x_k\right] \mid u \in y \} \]
satisfies the invariant.

\item If the last proof rule used is \rnexL
\begin{mathpar}
\inferrule{\Theta, \; x \in_T y; \; \Gamma, \; \phi \vdash t \in_{T'} v \qquad x \notin \freevars(\Theta, \Gamma, y, t, v)}{\Theta; \; \Gamma, \; \exists x \in_T y~~\phi \vdash t \in_{T'} v}
\end{mathpar}
we may apply the induction hypothesis to obtain $E'$ that also satisfy the invariant in the conclusion (note that $\freevars(E') \subseteq L \cap R$ since $x$ is fresh),
so we can conclude by taking $E=E'$.
\end{itemize}
\end{proof}


\section{Reduction to monadic schemas}
In the body of the paper we mentioned a reduction of problems about
$\nrc$ and interpretations to the case of Monadic schemas. This was
explicitly stated in Section 6, but we make use of it also in the arguments
for converting between interpretations and $\nrcwget$ in Section 5.

\subsection*{Reduction to monadic schemas for $\nrc$}
In the body of the paper
we mentioned that it is possible to reduce questions about
definability within $\nrc$ to the case of monadic schemas.
We now give the details of this reduction.

Recall that  \emph{monadic type} is a type built only using
the atomic type $\ursort$ and the type constructor
$\sett$. Monadic types are in one-to-one correspondence with
natural numbers by setting
$\ursort_0 \eqdef \ursort$ and $\ursort_{n+1} \eqdef \sett(\ursort_n)$.
A monadic type is thus a $\ursort_n$ for some $n \in \bbN$.
A nested relational schema is monadic if it contains only monadic types,
and a $\deltazero$ formula is said to be monadic if it all of its variables have monadic types.

We start with a version of the reduction only for $\nrc$ expressions:

\begin{proposition} \label{prop:reducemonadic-nrc}
For any nested relational schema $\aschema$, there is
a monadic nested relational schema $\aschema'$,
 an injection $\convert$ from instances of $\aschema$ to instances
of $\aschema'$ that is definable in $\nrc$, and an $\nrc[\nrcget]$ expression
$\convert^{-1}$
such that $\convert^{-1} \circ \convert$ is the identity transformation from
$\aschema \to \aschema$.

Furthermore, there is a $\deltazero$ formula $\image_\convert$ from $\aschema'$ to $\bool$
such that $\image_\convert(i')$ holds if and only if $i' = \convert(i)$
for some instance $i$ of $\aschema$.
\end{proposition}

To prove this we  give an encoding of
general nested relational schemas into monadic nested relational schemas that
will allow us to reduce the equivalence between $\nrc$ expression,
interpretations, and implicit definitions to the case where
input and outputs are monadic.

Note that it  will turn out to be crucial to check that this encoding may be
defined \emph{either} through $\nrc$ expressions or interpretations, but
in this subsection we  will give the definitions in terms of $\nrc$ expressions.

The first step toward defining these encodings is actually to
emulate in a sound way the cartesian product structure for
types $\ursort_n$. Here ``sound'' means that we should give
terms for pairing and projections that satisfy the usual
equations associated with cartesian product structure.

\begin{proposition}
\label{prop:nrc-monadicproduct}
For every $n_1, n_2 \in \bbN$, there are $\nrc$ expressions
$\widehat{\pair}(x,y) : \ursort_{n_1}, \ursort_{n_2} \to \ursort_{\max(n_1,n_2) + 2}$
and $\nrcwget$ expressions $\hat\pi_i(x) : \ursort_{\max(n_1,n_2) + 2} \to \ursort_{n_i}$ for $i \in \{1,2\}$
such that the following equations hold
$$
\hat{\pi}_1\left(\widehat{\pair}(a_1,a_2)\right) ~ = ~ a_1
\qquad \qquad
\hat{\pi}_2\left(\widehat{\pair}(a_1,a_2)\right) ~ = ~ a_2
$$
Furthermore, there is a $\deltazero$ formula $\image_{\widehat{\pair}}(x)$ such that
$\image_{\widehat{\pair}}(a)$ holds if and only if there exists $a_1, a_2$ such
that $\widehat\pair(a_1,a_2) = a$. In such a case, the following also holds
$$\widehat{\pair}(\hat\pi_1(a),\hat\pi_2(a)) ~ = ~ a$$
\end{proposition}

\begin{proof}
We adapt the Kuratowski encoding of pairs $(a,b) \mapsto \{ \{a\},\{a,b\}\}$.
The notable thing here is that, for this encoding
to make sense in the typed monadic setting, the types of $a$ and $b$ need to be the same.
This will not be an issue because we have $\nrc$-definable embeddings 
$$\uparrow_n^m : \ursort_n \to \ursort_m$$
for $n \le m$ defined as the $m - n$-fold composition of the singleton transformation $x \mapsto \{x\}$.
This will be sufficient to define the analogues of pairing for monadic types and thus
to define $\convert_T$ by induction over $T$.
On the other hand, $\convert^{-1}_T$ will require a suitable encoding of projections.
This means that to decode an encoding of a pair, we need to make use of a transformation inverse to
the  singleton construct $\uparrow$. But we have this  thanks to the $\nrcget$ construct. We let
$$\downarrow_n^m : \ursort_m \to \ursort_n$$ the transformation inverse to
 $\uparrow_n^m$, defined as the $m - n$-fold composition of $\nrcget$.

Firstly, we define the family of transformations $\widehat{\pair}_{n,m}(x_1,x_2)$,
where $x_i$ is an input of type $\ursort_{n_i}$ for $i \in \{1,2\}$ and
the output is of type $\ursort_{\max(n_1,m_2)+2}$, as follows
$$\widehat{\pair}_{n_1,n_2}(x_1,x_2) \; \eqdef \;
\{ \{ \uparrow x_1 \}, \{\uparrow x_1, \uparrow x_2\}\}$$

The associated projections
$\hat\pi_i^{n_1,n_2}(x)$ where $x$ has type $\ursort_{\max(n_1,n_2)+2}$
and the output is of type $\ursort_{n_i}$
are a bit more challenging to construct.
The basic idea is that there is
first a case distinction to be made for encodings $\widehat{\pair}_{n,m}(x_1,x_2)$:
depending on whether  $\uparrow x_1 = \uparrow x_2$
or not. This can be actually tested by a $\nrc$ expression.
Once this case distinction is made, one may informally
compute the projections as follows:
\begin{itemize}
\item if $\uparrow x_1 = \uparrow x_2$, both projections can be computed
as a suitable downcasting $\downarrow$ (the depth of the downcasting is determined by the output
type, which is not necessarily the same for both projections).
\item otherwise, one needs to single out the singleton $\{\uparrow x_1\}$ and the
two-element set $\{ \uparrow x_1, \uparrow x_2 \}$ in $\nrc$. Then, one
may compute the first projection by downcasting the singleton, and the second projection
by first computing $\{\uparrow x_2\}$ as a set difference and then downcasting with $\downarrow$.
\end{itemize}

We now give the formal encoding for projections, making a similar case distinction.  
To this end, we first define a generic $\nrc$ expression
$$\allpairs_T(x) : \sett(T) \to \sett(T \times T)$$
computing all the pairs of distinct elements of its input $x$
$$\allpairs_T(x) = \bigcup \{ \bigcup \{ \{(y, z)\} \mid y \in x \setminus \{z\}\} \mid z \in x \}$$
Note in particular that $\allpairs(i) = \emptyset$ if and only if $i$ is a singleton or the empty set.
The projections can thus be defined as
$$
\begin{array}{l!~c!~l}
\hat{\pi}_1(x) &\eqdef&
\case\left(\allpairs(x) = \emptyset,~ \downarrow x,~ \downarrow \bigcup\{ \pi_1(z) \cap \pi_2(z) \mid z \in \allpairs(x) \}\right) \\
\tilde{\pi_2}(x) &\eqdef& \case\left(\allpairs(x) = \emptyset, ~
\downarrow x, ~ \downarrow (x \setminus \uparrow \hat{\pi_1}(x)))\right) \\
\end{array}
$$
These definitions crucially ensure that,
for every object $a_i$ with $i \in \{1,2\}$, we  have
$$\hat{\pi}_i\left(\widehat{\pair}(a_1,a_2)\right) ~ = ~ a_i$$

Now all remains to be done is to define $\image_{\widehat\pair}$.
Before that, it is helpful to define a formula $\image_{\uparrow_n^m}(x)$
which holds if and only if $x$ is in the image of $\image_{\uparrow_n^m}$.

As a preliminary step, define generic $\deltazero$ formulas $\issing(x)$ and $\istwo(x)$
taking an object of type $\sett(T)$ and returning a Boolean indicating
whether the object is a singleton or a two-element set.
Defining $\image_{\uparrow_n^m}$ is straightforward using $\issing$ and Boolean connectives.
Then $\image_{\widehat{\pair}_{n,n}}(x)$ can be defined as follows for each $n \in \bbN$
$$
\begin{array}{lcl}
\image_{\widehat{\pair}_{n,n}}(x) &\eqdef&
\left(\issing(x) \wedge \image_{\widehat{\pair}_{n,n}}^{\issing}(x)\right) \vee 
\left(\istwo(x) \wedge \image_{\widehat{\pair}_{n,n}}^{\issing}(x)\right)
\\ 
\image_{\widehat{\pair}_{n,n}}^\issing(x) &\eqdef& \exists z \in x~~ \issing(z) \\
\image_{\widehat{\pair}_{n,n}}^\istwo(x) &\eqdef& \exists z \; z' \in x~~ (\istwo(z) \wedge \issing(z') \wedge \forall y \in z~~ y \in z')
\end{array}
$$
Then, the more general $\image_{\widehat{\pair}_{n_1,n_2}}$ can be defined using
$\image_{\uparrow_{n_i}^{m}}$ where $m = \max(n_1,n_2)$.
$$
\begin{array}{lll}
\image_{\widehat{\pair}_{n_1,n_2}}(x)
&\eqdef&
\image_{\widehat{\pair}_{m,m}}(x) ~\cap~
\image_{\uparrow_{n_1}^m}(\widehat\pi_1(x)) ~\cap~
\image_{\uparrow_{n_2}^m}(\widehat\pi_2(x))
\end{array}
$$

One can then easily check that $\image_{\widehat{\pair}}$ does have the advertised property:
if $\image_{\widehat{\pair}}(a)$ holds for some
object $a$, then there are $a_1$ and $a_2$ such that $\widehat\pair(a_1,a_2) = a$ and we have
$$\widehat{\pair}(\hat\pi_1(a),\hat\pi_2(a)) ~ = ~ a$$
\end{proof}

We are now ready to give the proof of the proposition  given
at the beginning of this subsection.

\begin{proof}

$\convert_T$, $\convert^{-1}_T$ and $\image_{\convert_T}$ are defined
by induction over $T$. Beforehand, define the
map $d$ taking a type $T$ to a natural number $d(T)$ so
that $\convert$ maps instances of type $T$ to monadic types $\ursort_{d(T)}$.
$$
\begin{array}{l!~c!~ l !\qquad l!~c!~l}
d(\ursort) &=& 0 &
d(\sett(T)) &=& 1 + d(T) \\
d(T_1 \times T_2) &=& 2 + \max(d(T_1),d(T_2))
&
d(\unit) &=& 0 \\
\end{array}
$$

$\convert_T$, $\convert^{-1}_T$ and $\image_{\convert_T}$ are then defined by the following clauses,
where we write $\map\left(z \mapsto E\right)(x)$ for the $\nrc$ expression $\bigcup\{\{E\} \mid z \in x \}$. 

$$
\begin{array}{lcl}
\convert_\ursort(x) &\eqdef& x \\
\convert_{\sett(T)}(x) &\eqdef& \map\left(z \mapsto \convert_{T}(z)\right)(x) \\
\convert_\unit(x) &\eqdef& c_0 \\
\convert_{T_1 \times T_2}(x) &\eqdef& \widehat{\pair}(\convert_{T_1}(\pi_1(x)), \convert_{T_2}(\pi_2(x))) \\
\\
\convert^{-1}_\ursort(x) &\eqdef& x \\
\convert^{-1}_{\sett(T)}(x) &\eqdef& \map\left(z \mapsto \convert^{-1}_{T}(z)\right)(x) \\
\convert_\unit(x) &\eqdef& () \\
\convert^{-1}_{T_1 \times T_2}(x) &\eqdef& \left\< \convert^{-1}_{T_1}(\hat\pi_1(x)), \convert_{T_2}^{-1}(\hat\pi_2(x)) \right\> \\
\\
\image_{\convert_\ursort}(x) &\eqdef& \true \\
\image_{\convert_{\sett(T)}}(x) &\eqdef& \forall z \in x ~~ \image_{\convert_{T}}(z) \\
\image_{\convert_{T_1 \times T_2}}(x) &\eqdef&
\image_{\pair_{d(T_1),d(T_2)}}(x) \wedge
\image_{\convert_{T_1}}(\hat\pi_1(x)) \wedge \image_{\convert_{T_2}}(\hat\pi_2(x))
\end{array}
$$
It is easy to check, by induction over $T$, that for every object $a$ of type $T$
$$\convert^{-1}(\convert(a)) = a$$
and that for every object $b$ of type $\ursort_{d(T)}$, if $\image_{\convert_T}(b) = \true$, then it lies in the image
of $\convert_T$ and $\convert(\convert^{-1}(b)) = b$.
\end{proof}


\subsection{Monadic reduction for interpretations}
We have  seen so far that it is possible to reduce questions about
definability within $\nrc$ to the case of monadic schema.
Now we turn to the analogous statement for interpretations, given by
the following proposition:

\begin{proposition}
\label{prop:reducemonadic-interp}
For any object schema $\aschema$, there is
 a monadic nested relational schema $\aschema'$,
 a $\deltazero$ interpretation $\interp_\convert$ from instances of $\aschema$ to instances
of $\aschema'$, and another interpretation
$\interp_{\convert^{-1}}$ from instances of $\aschema$ to instances of $\aschema'$ compatible
with $\convert$ and $\convert^{-1}$ as defined in Proposition~\ref{prop:reducemonadic-interp} in the following sense:
for every instance $I$ of $\aschema$ and
for every instance $J$ of $\aschema'$ in the codomain of $\convert$, we have

$$\convert^{-1}(J) = \collapse(\interp_{\convert^{-1}}(J)) \qquad \qquad \convert(I) = \collapse(\interp_\convert(I))$$
\end{proposition}

Before proving Proposition~\ref{prop:reducemonadic-interp}, it is helpful to check that
a number of basic $\nrc$ connectives may be defined at the level of interpretations.
To do so, we first present a technical result for more general interpretations.

\begin{proposition}
\label{prop:mostowski-collapse-interp}
For any sort $T$, there is an interpretation
of $\aschema_T$ into $\aschema_T$ taking a models $M$ whose every sort is non-empty and $\booltype$
has at least two elements to a model $M$ of $\oneobjth(T)$.
Furthermore, we have that $M'$ is (up to isomorphism) the largest quotient of $M'$ satisfying $\oneobjth(T)$.
\end{proposition}
\begin{proof}
This interpretation corresponds to a quotient of the input, that is definable at
every sort
\[
\begin{array}{lcl!\qquad lcl}
\varphi_\equiv^{\sett(T)}(x,y) &~~=~~&
\forall z~(z \in x \Leftrightarrow z \in y)
& \varphi_\equiv^{T_1 \times T_2}(x,y) &~~=~~& \pi_1(x) = \pi_1(y) \wedge \pi_2(x) = \pi_2(y) \\
\varphi_\equiv^{\unit}(x,y) &~~=~~& \top
& \varphi_\equiv^{\ursort}(x,y) &~~=~~& x =_\ursort y \\
\end{array}
\]
\end{proof}

\begin{proposition}
\label{prop:interp-sing-cup-map-pair}
The following $\deltazero$-interpretations are definable:
\begin{itemize}
\item $\interp_\sing$ defining the transformation $x \mapsto \{x\}$.
\item $\interp_\cup$ defining the transformation $x,y \mapsto x \cup y$.
\end{itemize}
Furthermore, assuming that $\interp$ is a $\deltazero$-interpretation
defining a transformation $E$ and $\interp'$ is a $\deltazero$-interpretation
defining a transformation $R$, the following $\deltazero$-interpretations
are also definable:
\begin{itemize}
\item $\map(\interp)$ defining the transformation $x \mapsto \{ E(y) \mid x \in y\}$.
\item $\<\interp,\interp'\>$ defining the transformation $x,y \mapsto (E(x),F(y))$.
\end{itemize}
\end{proposition}
\begin{proof}
\begin{itemize}
\item
For the singleton construction $\{e\}$ with $e$ of type $T$, we take the interpretation
$\interp_e$ for $e$, where $e$ itself is interpreted by a constant $c$ and we add an extra
level represented by an input constant $c'$. Then $\varphi_\domainof^{\sett(T)}(x)$ is set to $y = c'$ and
$\phi_\in^T(x,y)$ to $x = c \wedge y = c'$.
\item The empty set $\{\}$ at type $\sett(T)$ is given by the trivial interpretation
where $\varphi_\domainof^{\sett(T)}(x)$ is set to $x = c$ for some constant $c$
and $\varphi_\domainof^{T'}$ is set to false for $T'$ a component type of $T$, as well as all the $\varphi_\in^T$.
\item For the binary union $\cup : \sett(T), \sett(T) \to \sett(T)$, the interpretation is easy: $T$
is interpreted as itself. The difference between input and output
is that $\sett(T) \times \sett(T)$ is not an output sort and that $\sett(T)$ is interpreted as a single
element, the constant $()$ of $\unit$.
$$
\begin{array}{lcl}
\phi_\domainof^{\sett(T)}(x) &~~ \eqdef ~~& x = () \\
\phi_\in^T(z,x) &~~ \eqdef ~~& z \in \pi_1(\inobj) \vee z \in \pi_2(\inobj)
\end{array}
$$
\item We now discuss the $\map$ operator. Assume that we have an interpretation
$\interp$ defining a transformation $S \to T$ that we want to lift to an interpretation $\map(\interp) : \sett(S) \to \sett(T)$.
Let us write $\psi_\domainof^{T'}$, $\psi_\in^{T'}$ and $\psi_\equiv^{T'}$ for the formulas making
up $\interp$ and reserve the $\phi$ formulas for $\map(\interp)$.
At the level of sort, let us write $\interpsort^\interp$ and $\interpsort^{\map(\interp)}$ to distinguish the two.

For every $T' \subtype T$ such that $T'$ is not a cartesian product or a component type of $\booltype$,
we set $\interpsort^{\map(\interp)}(T') = S, \interpsort^{\interp}$. This means that objects of
sort $T'$ are interpreted as in $\interp$ with an additional tag of sort $S$.
We interpret the output object $\sett(T)$ as a singleton by setting $\interpsort^{\map(\interp)}(\sett(T)) = \unit$.

Assuming that $T \neq \ursort, \unit$, $\map(\interp)$ is determined by setting the following 
$$
\begin{array}{lcl}
\phi_\domainof^\ursort(a) &~~\eqdef~~& \exists s \in \inobj~~ \psi_\domainof(a)[s/\inobj] \\
\phi_\in^\ursort(a, s, \vec x) &~~\eqdef~~& \psi_\in^\ursort(a,\vec x)[s/\inobj] \\
\\
\phi_\domainof^{T'}(s, \vec x) &~~\eqdef~~& \psi_\domainof^{T'}(x)[s'/\inobj] \\
\phi_\in^{T'}(s, \vec x, s', \vec y) &~~\eqdef~~& \exists \vec{x'}~~ \psi_\in^{T'}(\vec{x'}, \vec y)[s'/\inobj] \wedge \phi^{T'}_\equiv(s,\vec x,s',\vec x') \\
\\
\phi_\domainof^{T}(s, \vec x) &~~\eqdef~~& s \in \inobj \\
\phi_\in^{T}(s, \vec x) &~~\eqdef~~& \phi_\domainof^T(s, \vec x) \\
\end{array}
$$
where $[x/\inobj]$ means that we replace occurrences of the constant
$\inobj$ by the variable $x$ and sorts $T'$ and $T' \times T''$ are component types of
$T$. Note that this definition is technically by induction over the type, as we use
$\phi^{T'}_\equiv$ to define $\phi_\in^{T'}$.
In case $T$ is $\ursort$ or $\unit$, the last two formulas $\phi_\domainof^{T}$ and $\phi_\in^{T}$ need to change.
If $T = \unit$, then we set
$$\phi_\domainof^{\unit}(c_0) \eqdef \phi_\in^\unit(c_0,c_0) \eqdef \exists s \in \inobj~~ \top$$
and if $T = \ursort$, we set
$$\phi_\domainof^{\ursort}(a) \eqdef \phi_\in^\ursort(a) \eqdef  \exists s \in \inobj~~ \psi_\domainof(a)[s/\inobj]$$
\item Finally we need to discuss the pairing of two interpretation-definable
transformations $\< \interp_1, \interp_2 \> : S \to T_1 \times T_2$.
Similarly as for map we reserve
$\phi_\domainof^T$, $\phi_\in^T$ and $\phi_\equiv^T$
formulas for the interpretation $\< \interp_1, \interp_2 \>$.
We write
$\psi_\domainof^T$, $\psi_\in^T$ and $\psi_\equiv^T$
for components of $\interp$ and
$\theta_\domainof^T$, $\theta_\in^T$ and $\theta_\equiv^T$
for components of $\interp'$.

Now, the basic idea is to interpret output sorts of $\< \interp_1, \interp_2 \>$ as tagged unions
of elements that either come from $\interp_1$ or $\interp_2$.
Here, we exploit the assumption that $\aschema_T$ contains the sort $\booltype$.
and that every sort is non-empty to interpret the tag of the union.
The union itself is then encoded as a concatenation of a tuple representing
a would-be element form $\interp_1$ with another tuple representing a would-be element from $\interp_2$,
the correct component being selected with the tag. For that second trick to work, note that
we exploit the fact that every sort has a non-empty denotation in the input structure.
Concretely, for every $T$ component type of either $T_1$ or $T_2$, we thus set
$$
\begin{array}{lcll}
\interpsort^{\<\interp_1,\interp_2\>}(T) &~~ \eqdef ~~& \booltype, \interpsort^{\interp_1}(T), \interpsort^{\interp_2}(T)
\\
\phi_\domainof^T(u, \vec x, \vec y) &~~ \eqdef ~~& (u = \booltt \wedge \psi_\domainof^T(\vec x)) \vee (u \neq \booltt \wedge \theta_\domainof^T(\vec y)) \\
\phi_\in^T(u, \vec x, \vec y, u', \vec x', \vec {y'}) &~~ \eqdef ~~& (u = u' = \booltt \wedge \psi_\in^T(\vec x, \vec x')) \vee (u = u' = \boolff \wedge \theta_\in^T(\vec {y}, \vec {y'})) \\
\phi_\equiv^T(u, \vec x, \vec y, u', \vec x', \vec {y'}) &~~ \eqdef ~~& (u = u' = \booltt \wedge \psi_\equiv^T(\vec x, \vec x')) \vee (u = u' = \boolff \wedge \theta_\equiv^T(\vec {y}, \vec {y'})) \\
\end{array}
$$
Note that this interpretation does not quite correspond to a pairing
because it is not a complex object interpretation: the interpretation of common subobjects of
$T_1$ and $T_2$ are not necessarily identified, so the output is not necessarily a model of $\oneobjth$.
This is fixed by postcomposing with
the interpretation of Proposition~\ref{prop:mostowski-collapse-interp} to obtain $\<I_1,I_2\>$.
\end{itemize}
\end{proof}

\myeat{Proposition~\ref{prop:interp-sing-cup-map-pair} only 
only handles
the more basic transformations like $x,y \mapsto x \cup y$.
It does not by itself give us an implementation of  basic $\nrc$ connectives such as the union operation
$\vec x \mapsto E_1(\vec x) \cup E_2(\vec x)$. 
We can make use of Proposition \ref{prop:fointerp-comp}
to get the more general operation.
}

\begin{proof}[Proof of Proposition~\ref{prop:reducemonadic-interp}]

Similarly as with Proposition~\ref{prop:reducemonadic-nrc}, we define auxiliary
interpretations $\interp_{\uparrow}$, $\interp_{\downarrow}$
$\interp_{\widehat\pair}$, $\interp_{\hat\pi_1}$ and $\interp_{\hat\pi_2}$
mimicking the relevant constructs of Proposition~\ref{prop:reducemonadic-nrc}.
Then we will dispense with giving the recursive definitions of $\interp_{\convert_T}$
and $\interp_{\convert^{-1}_T}$, as they will be obvious from inspecting the 
clauses given in the proof of Proposition~\ref{prop:reducemonadic-nrc} and
replicating them using Proposition~\ref{prop:interp-sing-cup-map-pair} together with closure under composition
of interpretations. 

$\interp_{\uparrow}$, $\interp_{\downarrow}$
and $\interp_\pair$ are easy to define through Proposition~\ref{prop:interp-sing-cup-map-pair}, so
we focus on the projections $\interp_{\hat\pi^{n_1,n_2}_1}$ and $\interp_{\hat\pi^{n_1,n_2}_2}$,
defining transformations from $\ursort_m$ to $\ursort_{n_i}$ for $i \in \{1,2\}$ where $m \eqdef \max(n_1,n_2)$.
Note that in both cases, the output sort is part of the input sorts. Thus an output sort will be interpreted
by itself in the input, and the formulas will be trivial for every sort lying strictly below the output sort:
we take 
$$\phi_{\in_{\ursort_{k}}}(x,y) ~\eqdef~ x \in y \wedge \phi_\domainof^{\ursort_{k+1}}(y) \qquad
  \phi_{\equiv}^{\ursort_k}(x,y) ~\eqdef~ x = y \qquad
  \phi_\domainof^{\ursort_k}(x) ~\eqdef~ \top
$$
for every $k < n_i$ ($i$ according to which projection we are defining).
The only remaining important data that we need to provide are the 
formulas $\phi_\domainof^{\ursort_{n_i}}$,
which, of course, differ for both projections. We provide those below, calling $o_{in}$ the designated input
object. For both cases, we use an auxiliary predicate  $x \in^k y$ standing for $\exists y_1 \in y \ldots \exists y_{k-1}
\in y_{k-2}~~ x \in y_{k-1}$ for $k > 1$; for $k=0,1$, we take $x \in^1 y$ to be $x \in y$ and $x \in^0 y$ for $x = y$.
\begin{itemize}
\item For $\interp_{\hat\pi^{n_1,n_2}_1}$, we set
$$\phi_\domainof^{\ursort_{n_1}}(x) ~ \eqdef ~ \forall z \in o_{in}~\exists z' \in z ~~~ x \in^{m-n_1} z'$$
The basic idea is that the outermost $\forall\exists$ ensures that we compute the intersection of the two sets
contained in the encoding of the pair. 
\item For $\interp_{\hat\pi^{n_1,n_2}_2}$, first note that there are obvious $\deltazero$-predicates $\issing(x)$ and
$\istwo(x)$ classifying singletons and two element sets. This allows us to write the following $\deltazero$ formula
$$
\small
\phi_\domainof^{\ursort_{n_1}}(x) ~ \eqdef ~
\bigvee \left[ \begin{array}{l}
\issing(x) \wedge \forall z \in o_{in}~\exists z' \in z~~~ x \in^{m-n_2} z' \\
\istwo(x) \wedge \exists z \; z' \in o_{in}~ \exists y \in z'~ (y \notin z \wedge x \in^{m-n_2} z')
\end{array}\right.
$$
\end{itemize}
It is then easy to check that, regarded as transformations, those interpretation also implement the
projections for Kuratowski pairs.

\end{proof}


\section{Proofs for Section 5: equivalence of nested relational 
transformations and interpretations}

\myparagraph{From $\nrcwget$ expressions to interpretations}
In the body of the paper we claimed that $\nrcwget$ expressions have
the same expressiveness as interpretations.
One direction of this expressive equivalence is given in the
following lemma:

\begin{lemma} \label{lem:forward} There is an $\exptime$ computable function
taking an $\nrcwget$ expression $E$ to an equivalent FO interpretation $\interp_E$.
\end{lemma}

As we mentioned in the body of the paper,  
very similar results occur in the prior literature, going
as far back as \cite{simulation}. 
\myeat{
The underlying idea is
simply that we can ``shred''  a transformation to run  over a flat encoding
of a nested object, which
 has been investigated in several communities
for languages related to $\nrc$
\cite{cheneysigmod, xqueryinterp, cheneydbpl}.
There is also similarity to results
from the 1960's of  Gandy \cite{gandy}. Gandy
 defines a class of set functions that are similar
to $\nrc$, and shows that they are ``substitutable''. This is the core
of the argument for translating $\nrc$ to interpretations.
}

\begin{proof}
We can assume that the input and output schemas are monadic, using the
reductions to monadic schemas given previously.
Indeed, if we solve the problem
for expressions where input and output
schemas are monadic, we can reduce the problem of finding an interpretation for an
arbitrary $\nrcwget$ expression $E(x)$ as follows:
construct a $\deltazero$ interpretation $\interp$
for the  expression $\convert(E(\convert^{-1}(x)))$ --
where $\convert$ and $\convert^{-1}$
are taken as in Proposition~\ref{prop:reducemonadic-nrc} --
and then, using closure under composition of interpretations (see e.g.~\cite{xqueryinterp}),
one can then leverage Proposition~\ref{prop:reducemonadic-interp}
to produce the composition of $\interp_{\convert^{-1}}$, $\interp$ and $\interp_{\convert}$
which is equivalent to the original expression $E$.

The argument proceeds by induction on the structure of $E : \vec T \to S$ in
$\nrc$. Some atomic operators were treated in the prior section, like singleton
$\cup$, tupling, and projections.  Using
closure of interpretations under composition, we are thus able to translate
compositions of those operators.
We are only left with a few cases.
\begin{itemize}
\item
For the set difference, since interpretations are closed under
composition, it suffices to prove that we can code the transformation
$$(x, y) ~~\mapsto~~ x \setminus y$$
at every sort $\sett(\ursort_n)$. Each sort gets interpreted by itself.
We thus set
$$
\begin{array}{lcll}
\phi_\domainof^{\ursort_n}(z) &~~\eqdef~~&
z \in \pi_1(\inobj) \wedge z \notin \pi_1(\inobj) \\
\phi_\domainof^{\ursort_k}(z) &~~ \eqdef ~~&
\exists z' ~~(\phi_\domainof^{\ursort_n} \wedge z \in^{n-k} z') \\
\phi_\in^{\ursort_k}(z, z') &~~ \eqdef ~~& z \in z' \wedge \phi^{\ursort_k}_\domainof(z) \wedge \phi^{\ursort_{k+1}}_\domainof(z') \\
\end{array}
$$
\item To get $\nrcwget$ expressions, it suffices 
to create a
$\deltazero$ interpretation corresponding to  $\nrcget$
which follows
$$\phi_\domainof^\ursort(a) ~~ \eqdef~~ (\exists !~ z \in \inobj ~~ z = a) \vee (\neg (\exists !~ z \in \inobj) \wedge a = c_0)$$
\item For the binding operator
$$\bigcup \{ E_1 \mid x \in E_2 \}$$
we exploit the classical decomposition
$$\bigcup ~~\circ~~ \map(E_1) ~~\circ~~ E_2$$
As interpretations are closed under composition and the mapping operations
was handled in Proposition~\ref{prop:interp-sing-cup-map-pair},
it suffices to give an interpretation for the expression $\bigcup : \sett(\sett(T)) \to \sett(T)$
for every sort $T$. This is straightforward: each sort gets interpreted as itself,
except for $\sett(T)$ itself which gets interpreted as the singleton $\{c_0\}$.
The only non-trivial clause are the following
$$
\phi_\in^{T}(x,y) ~~\eqdef~~ \phi_\domainof^T ~~\eqdef~~ \exists y' \in \inobj ~~ x \in y'
$$
\end{itemize}
\end{proof}

\myparagraph{From interpretations to $\nrcwget$ expressions}
The other direction of the expressive equivalence is provided by the following
lemma:

\begin{lemma} \label{lem:back} There is a polynomial time function
taking a $\deltazero$ interpretation to an equivalent $\nrcwget$ expression.
\end{lemma}

This direction is not used directly in the conversion
from implicitly definable transformations to $\nrcwget$, 
but it is of interest
in showing that $\nrcwget$ and $\deltazero$ interpretations are equally expressive.

\begin{proof}(of Lemma \ref{lem:back})
Using the reductions to monadic schemas, it suffices to show this
for transformations that have monadic input schemas as input and output.

Fix a $\deltazero$ interpretation $\interp$ with input $\ursort_n$ and output $\ursort_m$.

Before we proceed, first note that for every $d \le m$,
there is an $\nrc$ expression
$$E_d : \ursort_n \to \sett(\ursort_d)$$
collecting all of the subobjects of its input of sort $\ursort_d$.
It is formally defined by the induction over $n - d$.
$$E_m(x) ~~ \eqdef~~ \{x\} \qquad \qquad E_d(x) = \bigcup E_{d-1}(x)$$
Write $E_{d_1, \ldots, d_k}(x)$ for $\< E_{d_1}, \ldots, E_{d_k} \>(x)$
for every tuple of integers $d_1 \ldots d_k$.

For $d \le m$, let $d_1, \ldots, d_k$ be the tuple such that the output
sort $\ursort_d$ is interpreted by the list of input sorts
$\ursort_{d_1}, \ldots, \ursort_{d_k}$.
By induction over $d$, we build  $\nrc$ expressions
$$E_d : \ursort_m, \ursort_{d_1}, \ldots, \ursort_{d_k} \to \ursort_d$$
such that, provided that $\phi_\domainof^{\ursort_d}(\vec a)$ and $\phi_\domainof^{\ursort_{d+1}}(\vec b)$ hold, we have
$$\phi_\in^{\ursort_d}(\vec a, \vec b) \qquad \qquad  \text{if and only if} \qquad \qquad  E_d(\vec a) \in E_{d+1}(\vec b)$$

For $E_0 : \ursort_m, \ursort \to \ursort$, we simply take the second projection.
Now assume that $E_d$ is defined and that we are
looking to define $E_{d+1}$.
We want to set
\begin{align*}
E_{d+1}(x_{in},\vec{y}) \eqdef \{E_d(x_{in}, \vec x) \mid \vec{x} \in E_{d_1, \ldots, d_k}(x_{in}, \vec y) \wedge \verify_{\phi_{\in}^i}(x_{in}, \vec x, y_{in}, \vec y) \}
\end{align*}
which is $\nrc$-definable as follows 
$$\bigcup \left\{
\case\left(\verify_{\phi_{\in}^i}(x_{in}, \vec x, y_{in}, \vec y), ~ \{E_d(x_{in})\}, ~ \{\}\right)
\mid \vec x \in E_{d_1,\ldots,d_k}(x_{in}) \right\}
$$
where $\verify$ is given as in the Verification Proposition proven earlier in the 
supplementary materials
and $\{E(\vec x, \vec y) \mid \vec x \in E'(\vec y)\}$ is a
notation for $\bigcup \{ \ldots \bigcup \{ E(\vec x, \vec y) \mid x_1 \in \pi_1(E'(\vec y)) \} \ldots \mid x_k \in \pi_k(E'(\vec y))\}$.
It is easy to check that the inductive invariant holds.

Now, consider the  transformation $E_m : \ursort_n, \ursort_{m_1}, \ldots, \ursort_{m_k} \to \ursort_m$.
The transformation
$$R ~~ \eqdef ~~\{ E_m(x_{in}, \vec y) \mid \vec y \in E_{m_1, \ldots, m_k}(x_{in}) \wedge \phi_\domainof^{\ursort_m}(\vec y)\}$$
is also $\nrc$-definable using $\verify$.
Since the inductive invariant holds at level $m$, $R$ returns the singleton containing the output of $\interp$. Therefore $\nrcwget(R) : \ursort_n \to \ursort_m$ is the desired $\nrcwget$ 
expression equivalent to the interpretation $\interp$.
\end{proof}

Note that  the argument can be easily modified to
 produce an $\nrcwget$ expression that is \emph{composition-free}:  in
union expressions $\bigcup\{E_1 \mid x \in E_2\}$, the range $E_2$ of
the variable $x$ is always another variable.
In composition-free expressions, we allow as a native  construct $\case(B,E_1,E_2)$ where $B$ is a Boolean combination
of atomic transformations with Boolean output, since we cannot use composition to derive the conditional from the other operations.


Thus every $\nrcwget$ expression can be converted to one that is composition-free, and similarly
for $\nrcwget$.  The analogous
statements have been observed before for related languages like XQuery \cite{xqueryinterp}.


\section{Proofs for Section 6: proof details concerning generating interpretations from classical proofs} 

\subsection{Requirement that not all input sorts be singletons}
Recall  from Section 6 that in our main theorem 
relating implicit and explicit interpretability within multi-sorted logic,
 we required   that the theory
$\Sigma$ entails the existence of a sort in $\smallsorts$ with more than one element.

We now explain that this requirement
is essential. Otherwise we might have $\smallsorts$ entailed
by $\Sigma$ to consist of a single element which is named by a constant, 
while $\bigsorts$ has another sort with two elements, each named by a constant.
Since every element of the models of $\Sigma$ is named by a constant, all
models are isomorphic, and hence we have implicit interpretability vacuously.
But we cannot explicitly interpret $\bigsorts$ in $\smallsorts$ simply
for cardinality reasons.

\subsection{Details of the reduction allowing us to  drop additional parameters}

Recall that in the body of the paper we claimed that
to be able to generate $\nrcwget$ expressions from projective implicit
definitions, it suffices to deal with implicit definitions: formulas
$\Sigma'(\inobj, \outobj)$ with no auxiliary variables $\vec a$:

\medskip

For any $\deltazero$ formula $\Sigma(\inobj, \outobj, \vec a)$ 
that implicitly defines $\outobj$ as a  function of
$\inobj$, there is another $\deltazero$ formula $\Sigma'(\inobj, \outobj)$ which 
implicitly $\outobj$ as a function of $\inobj$  such that
$\Sigma(\inobj, \outobj, \vec a) \Rightarrow \Sigma'(\inobj, \outobj)$.

\medskip

We now give the proof:

\begin{proof}
The assumption that $\Sigma$ implicitly defines $\outobj$ as a function
of $\inobj$  means that we have an entailment
$$\Sigma(\inobj, \outobj, \vec a) \models \Sigma(\inobj, \outobj', \vec{a'}) \Rightarrow \outobj = \outobj'$$
Applying $\deltazero$ interpolation 
we may obtain a formula $\theta(\inobj,\outobj)$ such that
$$\Sigma(\inobj, \outobj, \vec a) \models \theta(\inobj, \outobj) \qquad \text{and} \qquad \theta(\inobj, \outobj) \wedge \Sigma(\inobj, \outobj', \vec {a'}) \models \outobj = \outobj'$$
Now we can derive the following entailment
$$\Sigma(\inobj, \outobj, \vec a) \models [\theta(\inobj,\outobj') \wedge \theta(\inobj,\outobj'')] \Rightarrow \outobj' = \outobj''$$
This entailment is obtained from the second property of $\theta$, since
we can infer that $\outobj'=\outobj$ and $\outobj''=\outobj$.

Now we can  apply interpolation again to obtain a formula $D(\inobj)$ such that
$$\Sigma(\inobj, \outobj, \vec a) \models D(\inobj) \qquad \text{and} \qquad D(\inobj) \wedge \theta(\inobj, \outobj') \wedge \theta(\inobj, \outobj'') \models \outobj' = \outobj''$$
We now claim that $\Sigma'(\inobj, \outobj) \eqdef D(\inobj) \wedge \theta(\inobj, \outobj)$ is an implicit definition
extending $\Sigma$.
Functionality of $\Sigma'$ is a consequence of the second entailment witnessing that $D$ is an interpolant.
Finally, the implication $\exists \vec a~~ \Sigma(\inobj,\outobj, \vec a) \models \Sigma'(\inobj, \outobj)$ is given by the
combination of the first entailments witnessing that $\theta$ and $D$ are interpolants.
\end{proof}


\subsection*{Reduction to complete theories} \label{sec:assumecompleteproof}
Recall the result on multi-sorted first-order logic  in the body of the paper:

\medskip

For any $\Sigma, \smallsorts, \bigsorts$ such that $\Sigma$ entails that a sort
of $\smallsorts$ has at least two elements, $\bigsorts$ is explicitly
interpretable over $\smallsorts$ if and only if it is implicitly interpretable
over $\smallsorts$.

\medskip

In the body of the paper, we argued that it suffices to prove this
for the case when $\Sigma$ is a complete theory. We now prove this:

\begin{proof}
Fix a $\Sigma$ 
satisfying the hypothesis, but not the conclusion,
and let $\rho$ be a sentence in the vocabulary of $\Sigma$.  We claim that
one of $\rho, \neg \rho$ can be added to $\Sigma$ in such
a way that the conclusion of the theorem still fails.
This would suffice, since then we can inductively complete
$\Sigma$ to a complete theory in which every  finite subset is satisfiable,
and hence by compactness a satisfiable theory.

The hypothesis of the theorem, implicit interpretability of $\bigsorts$
over $\smallsorts$ relative to $\Sigma$, is preserved
under extending $\Sigma$, and thus both $\Sigma \cup \{\rho\}$ and
$\Sigma \cup \{ \neg \rho \}$  implicitly define $\bigsorts$ as well.
Suppose by way of contradiction that in both extensions $\bigsorts$ is  explicitly
interpretable over $\smallsorts$. That is, suppose $\bigsorts$ is
explicitly interpretable over $\smallsorts$ via
$\Theta_1$ relative to $\Sigma \cup  \{\rho\}$, and also that
$\bigsorts$ is explicitly interpretable over $\smallsorts$
via $\Theta_2$ relative to $\Sigma \cup \{\neg \rho\}$.
At this point we would like to combine $\Theta_1$ and $\Theta_2$
to get an explicit interpretation relative to $\Sigma$, contradicting
the assumption. The obvious way to do  this would be to apply
$\Theta_1$ or $\Theta_2$ conditioning
on $\rho$. However, $\rho$ may make use of sorts outside of $\smallsorts$.

Consider the sentence $\Sigma_1$ stating that $\Sigma$
holds and if $\rho$ holds then $\bigsorts$ is interpreted
via $\Theta_1$ applied to $\smallsorts$. Then $\Sigma_1$
is implicitly definable over $\smallsorts$, and thus
by the standard Beth Definability theorem \cite{beth,craig57beth},  there is
a sentence $\Sigma'_1$ over $\smallsorts$ that holds
of models $M$ that extend to a $\Sigma_1$ structure.
Similarly we get a sentence $\Sigma'_2$ over $\smallsorts$
that holds of a $\smallsorts$ structure $M$ 
whenever $M$ has an expansion that either  satisfies
$\rho$ or agrees with $\Theta_2$.
We can form an interpretation that acts as $\Theta_1$
when $\Sigma'_1$ holds and as $\Theta_2$ when
$\Sigma'_2$ holds, and this gives a contradiction of the assumption
that the theorem failed for $\Sigma$.
\end{proof}


\subsection*{Proof of the final equivalence}
Recall that in the body of the paper we stated the following result:

\medskip

The following are equivalent for a transformation $\trans$:
\begin{compactitem}
\item $\trans$ is projectively implicitly definable by a $\deltazero$ formula
\item $\trans$ is implicitly definable by a $\deltazero$ formula
\item $\trans$ is definable via a $\deltazero$ interpretation
\item $\trans$ is $\nrcwget$ definable
\end{compactitem}

\medskip

The directions from the first bullet through to the fourth are proven
in the paper.
What remains is to show the following ``easy implication''.

\medskip

For every $\nrcwget$ expression  $E$ we can obtain 
a $\deltazero$ formula that  implicitly defines $E$.

\medskip

This can be done by induction on the structure of $E$.
For example, consider the case of the singleton constructor $E=\{F \}$.
Inductively we have $\phi_F(\vec x, q_2)$  defining $F$, and
from there we can define $E$ by:
\[
(\exists q_2 \in q_1 ~ \top)  \wedge (\forall q_2 \in q_1 ~ \phi_F(\vec x, q_2))
\]

We discuss briefly the inductive case of the union operator. One approach,
is to break this operator down into a simpler union operator where the variable
can only iterate over another variable. The full union operator can be recovered
if we also allow a composition operation. The simpler operator is easy to handle inductively.
Composition can be handled without a blow-up if we allow \emph{projective implicit definitions},
because projective implicit definitions are closed under composition.
From our prior results, we know that projective implicit definitions are no more expressive
than implicit ones.

An alternative is to rely on the $\nrcwget$ normalization result mentioned at the end of 
Lemma  \ref{lem:back}: we can pre-process $\nrcwget$ expressions to be composition-free: in unions
we do not iterate over  complex expressions. For these normalized expressions, the creation
of implicit definitions can be done in $\ptime$.




\end{document}